\documentclass[11pt]{article}
\usepackage[T1]{fontenc}
\usepackage[letterpaper, margin=1in]{geometry}
\usepackage{graphicx} \usepackage{amsmath}
\usepackage{amssymb}
\usepackage{amsthm}
\usepackage{bbm}
\usepackage[colorlinks=true, allcolors=blue, pagebackref=true]{hyperref}
\usepackage{xcolor}
\usepackage{stmaryrd}
\usepackage{subcaption}
\usepackage{paralist}
\usepackage{wrapfig}

\newtheorem{theorem}{Theorem}[section]
\newtheorem{lemma}{Lemma}[section]

\newtheorem{corollary}{Corollary}[section]
\newtheorem{proposition}{Proposition}[section]

\theoremstyle{definition}
\newtheorem{definition}{Definition}[section]

\newtheorem{remark}{Remark}[section]
\newtheorem{conjecture}{Conjecture}[section]

\newcommand{\iv}[1]{\llbracket {#1} \rrbracket}

\renewcommand{\emptyset}{\varnothing}

\newcommand{\calD}{\mathcal D}
\newcommand{\calR}{\mathcal R}
\newcommand{\calX}{\mathcal X}

\newcommand{\bbC}{\mathbb C}
\newcommand{\bbQ}{\mathbb Q}
\newcommand{\bbR}{\mathbb R}
\newcommand{\bbE}{\mathbb E}
\newcommand{\bbF}{\mathbb F}
\newcommand{\bbK}{\mathbb K}
\newcommand{\bbZ}{\mathbb Z}
\newcommand{\bbN}{\mathbb N}

\newcommand{\OR}{\mathsf{OR}}
\newcommand{\PAR}{\mathsf{PAR}}

\newcommand{\one}{\mathbbm {1}}

\newcommand{\sfT}{\mathsf{T}}

\newcommand{\Sym}{\mathfrak{S}}

\newcommand{\DFT}{\mathrm{DFT}}

\renewcommand{\setminus}{\smallsetminus}

\newcommand{\SVMV}{\mathsf{SVMV}}
\newcommand{\OV}{\mathsf{OV}}

\DeclareMathOperator{\GL}{GL}

\DeclareMathOperator{\nnz}{nnz}

\DeclareMathOperator{\Size}{S}
\DeclareMathOperator{\Degree}{D}
\DeclareMathOperator{\Eq}{Q}

\let\R\relax
\DeclareMathOperator{\R}{R}
\DeclareMathOperator{\AR}{\tilde{R}}
\DeclareMathOperator{\freeM}{\mathsf{M}}

\let\H\relax
\DeclareMathOperator{\H}{H}
\DeclareMathOperator{\leaves}{leaves}
\DeclareMathOperator{\type}{type}

\DeclareMathOperator{\supp}{supp}

\DeclareMathOperator*{\Ex}{\mathbb E}

\title{Kronecker Powers, Orthogonal Vectors, \\and the Asymptotic Spectrum}
\author{Josh Alman\thanks{Columbia University. \texttt{josh@cs.columbia.edu}. Supported in part by NSF Grant CCF-2238221 and a Packard Foundation Fellowship.} \and
    Baitian Li\thanks{Columbia University. \texttt{bl3052@columbia.edu}. Part of the work was performed while the author was an undergraduate student at Tsinghua University.}}
\date{}

\begin{document}

\maketitle

\begin{abstract}

We study circuits for computing linear transforms defined by Kronecker power matrices. 
The best-known (unbounded-depth) circuits, including the widely-applied fast Walsh--Hadamard transform and Yates' algorithm, can be derived from the best-known depth-2 circuits using known constructions, so we focus particularly on the depth-2 case.
Recent work~\cite{jukna2013complexity,alman2021kronecker,AGP23KroneckerCircuits,Sergeev22KroneckerCoverings} has improved on decades-old constructions in this area using a new \emph{rebalancing} approach, but it was unclear how to apply this approach optimally, and the previous versions had complicated technical requirements. 

We find that Strassen's theory of asymptotic spectra can be applied to capture the design of these circuits. This theory was designed to generalize the known techniques behind matrix multiplication algorithms as well as a variety of other algorithms with recursive structure, and it brings a number of new tools to use for designing depth-2 circuits.
In particular, in hindsight, we find that the techniques of recent work on rebalancing were proving special cases of the \emph{duality theorem} which is central to Strassen's theory. We carefully outline a collection of \emph{obstructions} to designing small depth-2 circuits using a rebalancing approach, and apply Strassen's theory to show that our obstructions are \emph{complete}.

    Using this connection, combined with other algorithmic techniques (including matrix rigidity upper bounds, constant-weight binary codes, and a ``hole-fixing lemma'' from recent matrix multiplication algorithms), we give new improved circuit constructions as well as other applications, including:
    \begin{itemize}
        \item The $N \times N$ disjointness matrix has a depth-2 linear circuit of size $O(N^{1.2495})$ over any field. This is the first construction which surpasses exponent 1.25, and thus yields smaller circuits for many families of matrices using reductions to disjointness, including all Kronecker products of $2 \times 2$ matrices, without using matrix rigidity upper bounds.
        \item Barriers to further improvements, including that the Strong Exponential Time Hypothesis implies an $N^{1 + \Omega(1)}$ size lower bound for depth-2 linear circuits computing the Walsh--Hadamard transform (and the disjointness matrix with a technical caveat), and that proving such a $N^{1 + \Omega(1)}$ depth-2 size lower bound would imply breakthrough threshold circuit lower bounds.
        \item The Orthogonal Vectors (OV) problem in moderate dimension $d$ can be solved in deterministic time $\tilde{O}(n \cdot 1.155^d)$, derandomizing an algorithm of Nederlof and W\k{e}grzycki~\cite{NW21OV}, and the counting problem can be solved in time $\tilde{O}(n \cdot 1.26^d)$, improving an algorithm of Williams~\cite{Williams2024equalityrank} which runs in time $\tilde{O}(n \cdot 1.35^d)$. We design these new algorithms by noticing that prior algorithms for OV can be viewed as corresponding to depth-2 circuits for the disjointness matrix, then using our framework for further improvements.
    \end{itemize}
                            
        \end{abstract}

\thispagestyle{empty}
\newpage\pagenumbering{arabic}

\section{Introduction}

For an $n_a \times m_a$ matrix $A$, and a $n_b \times m_b$ matrix $B$, their Kronecker product, denoted $A \otimes B$, is an $(n_a n_b) \times (n_a m_b)$ matrix, whose entries are given by, for $i_a, \in [n_a]$, $j_a \in [m_a]$, $i_b \in [n_b]$, $j_b \in [m_b]$, 
$A \otimes B [(i_a, i_b), (j_a, j_b)] = A[i_a,j_a] \cdot B[i_b,j_b].$ For a positive integer $k$, the Kronecker power $A^{\otimes k}$ is the $n_a^k \times m_a^k$ matrix formed by taking the Kronecker product of $k$ copies of $A$. 

We study circuits for computing the linear transform represented by such a matrix. In other words, for the matrix family $A^{\otimes k}$ for fixed $n \times n$ matrix $A$ and growing $k$, the circuit is given as input a vector $v$ of length $N = n^k$, and should output the length-$N$ vector $A^{\otimes k} v$. We focus in particular on \emph{linear circuits}, whose gates compute fixed linear combinations of their inputs. The straightforward approach for such a matrix-vector multiplication would use $\Theta(N^2)$ arithmetic operations, but it's known, using the approach of Yates' algorithm from nearly a century ago~\cite{yates1937design}, that such matrices have depth-2 linear circuits of size $O(N^{1.5})$ and depth-$O(\log N)$ linear circuits of size $O(N \log N)$.

Two of the most prominent families of Kronecker powers are the Walsh--Hadamard transform $H_k := H_1^{\otimes k}$, and the disjointness matrix $R_k := R_1^{\otimes k}$, for the fixed matrices $$H_1 = \begin{bmatrix} 1 & 1 \\ 1 & -1 \end{bmatrix}, ~~~ \text{ and } ~~~ R_1 = \begin{bmatrix} 1 & 1 \\ 1 & 0 \end{bmatrix}.$$
Yates' algorithm applied to $H_k$ is known as the fast Walsh--Hadamard transform, and has numerous applications in algorithm design and complexity theory, particularly in compression, encryption, and signal processing. It is also used as a subroutine in the fastest Fourier transform algorithm~\cite{alman2023faster}.  $R_k$ is also naturally widespread, particularly in communication complexity, where it is the communication matrix of disjointness. Its linear transform, sometimes known as the ``zeta transform for the lattice $(2^{[n]}, \subseteq)$'', is used throughout lattice algorithms, as well as in algorithms for solving systems of polynomial equations~\cite{bjorklund2019solving}. 

\subsection{Depth-2 Linear Circuits}

We focus in this paper specifically on depth-2 circuits for a few reasons. First, there is a known way to convert depth-2 circuits for Kronecker powers into smaller, higher-depth circuits:

\begin{lemma}[See~{\cite[Lemma~1.1]{AGP23KroneckerCircuits}} and {\cite[Theorem~3.4]{alman2021kronecker}}]\label{intro:depthincrease}
    If $A$ is a fixed $n \times n$ matrix, and there is a constant $c>0$ such that the matrix $A^{\otimes k}$ has a depth-2 circuit of size $O(n^{(1+c)k})$ for all $k$, then for all integers $d>0$, the matrix $A^{\otimes k}$ has a depth-$2d$ circuit of size $O(d \cdot 2^{(1+c/d)k})$.
\end{lemma}
\noindent All the smallest-known (unbounded-depth) circuits for Kronecker powers can be achieved in this way (by appropriately picking $d = O(k)$ in Lemma~\ref{intro:depthincrease}). Thus, improved depth-2 circuits also improve the best unbounded-depth circuits, and in particular, the leading constants of algorithms like Yates' algorithm and the fast Walsh--Hadamard transform, which are used prominently in both theory and practice, are directly improved by combining smaller depth-2 circuits with Lemma~\ref{intro:depthincrease}.

Low-depth circuits are also of interest in their own right, for instance because they are easily evaluated in parallel. Finally, beyond the goal of computing linear transforms, we will later take advantage of the structure of depth-2 circuits when designing algorithms for other problems.

A depth-2 linear circuit for an $N \times N$ matrix $M$ can be viewed equivalently as a factorization $M = U \times V^\sfT$ for $N \times G$ matrices $U, V$. Its size is $\nnz(U) + \nnz(V)$, where $\nnz$ denotes the number of nonzero entries in a matrix; given an input vector $v$, we may multiply $Mv = U(V^\sfT v)$ using $O(\nnz(U) + \nnz(V))$ operations in the straightforward way.
Here, $G$ is the number of linear gates in the middle layer of the circuit. We will often equivalently write $M = \sum_{i=1}^G U_i \times V_i^\sfT$, where $U_i, V_i$ are the $i$th column of $U$ and $V$ respectively, for emphasis: $V_i$ is the linear combination computed by the $i$th gate of the circuit, and $U_i$ shows how the circuit outputs use the result of that gate.

\subsection{Constructions and Recent Improvements}

Yates' algorithm gives a depth-2 circuit for any Kronecker power of size $O(N^{1.5})$. For the sparser disjointness matrix $R_n$, it gives a smaller depth-2 circuit of size $O(N^{\log_2(6)/2}) < O(N^{1.293})$. 

Until recently, it was believed that these Yates' algorithm circuits were essentially optimal, and barriers were proved against improving them. 
Pudl{\'a}k~\cite{pudlak2000note} proved that the Yates' algorithm circuit for $H_k$ is optimal among ``bounded-coefficient'' circuits (recent improvements used unbounded coefficients). 
Valiant~\cite{valiant1977graph} also introduced the notion of \emph{matrix rigidity} toward proving the optimality of these circuits. A matrix is called \emph{rigid} if it cannot be written as the sum of a low-rank matrix and a sparse matrix. Valiant proved that rigid matrices do not have $O(N)$-size linear circuits, even up to $O(\log N)$-depth, but it was recently shown that $H_k$, $R_k$, and many other important families are actually not rigid enough to prove lower bounds in this way~\cite{AW17,DL19,alman2021kronecker,kivva2021improved,bhargava2023fast}.

However, a new line of work gives the first improvements: a circuit for $H_k$ of size $O(N^{1.443})$~\cite{alman2021kronecker,AGP23KroneckerCircuits}, improving size $O(N^{1.5})$ from Yates' algorithm, and a circuit for $R_k$ of size $O(N^{1.251})$~\cite{jukna2013complexity,AGP23KroneckerCircuits,Sergeev22KroneckerCoverings}, improving size $O(N^{1.293})$ from Yates' algorithm. Both constructions work over every field $\mathbb{F}$, and come from an improved version of the same high-level approach as Yates' algorithm:
\begin{compactenum}
    \item Start with a fixed depth-2 circuit for a constant Kronecker power, then
    \item Give an approach for converting a circuit for a fixed power into a circuit for all powers.
\end{compactenum}
The new constructions for fixed matrices (component 1) are, with one exception, based on new matrix rigidity \emph{upper bounds}. Alman~\cite{alman2021kronecker} used a new rigidity upper bound for the Kronecker cube $A^{\otimes 3}$ of any fixed $2 \times 2$ matrix $A$, then Alman, Guan, and Padaki~\cite{AGP23KroneckerCircuits} later extended this to constant Kronecker powers of any square $A$ (not just $2 \times 2$).
The one exception is $R_k$. Jukna and Sergeev~\cite{jukna2013complexity} used the sparsity of $R_k$ to find a much simpler decomposition based on partitioning the $1$s of $R_1$ into two small rectangles. Later work by~Alman, Guan, and Padaki~\cite{AGP23KroneckerCircuits} and Sergeev~\cite{Sergeev22KroneckerCoverings} used similar combinatorial decompositions of larger disjointness matrices.

The approach (component 2) of Yates' algorithm (and~\cite{alman2021kronecker}) uses a simple application of the mixed product property of the Kronecker product:
\begin{lemma}[Mixed Product Property]\label{intro:mpp}
    Suppose $A$ is a symmetric matrix, and there are matrices $U,V$ such that $A^{\otimes k_0} = U \times V^\sfT$ for a fixed $k_0$. Then for all $k$, $A^{\otimes k}$ has a depth-2 circuit of size $$\left( \sqrt{\nnz(U) \cdot \nnz(V) }\right)^{k/k_0 + o(k)}.$$
\end{lemma}
Indeed, for any symmetric $A$, Yates' algorithm can be achieved using Lemma~\ref{intro:mpp} with the simple choices $k_0=1$, $U = A$, and $V = I$. Alman~\cite{alman2021kronecker} gave improvements to Yates' algorithm by giving better fixed constructions to use with Lemma~\ref{intro:mpp}.

Subsequent work replaced Lemma~\ref{intro:mpp} with a new idea called \emph{rebalancing}~\cite{AGP23KroneckerCircuits}, which in turn generalizes a recursive approach by Jukna and Sergeev~\cite[Lemma 4.2]{jukna2013complexity}. We will discuss rebalancing in more detail in Section~\ref{sec:technicaloverview} below, but we note that it requires $U$ and $V$ to be ``sufficiently imbalanced'', meaning roughly that different gates in the circuit have very different degrees. It then converts $U$ and $V$ into smaller, more balanced circuits for higher Kronecker powers:
\begin{lemma}[Rebalancing, informal~\cite{jukna2013complexity,AGP23KroneckerCircuits,Sergeev22KroneckerCoverings}]\label{intro:rebalancing}
    Suppose $A$ is a symmetric matrix, and there are ``sufficiently imbalanced'' matrices $U,V$ such that $A^{\otimes k_0} = \sum_{i=1}^s U_i \times V_i^\sfT$ for a fixed $k_0$. Then for all $k$, $A^{\otimes k}$ has a depth-2 circuit of size $$\left( \sum_{i=1}^m \sqrt{\nnz(U_i) \cdot \nnz(V_i)} \right)^{k/k_0 + o(k)}.$$
\end{lemma}
The AM-GM inequality shows that Lemma~\ref{intro:rebalancing} yields at least as small a circuit as Lemma~\ref{intro:mpp}, and is almost always strictly better, but it requires $U$ and $V$ to be ``sufficiently imbalanced'' (which has a complicated definition we discuss later).

Circuits from matrix non-rigidity are naturally sufficiently imbalanced, and the decompositions of $R_k$ by~\cite{jukna2013complexity,AGP23KroneckerCircuits} are as well. 
However, Sergeev~\cite{Sergeev22KroneckerCoverings} gave a new decomposition of $R_{18}$ which would yield a smaller circuit of size $O(N^{1.2503})$ for $R_k$ when used in Lemma~\ref{intro:rebalancing}, except that his decomposition is \emph{not} sufficiently imbalanced. It seems inherent to the rebalancing approach of prior work that a ``sufficiently imbalanced'' requirement appear (see Section~\ref{sec:technicaloverview} as well as \cite[Section~2]{AGP23KroneckerCircuits} and \cite[Section~2]{Sergeev22KroneckerCoverings}), but removing it yield an improved circuit from Sergeev's decomposition.

\subsection{Unified Approach via Asymptotic Spectrum}

In this paper, we give a new approach to designing algorithms and circuits for Kronecker powers, which both unifies and improves the prior work. A key idea behind our approach is to show that an appropriate \emph{preordered semiring} can be defined to capture such depth-2 circuits, so that we may make use of tools from Strassen's theory of the \emph{asymptotic spectrum}. 

Strassen developed this theory in a sequence of works~\cite{strassen1986asymptotic,strassen1987relative,strassen88spectrum,strassen1991degeneration} toward understanding the landscape of possible improvements to his breakthrough matrix multiplication algorithm~\cite{strassen1969gaussian}.
This theory is designed in a very general way, and applies to any complexity measure on problems which can be captured by a preordered semiring. Many problems with recursive structures can be described in this way, including in graph theory, communication complexity, and quantum information theory. Ideas from the theory underlie the best known algorithms and barriers for matrix multiplication today, as well as results in these other areas. 

A centerpiece of the theory is a \emph{duality theorem}, a generalization of LP duality which can be used to compute the complexity of problems. Beyond this, many other results about matrix multiplication and other instances of the asymptotic spectrum also generalize to hold for any instantiation of the theory. 
We refer the unfamiliar reader to the recent survey by Wigderson and Zuiddam~\cite{WZ22spectra} which discusses the theory, its applications, and recent developments in great detail.

We give a new application of Strassen's theory to capture the design of linear circuits for Kronecker power matrix linear transforms. 
The connection requires a more abstract semiring than in previous applications of the theory to tensors~\cite{strassen88spectrum}
and graphs~\cite{Zui19Graph}, and is inspired by insights from rebalancing. 
Using it, we give improved circuit constructions and find new applications for them. 
Among other improvements, we remove the technical restrictions of prior instantiations of rebalancing (including the need to start with a ``sufficiently imbalanced'' circuit in Lemma~\ref{intro:rebalancing}), and can thus make use of \emph{any} fixed constructions. 
In hindsight, we find that the implementations of rebalancing from recent work were proving special cases of Strassen's duality theorem. 
We will describe our new approach in more detail shortly in Section~\ref{sec:technicaloverview} below.

Our approach has applications in a number of different settings which we describe next.

\subsection{Application: Circuits for Disjointness}

Using our approach, we design a new, smaller circuit for the disjointness matrix. 

\begin{theorem} [Theorem \ref{thm:disjointness_par}, Corollary \ref{cor:depthd}] \label{thm:intro1249}
    Over any field, the $N\times N$ disjointness matrix $R_k$ has a depth-2 linear circuit of size $O(N^{1.249424})$, and more generally, for any even $d$, a depth-$d$ linear circuit of size $O(N^{1 + 0.498848/d})$.
\end{theorem}

Notably, by removing the ``sufficiently imbalanced'' condition of Lemma~\ref{intro:rebalancing}, our approach is able to use the construction of~\cite{Sergeev22KroneckerCoverings}, but this would only give a depth-2 circuit of size $O(N^{1.2503})$ (actually slightly worse; see the discussion in Section~\ref{sec:technicaloverview} below). 
We ultimately use the clarified and simplified objective of our approach to find a new, even better decomposition of $R_k$ in order to prove Theorem~\ref{thm:intro1249}.

\subsection{Application: Circuits for Non-Kronecker-Power Matrices}

\begin{table}[b] 
\begin{tabular}{|l|l|}
\hline
Exponent & Reference \\
\hline
1.293 & \cite{yates1937design}  \\
\hline
1.272 & \cite{jukna2013complexity} \\
\hline
1.258 & \cite{AGP23KroneckerCircuits} \\
\hline
1.251 & \cite{Sergeev22KroneckerCoverings} \\
\hline
1.2503 & \cite{Sergeev22KroneckerCoverings} if ``sufficiently imbalanced'' requirement could be removed from Lemma~\ref{intro:rebalancing} \\
\hline
1.2495 & New, Theorem~\ref{thm:intro1249} \\
\hline
\end{tabular}
\centering
\caption{Bounds on the exponent $\gamma$ such that the $N \times N$ matrix $R_n$ has a depth-2 circuit of size $O(N^{\gamma})$ over any field. }
\label{table:pastsizes}
\end{table}
Perhaps surprisingly, our Theorem~\ref{thm:intro1249} achieves a better exponent than $1.25$. The authors previously believed this wouldn't be possible for two reasons. First, numerically, the exponents of prior constructions appeared to be converging to $1.25$ (see Table~\ref{table:pastsizes}). 
Second, because of the generality of $R_n$, many other problems can be reduced to designing circuits for $R_n$, and $1.25$ is the threshold at which these reductions give new, general results. Indeed, using our construction, we prove:

\begin{corollary} [Corollary \ref{cor:or-circulant}] \label{intro:ormatrix}
    For any positive integer $n$, any field $\mathbb{F}$, and any function $f \colon \{0,1\}^n \to \mathbb{F}$, let $N = 2^n$, and consider the matrix $M_f \in \mathbb{F}^{N \times N}$ which is given by, for $x,y \in \{0,1\}^n$, $$M_f[x,y] = f(x_1 \vee y_1, x_2 \vee y_2, \ldots, x_n \vee y_n).$$ Then, $M_f$ has a depth-4 circuit of size $O(N^{1.249424})$, and more generally, for any even $d$, a depth-$2d$ linear circuit of size $O(N^{1 + 0.498848/d})$.
\end{corollary}

Many prominent families of matrices can be written in the form $M_f$ for different choices of $f$. For instance, Kronecker products of sequences of $2 \times 2$ matrices can be expressed in this form. Yates' algorithm gives a depth-4 circuit for these Kronecker products of size $O(N^{1.25})$, and our Corollary~\ref{intro:ormatrix} gives a new proof of the recent fact that Kronecker powers have smaller depth-4 circuits than what is given by Yates' algorithm. The previous proof of this~\cite{alman2021kronecker} critically used matrix rigidity upper bounds; we instead combine a combinatorial decomposition of $R_n$ with our asymptotic spectrum approach. We also note that the previous constructions had unbounded coefficients in all four layers of the circuit, whereas our new construction only has unbounded entries in one of the layers. (Pudl{\'a}k's barrier~\cite{pudlak2000note} shows that unbounded coefficients are needed somewhere in the circuit to beat Yates' algorithm.)

That said, many matrices which are not Kronecker products can also be expressed as $M_f$. For one example, picking $f$ to be the majority function gives Hamming distance matrices previously studied in the context of communication complexity and matrix rigidity~\cite{alman2025low}, and whose linear transform can be used to solve a basic nearest neighbor search problem. Our Corollary~\ref{intro:ormatrix} gives a generic improvement to $O(N^{1.2495})$.

\subsection{Application: Algorithms for Orthogonal Vectors}

As an application of our framework and circuits, we give faster algorithms for certain regimes of the \emph{Orthogonal Vectors} ($\OV$) problem. In this problem, one is given as input vectors $u_1, \ldots, u_n, v_1, \ldots, v_n \in \{0,1\}^{d}$, and must determine whether there are $i,j \in [n]$ with $\langle u_i, v_j \rangle = 0$, where the inner product is taken over $\mathbb{Z}$. In the \emph{Counting $\OV$} ($\#\OV$) problem, the goal is to count the number of such $i,j$.

The straightforward algorithm runs in time $O(n^2 d)$. In the case when $d = c \log n$, this can be improved slightly, for both $\OV$ and $\#\OV$, to $n^{2 - 1/O(\log c)}$~\cite{abboud2014more,chan2020deterministic}. In particular, as $c$ grows, this running time approaches quadratic, which is popularly conjectured to be necessary:
\begin{conjecture}[Orthogonal Vectors Conjecture (OVC)]\label{conj:ovc}
    For every $\varepsilon>0$, there is a $c >0$ such that $\OV$ for dimension $d = c \log n$ requires time $\Omega(n^{2 - \varepsilon})$.
\end{conjecture}
OVC is implied by the Strong Exponential Time Hypothesis (SETH)~\cite{impagliazzo2001complexity, williams2005new}, and most known consequences of SETH in fine-grained complexity are proved using OVC; see the survey by Vassilevska Williams for many examples in diverse areas of algorithm design~\cite{williams2018some}.

$\OV$ has an equivalent formulation using $R_d$: Given input vectors $ x, y \in \{0,1\}^{2^d} $ with $ n $ nonzero entries each, the counting version $\#\OV$ asks us to compute $ y^\sfT R_d x $, while the decision version $\OV$ asks whether this quantity is nonzero. Applying Yates' algorithm yields another folklore bound of $ O(n + d \cdot 2^d) $, which is faster in small dimensions where $ d < 2 \log n $. A faster, alternative analysis has also been given when the inputs are low Hamming weight vectors~\cite{BHKK09halve}.

Recently, two different papers have given algorithms for $\OV$ focusing on an \emph{intermediate} dimension regime. They improve the previous algorithms when $d = c \log n$ for $1.3 < c < 4.8$.

Nederlof and W\k{e}grzycki~\cite[Section~6]{NW21OV} provided
a randomized algorithm for a special case of $\OV$, but their approach generalizes to give a randomized algorithm for $\OV$ in dimension $d$ running in time $\tilde O(n \cdot 1.155^d)$. Their algorithm is based on a random hashing approach, and they use it as a critical subroutine of a fast, space-efficient algorithm for subset sum. See~\cite[Section~1.1]{NW21OV} for a discussion of how many prior algorithms for subset sum can also be interpreted as reducing to moderate-dimension $\OV$.

Williams~\cite{Williams2024equalityrank} gave a new deterministic algorithm running in time $\tilde{O}(n \cdot 1.35^d)$ for $\OV$, and $\tilde{O}(n \cdot 1.38^d)$ for $\#\OV$. His algorithm is based on a new measure of complexity of a matrix called its \emph{equality rank}; he uses SAT solvers to find equality rank upper bounds for fixed $R_d$ matrices, then combines these with a Yates' algorithm-like approach.  
Based on these algorithms and other connections with threshold circuits, Williams expresses optimism that $R_d$ could have small enough equality rank that this approach would yield an algorithm with running time $n \cdot 2^{o(d)}$ and thus refute the OVC; see~\cite[page 4]{Williams2024equalityrank}.

We observe that depth-2 circuits for $R_d$ can be used to design $\OV$ algorithms in this intermediate dimension regime, and furthermore that both prior algorithms~\cite{NW21OV,Williams2024equalityrank} can (with work) be interpreted as giving depth-2 circuits for $R_d$.
We approach $\OV$ as a special case of the \emph{sparse vector-matrix-vector multiplication} problem  for $R_d$ (see Sections \ref{sec:technicaloverview} and \ref{sec:algorithms}). This problem is also known as the ``pair problem'' in previous literature~\cite[Section 5]{ned20survey}, where alternative matrices are studied in the context of Hamiltonian path problems~\cite{BCKN15fact,CKN18fact,ned20bitsp}.

Using our new, improved depth-2 circuit constructions, we improve these algorithms, perhaps taking a step toward Williams' goal.

\begin{theorem}[Corollary \ref{cor:ov-par}] \label{intro:ov-par}
    $\#\OV$ can be solved in deterministic $\tilde O(n \cdot 1.26^d)$ time.
\end{theorem}

By comparison, for $\#\OV$,~\cite{Williams2024equalityrank} achieved deterministic $O(n \cdot 1.38^d)$ time.

\begin{theorem}[Corollary \ref{cor:ov}] \label{intro:ov}
    $\OV$ can be solved in deterministic $\tilde O(n \cdot 1.155^d)$ time.
\end{theorem}

By comparison, for $\OV$,~\cite{Williams2024equalityrank} achieved  deterministic $O(n \cdot 1.35^d)$ time, and the approach of~\cite{NW21OV} achieves randomized $\tilde O(n \cdot 1.155^d)$ time; our algorithm can be viewed as
a derandomization of their approach. 
Independently and concurrently with our work, Dürr, Kipouridis, and Węgrzycki~\cite{DKW25kOV} also use the approach of~\cite{NW21OV} and achieve the same running time for the decision version of $\OV$ as we do in Theorem~\ref{intro:ov}; their work further generalizes to the $k$-$\OV$ problem under parameter regimes analogous to ours.

We note that improving the $1.155$ in the base of the algorithms of~\cite{NW21OV} and Theorem~\ref{intro:ov} appears quite challenging. First, both algorithms ultimately correspond to a depth-2 circuit which covers (not an exact cover!) the $1$s of $R_n$ by rectangles, and one can show that the resulting circuit is already optimal among such covers (see Theorem~\ref{thm:disjointness_delta_or} for a simple proof); one would  need to use more general circuits. Second, an improvement (at least for inputs of a specific Hamming weight) would speed up the longstanding fastest running time for subset sum~\cite{NW21OV}.

\subsection{Application: Further Algorithms}

In a similar way, depth-2 circuits for other matrices are able to give faster algorithms for other problems. As an example, we also design new algorithms for counting orthogonal vectors modulo $m$. In this problem, denoted $\#\OV_{\bbZ_m}$, one is given as input $u_1, \ldots, u_n, v_1, \ldots, v_n \in \bbZ_m^d$, and must count the number of $i,j \in [n]$ for which $\langle u_i, v_j \rangle \equiv 0 \pmod{m}$. $\#\OV_{\bbZ_m}$ has applications to finite geometry, coding theory, and randomness extractors; see~\cite[Section~1]{williams2014finding}. Again, the straightforward running time is $O(n^2 \cdot d)$ and Yates' algorithm gives $\tilde{O}(n + m^d)$. Our approach applied to the depth-2 circuit from Yates' algorithm gives $\tilde O(n \cdot m^{d/2})$, but we are able to further improve this:

\begin{theorem}[Corollary \ref{cor:ov-dft}] \label{intro:ov-dft}
    The problem $\#\OV_{\bbZ_2}$ can be solved in $\tilde O(n \cdot 1.36^d)$ time.
    Moreover, the problem $\#\OV_{\bbZ_m}$ can be solved in $\tilde O(n \cdot m^{(1/2-a_m) d})$ time,
    where $a_m = \Omega(1/m^2 \log m)$.
\end{theorem}
Our algorithm for $\#\OV_{\bbZ_m}$ is derived by applying our approach to constructions of depth-2 linear circuits for the \emph{discrete Fourier transform} (DFT) from prior work~\cite{AGP23KroneckerCircuits}. Notably, these circuits in turn make use of matrix rigidity upper bounds for discrete Fourier transforms~\cite{DL19,AGP23KroneckerCircuits}.

We note that there is a known algorithm running in time $O(n \cdot d^{m-1})$ for the decision problem $\OV_{\bbZ_m}$ 
when $m$ is a prime power~\cite{williams2014finding}. That algorithmic approach, which uses matrix product verification to detect orthogonal pairs, does not appear to work for the counting problem $\#\OV_{\bbZ_m}$, and moreover, a running time of $O(n \cdot d^{m-1})$ for $\#\OV_{\bbZ_m}$ would actually refute the OVC (Conjecture~\ref{conj:ovc}) by a known reduction~\cite[Theorem~5]{Williams18Reductions}.

\subsection{Application: Circuit Lower Bounds}

In light of Theorem~\ref{thm:intro1249} and other recent improved depth-2 circuit constructions for $H_n$ and $R_n$~\cite{jukna2013complexity,alman2021kronecker,AGP23KroneckerCircuits,Sergeev22KroneckerCoverings}, it is natural to wonder what the true exponent is, and whether a depth-2 circuit of size $N^{1 + o(1)}$ for $H_n$ or $R_n$ may be possible. Using our new connection with the approach of Williams~\cite{Williams2024equalityrank}, we prove that this would refute popular conjectures in fine-grained complexity:

\begin{theorem} [Theorem \ref{thm:ovc-and-hadamard}, \ref{thm:ovc-and-disjointness}]
    The Orthogonal Vectors Conjecture (OVC) or the Strong Exponential Time Hypothesis (SETH) each imply that the Walsh--Hadamard matrix $H_d$ of side-length $N$ does not have a depth-2 linear circuit of size $N^{1 + o(1)}$. A similar result holds for the disjointness matrix $R_d$ in place of $H_d$ (with a restriction that the $R_d$ circuit is somewhat balanced).
        \end{theorem}
\noindent By comparison, the best known unconditional lower bounds on the size of a depth-2 circuit for an explicit function is  $\Omega(N \log^2 N / \log \log N)$ for matrices based on superconcentrator graphs~\cite{radhakrishnan2000bounds}, and the best for $H_d$ is $\Omega(N \log N)$~\cite{alon1990linear}.

We also give evidence that proving such a circuit lower bound would be difficult, by showing that it would imply a breakthrough lower bound for depth-2 threshold circuits for computing \emph{single-output Boolean} functions $\{0,1\}^n \to \{0,1\}$.
A $\mathsf{SUM} \circ \mathsf{ETHR}$ circuit is a depth-2 circuit where the bottom gates are \emph{exact threshold} gates,
which test whether a prescribed linear combination of their inputs over $\mathbb{R}$ is equal to a
fixed target, and the top gate simply
outputs a linear combination of its inputs. Proving super-polynomial size lower bounds on such circuits has been a major open problem for decades~\cite{roychowdhury1994lower,williams2018limits}; see also~\cite[Section~2]{Williams2024equalityrank}.
We show that a depth-2 linear circuit lower bound would yield a \emph{subexponential} depth-2 threshold circuit lower bound:

\begin{theorem} [Theorem \ref{thm:threshold-circuit}, \ref{thm:clb-hadamard}]
    If $R_d$ of side-length $N$ does not have a depth-2 linear circuit of size $N^{1 + o(1)}$, then Boolean Inner Product on $n$-bit vectors does not have     a $2^{o(n)}$-size $\mathsf{SUM} \circ \mathsf{ETHR}$ circuit. The same holds for $H_d$ and Boolean Inner Product Mod 2.
\end{theorem}

\section{Technical Overview} \label{sec:technicaloverview}

\vspace{-6pt}\paragraph{Rebalancing Process}

We first introduce the rebalancing approach of prior work which gives Lemma~\ref{intro:rebalancing}. Suppose there are $s$ depth-2 circuits $C_1,\dots,C_s$ computing the matrix $M$, given by
the decompositions
\[ M = \sum_{i \in I_t} U_i V_i^\sfT \]
for $t = 1,\dots,s$ and vectors $U_i, V_i$. The following type of strategy to combine the advantages of the circuits $C_1,\dots,C_s$ to design a circuit for $M^{\otimes d}$ is called a \emph{rebalancing process}.
It starts with a trivial decomposition of $M^{\otimes 0} = I$ and recursively
decomposes $M^{\otimes d}$ for $d = 1,2,\dots$.
Suppose we have decomposed $M^{\otimes d}$ as
\[ M^{\otimes d} = \sum_{j\in J_d} U_j V_j^\sfT, \]
then the Kronecker power $M^{\otimes (d+1)}$ can be decomposed as
\[ M^{\otimes (d+1)} = \sum_{j\in J_d} (U_j V_j^\sfT) \otimes M, \]
where one has the freedom to choose which circuit to use to compute $M$ separately for each $j$.
If one chooses to assign $i_j$ for each $j\in J_d$, then the decomposition of $M^{\otimes (d+1)}$ can be written as
\[ M^{\otimes (d+1)} = \sum_{j\in J_d} \sum_{i\in I_{i_j}} (U_j\otimes U_{i_j})(V_j \otimes V_{i_j})^\sfT , \]
where $I_{i_j}$ is the index set of the decomposition of $M$ chosen for the $j$th term. Prior work~\cite{AGP23KroneckerCircuits,Sergeev22KroneckerCoverings} gives a particular scheme for assigning the $i_j$'s, and analyzes it to yield Lemma~\ref{intro:rebalancing} above.

\vspace{-6pt}\paragraph{The $\alpha$-volume} For each decomposition $C_t$, we define the $\alpha$-volume of the decomposition as
follows
\[ \rho_t(\alpha) = \sum_{i\in I_t} \nnz(U_i)^\alpha \nnz(V_i)^{1-\alpha}. \]

This is a new quantity we introduce to measure the complexity of a decomposition. For any $\alpha \in [0, 1]$, we show any circuit of $M^{\otimes d}$ generated by a rebalancing process
has size at least $(\min_t \rho_t(\alpha))^d$. This can be proved by induction:
\begin{align*}
    \quad \sum_{j\in J_d} \sum_{i\in I_{i_j}} \nnz(U_j\otimes U_{i_j})^\alpha \nnz(V_j &\otimes V_{i_j})^{1-\alpha}
    = \sum_{j\in J_d} \nnz(U_j)^\alpha \nnz(V_j)^{1-\alpha} \sum_{i\in I_{i_j}} \nnz(U_{i_j})^\alpha \nnz(V_{i_j})^{1-\alpha}\\
    &\geq \sum_{j\in J_d} \nnz(U_j)^\alpha \nnz(V_j)^{1-\alpha} \left(\min_{1\leq t\leq s} \rho_t(\alpha)\right) 
    \geq \left(\min_{1\leq t\leq s} \rho_t(\alpha)\right)^{d+1}.
\end{align*}
Furthermore, since $\nnz(U)+\nnz(V) \geq \Omega(\nnz(U)^\alpha \nnz(V)^{1-\alpha})$, this gives a lower bound
on the size of any circuit generated by a rebalancing process.

The previous works \cite{AGP23KroneckerCircuits,Sergeev22KroneckerCoverings} noticed
that the case $\alpha = 1/2$ gives a lower bound, and provided some heuristic strategies and
analyses, showing that certain constructions can achieve the lower bound at $\alpha = 1/2$. Lemma~\ref{intro:rebalancing} was proved in this way, and the ``sufficiently imbalanced'' assumption was needed for the heuristics to apply. (The definition of ``sufficiently imbalanced'' is somewhat long and technical -- see \cite[Pages 4-5]{AGP23KroneckerCircuits} -- but with hindsight, it can be seen as a sufficient condition for the quantity $\rho_t(\alpha)$ to be minimized at $\alpha = 1/2$.)

In Section \ref{sec:spec}, we will show that these $\alpha$-volumes do not only give a lower bound,
but, in fact, also give a tight bound on the size of the circuit which can be generated by a rebalancing process,
i.e., we will show there is a sequence of choices throughout rebalancing which achieves size
\[ \rho^{d+o(d)}, \quad \rho := \left(\sup_{\alpha\in [0, 1]} \min_{1\leq t\leq s} \rho_t(\alpha)\right). \]
We call this the \emph{duality theorem} of depth-2 circuits (Theorem \ref{thm:main}).

\vspace{-6pt}\paragraph{Proof via Strassen's Duality}
Our proof of the duality theorem begins with the key observation that the theory of the asymptotic
spectrum can capture the optimal balancing of depth-2 circuits. In brief, the asymptotic spectrum
considers a commutative semiring $\calR$ equipped with a compatible preorder $\leqslant$ (the
\emph{Strassen preorder}), which preserves both addition  and multiplication ($a \leq b$ implies $a+c \leq b+c$ and $a \cdot c \leq b \cdot c$) along with some additional
technical conditions we discuss in Section~\ref{sec:strassen-duality} below. For this preorder, the rank of an element $a \in \calR$ is defined as
$\R(a) = \min \{ n \in \mathbb{N} : a \leqslant n \}$,
and this rank encodes the information of interest. See Section \ref{sec:strassen-duality} for the detailed
definitions.

Unlike previous applications of the asymptotic spectrum to tensors~\cite{strassen88spectrum}
and graphs~\cite{Zui19Graph}, where the semirings involve concrete elements, our approach
requires a more abstract semiring. Let $M$ denote the matrix under consideration. The commutative
semiring $\calR$ is generated by a free element $\freeM$ (representing $M$) and two families of specified
integers $n^{(1)}$ and $m^{(2)}$ for $n, m\in \bbZ^+$, which represent the number of wires allocated to
the circuit's two layers. The preorder $\leqslant$ is defined by including all possible decompositions
of $M^{\otimes k}$. Specifically, for any decomposition $I$ of $M^{\otimes k}$,
$M^{\otimes k} = \sum_{i\in I} U_i V_i^\sfT$,
we introduce the relation
\[\freeM^k \leqslant \sum_{i\in I} \nnz(U_i)^{(1)} \nnz(V_i)^{(2)}.\]
We also include the relation $n^{(1)}m^{(1)} \leqslant n+m$ in the preorder. With these generating
relations, one asks: what is the rank $\R(\freeM^n)$ of $\freeM^n$? Observe that any
rebalancing process yields an upper bound of the form $\freeM^n \leq r$ as follows. Starting from $\freeM^n$, we
can replace a factor $\freeM^k$ using its decomposition to obtain
$\freeM^n \leqslant \freeM^{n-k}\sum_{i\in I} \nnz(U_i)^{(1)} \nnz(V_i)^{(2)}$,
and then apply the same replacement to each summand $\nnz(U_i)^{(1)} \nnz(V_i)^{(2)}\freeM^{n-k}$ separately. Eventually,
when each summand takes the form $n^{(1)}m^{(2)}$, it is bounded above by $n+m$, representing the
cost of a single internal node $UV^\sfT$. A careful analysis shows that $\R(\freeM^n)$ equals
the smallest circuit size over all possible rebalancing processes.

This characterization allows us to conclude that the asymptotic depth-2 circuit size coincides with the asymptotic rank
$\AR(\freeM) = \lim_{n\to \infty} \R(\freeM^n)^{1/n}$,
which is precisely characterized by Strassen's duality theorem (Theorem \ref{thm:strassen}). As argued earlier, for each $\alpha \in [0, 1]$,
the $\alpha$-volume $\rho_t(\alpha)$ serves as an \emph{obstruction} to obtaining small depth-2 circuits.
Strassen's duality theorem asserts that this family of obstructions is \emph{complete}; that is, the
asymptotic upper bound exactly matches the best lower bound provided by the $\alpha$-volumes. This
equivalence is the essence of our duality theorem.
To make all these arguments rigorous, see further details in Section \ref{sec:spec}.

Although Strassen's duality facilitates the proof of the duality theorem, its nonconstructive nature is a drawback
that prevents us from obtaining a constructive bound on the size of depth-2 circuits. We thus also provide a constructive proof
in Appendix~\ref{sec:tree}. See
Section~\ref{sec:discuss} for further discussion.

\vspace{-6pt}\paragraph{Disjointness Matrix}

Since $R_d$ is a binary matrix, one approach to finding a decomposition of $R_d$ is to restrict our attention to \emph{partitions} of $R_d$. This means expressing the 1s of $R_d$ as the disjoint union of several sets of the form $A \times B$, where $A$ and $B$ are subsets of the rows and columns, respectively, of $R_d$. We refer to these sets as \emph{rectangles}. All previous works on the disjointness matrix $R_d$, including our own, can be viewed as efforts to find such a partition.

In his work, Sergeev~\cite{Sergeev22KroneckerCoverings} proposed a partition $C_t$ of the disjointness matrix $R_{t}$ using the following strategy: starting with the entire matrix, he allowed $k$ to grow from $0$ to $t$. At each step, he added rows and columns indexed by vectors in $\{0,1\}^t$ with $k$ ones, and then removed these rows and columns from the matrix. He demonstrated that when $t=15$, along with the trivial partition of $R_1$ (which has volume $1 + 2^{1-\alpha}$ or $1 + 2^{\alpha}$), these two types of partitions could be rebalanced to yield a depth-2 circuit for $R_d$ of size $O(N^{1.251})$. This is because $\rho_{15}^{1/15}(1/2) < 2^{1.251}$, and the partition satisfies his imbalance condition.

However, Sergeev observed that the optimal value of $\rho_{t}^{1/t}(1/2)$ from this strategy is actually attained at $t=18$, where $\rho_{18}^{1/18}(1/2) < 2^{1.2503}$. Unfortunately, his partition for $t=18$ is not imbalanced enough to apply his strategy. By utilizing our Theorem \ref{thm:weak-duality}, we can show that using only Sergeev's partition and the trivial partition for $R_1$, one cannot achieve the bound given by $\rho_{18}^{1/18}(1/2)$. This is due to the existence of an $\alpha > 1/2$ such that $\min(\rho_{18}^{1/18}(\alpha), 1 + 2^{1-\alpha}) > \rho_{18}^{1/18}(1/2)$.

To address this, we introduce a new class of partitions for $R_t$. In Sergeev's construction, the partition is created by alternatingly adding rows and columns. We generalize this approach by adding rows more frequently than columns, allowing us to reduce the $\alpha$-volume of the partition in certain regions where $\alpha > 1/2$. The precise construction and analysis are presented in Section \ref{sec:generalize-sergeev}.

Finally, we present an improved partition of $R_{18}$ that yields a circuit of size $O(N^{1.2495})$ for $R_d$. Instead of merely adding a row and a column at each step, we can choose a vector $x \in \{0,1\}^t$ and add the rectangle formed by row $x$ along with all rows gotten from changing one $1$ in $x$ to a $0$, with the columns corresponding to their common 1s. The selection of $S$ is carefully designed and involves a packing problem, which can be reduced to known results in the construction of constant weight codes. The detailed construction and analysis are provided in Section \ref{sec:complex}.

\vspace{-6pt}\paragraph{Improving the Degree}

Let $C$ be a depth-2 circuit for $M$ that expresses $M$ as $M = U V^{\sfT}$. Let $\nnz(U_i)$ and $\nnz(V_i)$ denote the degrees of the $i$th output and input gates of $C$, respectively (i.e., the $i$th row of $U$ and $V$). We define the \emph{degree} of $C$ as the maximum of all these degrees.

We are interested in the problem of \emph{sparse vector-matrix-vector multiplication}: given a matrix $M$ and two vectors $x$ and $y$, each with at most $n$ nonzero entries, compute $x^{\sfT} M y$. This problem has a natural connection to problems like $\#\OV$, as the $d$-bit $\#\OV$ problem can be reduced to a sparse vector-matrix-vector multiplication instance as described above. In this reduction, the matrix $M$ is the disjointness matrix $R_d$, and $x$ and $y$ are the indicator vectors of the input sets.

Intuitively, the sparse vector-matrix-vector multiplication problem can be solved in $\tilde{O}(n \cdot D)$ time,
where $D$ is the degree of the circuit for $M$. The idea is to compute the inner product from
two vector-matrix multiplication results: $x^\sfT M y = (x^\sfT U) (y^\sfT V)^\sfT$. Since each
row of $U$ and $V$ has at most $D$ nonzero entries, the vector-matrix multiplication yields at most $nD$ nonzero
entries, which can be computed in $\tilde{O}(nD)$ time. To formalize this intuition, we use
the Kronecker power structure of $R_d$, enabling efficient identification of the nonzero entries
in each row of $U$ and $V$ using a fixed decomposition. Therefore, a low degree decomposition of $R_d$ yields an efficient algorithm for $\#\OV$.

\vspace{-6pt}\paragraph{Matrix with Full Symmetry} We first consider the case where an $n\times m$ matrix $M$ has \emph{full symmetry}, meaning $M$ is invariant under a permutation group $G$ permuting the rows and columns of $M$, such that each row can be permuted by $G$ to any other row, and similarly for columns.

Consider a depth-2 circuit $C$ for $M$, and let $r_i$  and $c_i$ be the indegree and outdergee of the $i$th gate, respectively.
By directly taking the Kronecker power of the circuit, we only get a circuit of row maximal degree $(\max\{c_i\})^d$
and column maximal degree $(\max\{r_i\})^d$ when computing $M^{\otimes d}$.

Now consider the geometric mean of the degrees, i.e., $\overline{r} = (\prod_{i=1}^n r_i)^{1/n}$ and
$\overline{c} = (\prod_{j=1}^m c_j)^{1/m}$.

In the Kronecker power construction $C^{\otimes d}$, the degree of the $(i_1,\dots,i_d)$th row
is $r_{i_1} \cdot r_{i_2} \cdots r_{i_d}$. By using concentration inequalities, one can see that
most of the rows have degree close to $\overline{r}^d$, which is far smaller than $(\max\{r_i\})^d$
unless all $r_i$ are equal. The phenomenon is similar for columns.

If we delete those rows and columns with degree far from $\overline{r}^d$ and $\overline{c}^d$, the result is a circuit with much smaller maximal degree, but for computing a submatrix of $M^{\otimes d}$. We then use the symmetry of $M$ and a technique called \emph{fixing holes} to convert this into a circuit for the whole matrix.

\vspace{-6pt}\paragraph{Fixing Holes}

Let $M' = M^{\otimes d}$ having $\mathcal{N}=n^d$ rows and $\mathcal{M}=m^d$ columns,
and $C'$ be a circuit computing $M'|_{S, T}$ for some $|S| \geq (1-o(1)) \mathcal{N}$ and $|T| \geq (1-o(1)) \mathcal{M}$, where the rows and columns with degree above $\overline{r}^{d+o(d)}$ and $\overline{c}^{d+o(d)}$ are excluded.
The symmetry of $M$ can be extended to a full symmetry of $M^{\otimes d}$, by permuting the rows and columns
entry-wise. These two properties help us to extend the circuit $C'$ to a circuit for the whole $M'$.

Let $\sigma$ be a random permutation from the symmetry of $M'$, let
$\sigma(C')$ be the circuit obtained by permuting the rows and columns of $C'$ according to $\sigma$.
Then the circuit $\sigma(C')$ is a circuit for $M^{\otimes d}|_{\sigma(S), \sigma(T)}$. Now consider any
submatrix $M'|_{S', T'}$, with high probability, the rectangle $\sigma(S)\times \sigma(T)$
will overlap with $S'\times T'$ in a large portion, and the circuit $\sigma(C')$ can be used to compute
the submatrix $M'|_{S' \cap \sigma(S), T' \cap \sigma(T)}$. Therefore, the task of computing $M'$
over $S'\times T'$ is already resolved in a large portion, and the rest part breaks into two smaller
subrectangles
\vspace{-5pt}
\[ (S' \setminus \sigma(S)) \times T', \quad (S' \cap \sigma(S)) \times (T' \setminus \sigma(T)). \]
\vspace{-20pt}

One can then recursively apply the same technique to compute the two subrectangles.

By careful analysis, one can demonstrate that, the circuit for $M'$ can be constructed using $(\mathcal{NM})^{o(1)}$ fragments of $C'$. Returning to the original problem, we can then construct a circuit for $M^{\otimes d}$ with $(n^d m^d)^{o(1)}$ fragments of $C^{\otimes d}$. The maximal degree of the circuit is $\overline{r}^{d + o(d)}$ and $\overline{c}^{d + o(d)}$ for the rows and columns, respectively. The formal statement is provided in Theorem~\ref{thm:fix-holes}.

Our approach also applies to matrices like the Walsh--Hadamard matrices and discrete Fourier transform matrices, which have a more general type of full symmetry where one may also need to rescale rows and columns. In Section \ref{sec:fix-holes}, we formally define the general symmetry we allow.

This argument is referred to as the tensor \emph{hole-fixing lemma} in recent algorithms for fast matrix multiplication~\cite{DWZ23AsymmetricHashing,WXXZ24MatrixMultiplication,alman2024more}; we use a similar matrix version here.

\vspace{-6pt}\paragraph{Disjointness Matrix} The disjointness matrix $R_d$, unfortunately, does not have full symmetry,
which prevents one from improving the maximal degree of circuits directly.
However, it does have a \emph{partial symmetry}: it can be partitioned into $O(d^2)$ submatrices, corresponding to rows and columns of fixed Hamming weights, which each have full symmetry.
We thus use the hole fixing lemma separately to reduce the maximal degree of each submatrix.

Surprisingly, compared with Williams'~\cite{Williams2024equalityrank} previous $\#\OV$ algorithm, which can be viewed as a circuit construction of
degree $O(5^{d/5}) < 1.38^{d}$, and which was designed using SAT solvers to find a good decomposition of $R_5$, our improved construction for Theorem~\ref{intro:ov-par}
giving degree $2^{d/3 + o(d)} < 1.26^{d + o(d)}$ starts with a simple decomposition of $R_2$ and can be verified by hand. (Williams' approach makes circuits with a certain structure based on equality matrices; our new circuit does not have this structure.) See Corollary \ref{cor:disjointness}
for the construction and the proof.

\subsection{Outline}

In Section \ref{sec:prelim}, we provide the preliminaries and notation used throughout the paper.

In Section \ref{sec:spec}, we define a commutative semiring equipped with a Strassen preorder,
which captures depth-2 circuits. Then we invoke Strassen's theory of asymptotic spectrum to prove
the main theorem (Theorem \ref{thm:main}).

In Section \ref{sec:improve-circuit}, we show how to improve the size of circuits for the disjointness matrix
first to $O(N^{1.2503})$ (Section \ref{sec:generalize-sergeev}) and then to $O(N^{1.2495})$
(Section \ref{sec:complex}).

In Section \ref{sec:fix-holes}, we introduce the hole-fixing lemma (Theorem \ref{thm:fix-holes})
for depth-2 circuits and use it to prove the main theorem (Theorem \ref{thm:fix-holes-main}) for
reducing the maximal degree of circuits. We then use the main theorem to bound the maximal degree of
circuits for DFT matrices (Corollary \ref{cor:dft}) and the disjointness matrix (Corollary \ref{cor:disjointness}).

Finally, in Section \ref{sec:algorithms}, we show how to use the improved circuits to design algorithms
(Theorem \ref{thm:algebraic-algorithm}),
and give a new algorithm for $\OV$, $\#\OV$ over $\bbZ$ (Corollary \ref{cor:ov-par}
and \ref{cor:ov}) or $\bbZ_m$ (Corollary \ref{cor:ov-dft}).

\section{Preliminaries} \label{sec:prelim}

For a positive integer $n$, we use $[n]$ to denote $\{1, \dots, n\}$.
We use Iverson bracket notation: if $P$ is a logical proposition, we let $\iv{P}$ be $1$ if $P$ is true
and $0$ otherwise.

For any element $x\in S$ in domain $S$, the indicator vector of $x$ is the vector $\one_x$ such that
$(\one_x)_i = \iv{i = x}$. For a set or a multiset $X$, we use $\one_X$ to denote the sum of indicator vectors of elements in $X$,
i.e., $(\one_X)_i = \sum_{x\in X} \iv{i = x}$.

The binary entropy function is defined as
\[ \H(p) = - p\log_2 p - (1-p) \log_2 (1-p), \]
where we take $0\log 0 = 0$.

We use $\binom{n}{\geq m}$ to denote the partial sum of binomial coefficients, i.e.,
\[ \binom{n}{\geq m} = \sum_{i=m}^n \binom{n}{i}. \]

When $S$ is a set and $m$ is a nonnegative integer, we use $\binom{S}{m}$ to denote the family of all $m$-element subsets of $S$.

\subsection{Matrices}

Let $I, J$ be two index sets, a matrix $M \in \bbF^{I\times J}$ is indexed by $I$ and $J$.
When $I', J'$ are two subsets of $I$ and $J$, we use $M|_{I', J'}$ to denote the submatrix of $M$
indexed by $I'$ and $J'$. We also frequently use $M\in \bbF^{n\times m}$ to denote a matrix with
$n$ rows and $m$ columns, indexed by $[n]$ and $[m]$.
We use $\nnz(M)$ to denote the number of nonzero entries of $M$.

Let $M_1 \in \bbF^{I_1\times J_1}, M_2 \in \bbF^{I_2\times J_2}$ be two matrices, the
\emph{Kronecker} product of $M_1$ and $M_2$ is a matrix $M_1 \otimes M_2 \in \bbF^{(I_1\times I_2) \times (J_1\times J_2)}$,
indexed by $I_1\times I_2$ and $J_1\times J_2$, such that
\[ (M_1 \otimes M_2)_{(i_1, i_2), (j_1, j_2)} = (M_1)_{i_1, j_1} (M_2)_{i_2, j_2}. \]

The Kronecker product satisfies the following properties:
\begin{compactitem}
    \item (bilinearity) For $A, B \in \bbF^{n_1 \times m_1}$ and $C, D \in \bbF^{n_2 \times m_2}$, we have
    $(A+B) \otimes (C+D) = A\otimes C + A\otimes D + B\otimes C + B\otimes D$.
    \item (mixed-product property) For $A, B \in \bbF^{n_1 \times m_1}$ and $C, D \in \bbF^{n_2 \times m_2}$, we have
    $(A\otimes B) \cdot (C\otimes D) = (A\cdot C) \otimes (B\cdot D)$.
    \item $\nnz(A\otimes B) = \nnz(A) \nnz(B)$.
\end{compactitem}

The Kronecker power $M^{\otimes d}$ of a matrix $M$ is defined as the $d$-fold Kronecker product of $M$ with itself.
Therefore, we sometimes write $R_d$ as $R^{\otimes d}$ and $H_d$ as $H^{\otimes d}$.

\subsection{Measures of Matrices}

Let $M$ be an $n\times m$ matrix over $\bbF$. A \emph{(synchronous)\footnote{We focus here on `synchronous' depth-2 circuits here, in which all paths between inputs and outputs are of the same length.
For depth-2 circuits, one may focus on synchronous circuits almost without loss of generality since a general depth-2 circuit
can be transformed into a synchronous depth-2 circuit with at most a doubling of
the number of wires.} depth-2 circuit} $C$ computing $M$ is a pair of matrices $U \in \bbF^{n\times k}, V \in \bbF^{m\times k}$ such that
$M = U V^\sfT$. We will also interchangeably refer to this as a \emph{decomposition} of $M$. Sometimes we we would like to emphasize the different gates in the middle of the circuit, and we write it as
\[ M = \sum_{i\in I} U_i V_i^\sfT, \]
where $I$ is the index set of the decomposition, and each $U_i \in \bbF^n$ and $V_i \in \bbF^m$. Each support $\supp(U_i) \times \supp(V_i)$
is called a \emph{rectangle} in the decomposition.

There are several natural operations on circuits. Let $C$ be a circuit computing $M = UV^\sfT$, and $C'$ be a circuit computing $N = U' V'^\sfT$.
The \emph{Kronecker product} $C\otimes C'$ of $C$ and $C'$ is a circuit computing $M\otimes N$,
given by \[ M\otimes N = UV^\sfT \otimes U'V'^\sfT = (U\otimes U')(V\otimes V')^\sfT. \]
When $M, N$ are matrices of the same size, the \emph{sum} $C+C'$ of $C$ and $C'$ is a circuit computing $M + N$,
given by
\[ M + N = UV^\sfT + U'V'^\sfT = \begin{bmatrix}
    U & U'
\end{bmatrix} \begin{bmatrix}
    V & V'
\end{bmatrix}^\sfT. \]

\vspace{-6pt}\paragraph{Size}
We define the size of the circuit $C$ as $\Size(C) = \nnz(U) + \nnz(V)$.
The sizes of circuits are additive and submultiplicative under the sum and Kronecker product operations, i.e.,
\[ \Size(C + C') = \Size(C) + \Size(C'), \quad \Size(C\otimes C') \leq \Size(C) \Size(C'). \]

This leads to the definition of the (depth-2) circuit size of a matrix $M$ as
\[ \Size(M) = \min_{U V^\sfT = M} \nnz(U) + \nnz(V). \]
By the above properties, we have $\Size(\cdot)$ over matrices is subadditive and submultiplicative,
i.e.,
\[ \Size(M + N) \leq \Size(M) + \Size(N), \quad \Size(M\otimes N) \leq \Size(M) \Size(N). \]
Therefore, Fekete's lemma implies that the limit
\[ \sigma(M) = \lim_{n\to\infty} \Size(M^{\otimes n})^{1/n} \]
exists and is upper bounded by $\Size(M^{\otimes n})^{1/n}$ for any $n$. We call $\sigma(M)$ the \emph{asymptotic} (depth-2) circuit size of $M$.

\vspace{-6pt}\paragraph{Degree}
Let matrix $M$ be indexed by $I\times J$, let $C$ be a circuit computing $M$ with $U\in \bbF^{I\times k}, V\in \bbF^{J\times k}$.
The \emph{row degree} of $i\in I$ is defined as $r_C(i) = \nnz(U_i)$, and the \emph{column degree} of $j\in J$ is defined as $c_C(j) = \nnz(V_j)$.
The \emph{maximal degree} $\Degree(C)$ of the circuit $C$ is defined as the maximum of all row and column degrees.

Let $C, C'$ be two circuits computing $M, N$ respectively, the degrees satisfy
\[ r_{C+C'}(i) = r_C(i) + r_{C'}(i), \quad r_{C\otimes C'}((i_1,i_2)) = r_C(i_1) r_{C'}(i_2), \]
and similarly for columns. Therefore, the maximal degree has the properties
\[ \Degree(C + C') \leq \Degree(C) + \Degree(C'), \quad \Degree(C\otimes C') \leq \Degree(C) \Degree(C'). \]

Define $\Degree(M)$ as the minimal $\Degree(C)$ over all circuits $C$ computing $M$. The \emph{asymptotic degree} $\delta(M)$ of $M$ is defined as
\[ \delta(M) = \lim_{n\to\infty} \Degree(M^{\otimes n})^{1/n}. \]

\vspace{-6pt}\paragraph{Equality Rank}

An equality matrix is a matrix $M$ such that $M_{ij} = \iv{f(i) = g(j)}$ for some functions $f: I\to K, g: J\to K$ and any set $K$.
The \emph{equality rank} $\Eq(M)$ of $M$ is the minimal $q$ such that $M$ can be written as a linear combination of $q$ equality matrices.

One can show that $\Eq(M+N)\leq \Eq(M) + \Eq(N)$ and $\Eq(M\otimes N) \leq \Eq(M) \Eq(N)$, thus one can define
the \emph{asymptotic equality rank} of $M$ as
\[ \theta(M) = \lim_{n\to\infty} \Eq(M^{\otimes n})^{1/n}. \]

\vspace{-6pt}\paragraph{Relation between Measures}

The measures $\Size(\cdot), \Degree(\cdot)$ and $\Eq(\cdot)$ are related by the following inequalities:
Let $M$ be an $n\times m$ matrix, then
\begin{proposition} \label{prop:basic-inequalities}
    $\frac{\Size(M)}{n+m} \leq \Degree(M) \leq \Eq(M)$.
\end{proposition}
The first inequality is due to the fact that the maximal degree of a circuit computing $M$ is at least the average degree of the rows and columns.
The second inequality is because every nonzero equality matrix has degree $1$.

This yields the asymptotic inequalities
\begin{proposition} \label{prop:basic-asymptotic-inequalities}
    $\frac{\sigma(M)}{\max(n,m)} \leq \delta(M) \leq \theta(M)$.
\end{proposition}

We next show that the asymptotic circuit size and degree of a matrix does not depend on the base field.
\begin{theorem} \label{thm:field-extension}
    Let $M$ be matrix over $\bbF$ and $\bbK$ be a field extension of $\bbF$.
    Then $\sigma_{\bbK}(M) = \sigma_{\bbF}(M)$ and $\delta_{\bbK}(M) = \delta_{\bbF}(M)$.
\end{theorem}
\begin{proof}
    Clearly $\sigma_\bbK(M) \leq \sigma_\bbF(M)$ and $\delta_\bbK(M) \leq \delta_\bbF(M)$ both hold, since every circuit computing $M$ over $\bbF$
    is automatically a circuit computing $M$ over $\bbK$. We only need to show the other direction.

    For any $\epsilon > 0$, let $d$ be a positive integer and $C$ be a depth-2 circuit computing $M^{\otimes d}$ with degree at most
    $(\delta_\bbK(M) + \epsilon)^d$. The existence of this circuit can be rephrased as the existence of a certain
    system of polynomial equations with coefficients in $\bbF$ which has a solution in $\bbK$.
    Let $I$ be the polynomial ideal, where the indeterminates are the coefficients of the circuit $C$,
    and the equations are the constraints of the circuit $C$ to let it compute $M^{\otimes d}$.
    Then since $I$ has a solution in $\bbK$, we have $1\notin I$. Therefore, Hilbert's Nullstellensatz
    (see, e.g., \cite[IX, \S1]{lang2012algebra}) implies that
    $I$ also has a solution in the algebraic closure $\overline{\bbF}$. Let $C'$ be one of the
    solution of $I$ in $\overline{\bbF}$, and let $\bbE$ be the extension of $\bbF$ generated by the coefficients of $C'$.
    Since each coefficient in $C'$ is algebraic over $\bbF$, we have $\bbE$ is a finite extension of $\bbF$,
    let $\alpha_1,\dots,\alpha_{[\bbE:\bbF]}$ be a basis of $\bbE$ over $\bbF$ and $\alpha_1=1$.
    Now taking Kronecker powers of $C'$, we have $\bbE$-circuits computing
    \[ M^{\otimes n} = \sum_{i} U_i V_i^\sfT, \]
    and we can write $U_i$ and $V_i$ as a linear combination of $\alpha_1,\dots,\alpha_{[\bbE:\bbF]}$,
    i.e.,
    \[ U_i = \sum_{j=1}^{[\bbE:\bbF]} U_{i}^{(j)} \alpha_j, \quad
    V_i = \sum_{j=1}^{[\bbE:\bbF]} V_{i}^{(j)} \alpha_j, \]
    where $U_{i}^{(1)} = V_{i}^{(1)}$ are vectors over $\bbF$.
    Then we have
    \begin{align*}
        M^{\otimes n} &= \sum_{i} \left(\sum_{j=1}^{[\bbE:\bbF]} U_{i}^{(j)} \alpha_j\right) \left(\sum_{j=1}^{[\bbE:\bbF]} V_{i}^{(j)} \alpha_j\right)^\sfT\\
        &= \sum_{i} \sum_{j,k=1}^{[\bbE:\bbF]} \alpha_j \alpha_k U_{i}^{(j)} V_{i}^{(k)\sfT}.
    \end{align*}
    With the priori that $M$ is over $\bbF$, we extract the $\alpha_1$ coefficient of each $\alpha_j \alpha_k$ 
    to get the circuit computing $M^{\otimes d}$ over $\bbF$. This $\bbF$-circuit has size
    and degree at most a $[\bbE:\bbF]^2$ factor blow-up compared with the $\bbF$-circuit,
    i.e., having size $O((\sigma_\bbK(M)+\epsilon)^n)$ and degree $(\delta_\bbK(M)+\epsilon)^n$.
    Thus we have $\sigma_\bbF(M) \leq \sigma_\bbK(M)+\epsilon$ and $\delta_\bbF(M) \leq \delta_\bbK(M)+\epsilon$.
\end{proof}

\vspace{-6pt}\paragraph{Covering and Partitioning}

Let $M\in \{0, 1\}^{I\times J}$ be a binary matrix. A \emph{covering} of $M$ is a family of
subsets $I_i \subset I, J_i \subset J$ such that
\[ M = \bigvee_{i} \one_{I_i} \times \one_{J_i}^\sfT. \]
This can be interpreted as a depth-2 circuit computing $M$ with OR gates, or as a depth-2 circuit replacing $\bbF$
with the Boolean semiring $\mathbb{B} = (\{0, 1\}, \vee, \wedge)$. We can equivalently define it by treating $M$ as a subset of $I\times J$, so that a covering is a family of rectangles such that
their union is $M$, i.e.,
\[ \supp(M) = \bigcup_{i} I_i \times J_i. \]
This gives the covering size $\Size_\OR(M)$ and covering degree
$\Degree_\OR(M)$ of $M$, and their asymptotic counterparts $\sigma_{\OR}(M)$ and $\delta_{\OR}(M)$.
For example, the covering size $\Size_\OR(M)$ is
\[ \Size_\OR(M) = \min_{\bigcup_{i} I_i \times J_i = \supp(M)} \left(\sum_i |I_i| + |J_i|\right). \]
One can also define the covering equality rank
$\Eq_{\OR}(M)$ of $M$ as the minimal $q$ such that $M$ can be written as a union of $q$ equality matrices,
and the asymptotic covering equality rank $\theta_{\OR}(M)$ of $M$ can be defined similarly.

Similarly, a \emph{partitioning} of $M$ is a family of subsets $I_i \subset I, J_i \subset J$ such that
\[ M = \sum_{i} \one_{I_i} \times \one_{J_i}^\sfT, \]
which can be interpreted as a depth-2 circuit computing $M$ in which all gates have coefficient $1$.
An alternative definition is that, regarding $M$ as a subset of $I\times J$,
a partitioning is a family of rectangles such that they
form the \emph{disjoint union} of $M$, i.e.,
\[ \supp(M) = \bigsqcup_{i} I_i \times J_i. \]
This gives the partitioning size $\Size_\PAR(M)$ and partitioning degree
$\Degree_\PAR(M)$ of $M$, and their asymptotic counterparts $\sigma_{\PAR}(M)$ and $\delta_{\PAR}(M)$.
For example, the partitioning size $\Size_\PAR(M)$ is
\[ \Size_\PAR(M) = \min_{\bigsqcup_{i} I_i \times J_i = \supp(M)} \left(\sum_i |I_i| + |J_i|\right). \]
The partitioning equality rank $\Eq_{\PAR}(M)$ and the asymptotic partitioning equality rank $\theta_{\PAR}(M)$
can be defined similarly.

The basic relations (proposition \ref{prop:basic-inequalities} and \ref{prop:basic-asymptotic-inequalities})
hold for covering and partitioning measures as well. Moreover, since a binary matrix $M$ can be naturally identified a matrix over any field $\bbF$, we denote $\Size_\bbF(M)$, $\Degree_\bbF(M)$, $\Eq_\bbF(M)$
and their asymptotic counterparts of the measure when identifying $M$ as a matrix over $\bbF$.
The definition of partition measures gives the following inequalities.
\begin{proposition}
    For any binary matrix $M$ and field $\bbF$, we have
    $F_\bbF(M)\leq F_\PAR(M)$ and $F_\OR(M)\leq F_\PAR(M)$,
    when $F$ is $\Size$, $\Degree$, $\Eq$ or their asymptotic counterparts.
\end{proposition}

\subsection{Inequalities}

\begin{theorem}[Estimate of binomial coefficients]
    \[ \frac 1{n+1} 2^{\H(m/n) n} \leq \binom{n}{m} \leq 2^{\H(m/n)n}. \]
\end{theorem}

\begin{theorem}[Young's inequality]
    Let $a, b, p, q > 0$ be real numbers such that $1/p + 1/q = 1$. Then
    \[ ab \leq \frac{a^p}{p} + \frac{b^q}{q}. \]
\end{theorem}

\begin{theorem}[Chernoff bound]
    Let $X_1,\dots,X_n$ be independent random Bernoulli variables, and let $X = \sum_{i=1}^n X_i$. We have
    \[ \Pr[X \geq (1 + \delta) \bbE[X]] \leq \exp\left(\frac{-\delta^2 \bbE[X]}{3}\right) \]
    for $\delta \in [0, 1]$.
\end{theorem}

\begin{theorem}[Hoeffding's inequality]
    Let $X_1,\dots,X_n$ be independent random variables such that $a_i \leq X_i \leq b_i$, and let $X = \sum_{i=1}^n X_i$. For any $t > 0$, we have
    \[ \Pr[X - \Ex[X] \geq t] \leq \exp\left( -\frac{2t^2}{\sum_{i=1}^n (b_i - a_i)^2} \right). \]
\end{theorem}

\subsection{Strassen Duality} \label{sec:strassen-duality}

In this subsection we introduce Strassen's theory of asymptotic spectra \cite{strassen88spectrum}. We follow the presentation in the survey of Wigderson and Zuiddam \cite{WZ22spectra}.

In this paper, a \emph{commutative semiring} $(\calR, +, \cdot, 0, 1)$ is a set $\calR$ with two operations $+$ and $\cdot$
such that
\begin{compactenum}
    \item Addition $(\calR, +, 0)$ is a commutative semigroup with identity $0$.
    \item Multiplication $(\calR, \cdot, 1)$ is a commutative semigroup with identity $1$.
    \item Multiplication distributes over addition, i.e., for any $a, b, c\in \calR$, we have
    $a\cdot (b+c) = a\cdot b + a\cdot c$.
\end{compactenum}
For any nonnegative integer $n$, we recognize $n\in \calR$ as the sum of $n$ copies of $1$.
We write $ab$ as a shorthand for $a\cdot b$, and $a^n$ as the $n$-fold product of $a$.

A preorder $\leqslant$ on $\calR$ is called a \emph{semiring preorder} if it satisfies the following properties for all $a,b,c\in \calR$
\begin{compactenum}
    \item $0\leqslant a$. \hfill (Nonnegative)
    \item $a\leqslant a$ . \hfill (Reflexive)
    \item $a\leqslant b$ and $b\leqslant c$ implies $a\leqslant c$. \hfill (Transitive)
    \item $a\leqslant b$ implies $a+c\leqslant b+c$. \hfill (Additive)
    \item $a\leqslant b$ implies $a\cdot c\leqslant b\cdot c$. \hfill (Multiplicative)
\end{compactenum}
\begin{definition}[Strassen preorder]
    A semiring preorder $\leqslant$ on a commutative semiring $\calR$ is called a \emph{Strassen preorder} if
    it satisfies the following additional properties:
    \begin{compactenum}
        \item For any nonnegative integers $n, m$, 
        we have $n\leqslant m$ in $\calR$ if and only if $n \leq m$.
        \hfill (Embedding of natural numbers)
        \item For any $a\in \calR \setminus \{0\}$, there exists a positive integer $n$ such that $1\leqslant a \leqslant n$.
        \hfill (Strong Archimedean property)
    \end{compactenum}
\end{definition}

Given a Strassen preordered semiring $\calR$, the \emph{rank} of an element $a\in \calR$ is the minimal integer $n$ such that $a\leqslant n$,
denoted as $\R(a)$. The \emph{asymptotic rank} of $a$ is the limit
\[ \AR(a) = \lim_{n\to\infty} \R(a^{n})^{1/n}, \]
which is always well-defined.

\begin{definition}[Asymptotic spectrum]
    The \emph{asymptotic spectrum} $\calX$ of a Strassen preordered semiring $\calR$ is the set of functions
    $\phi \colon \calR \to \bbR$, such that
    \begin{compactenum}
        \item $\phi(1) = 1$. \hfill (Normalized)
        \item For any $a, b\in \calR$, we have $\phi(a+b) = \phi(a) + \phi(b)$. \hfill (Additive)
        \item For any $a, b\in \calR$, we have $\phi(ab) = \phi(a) \phi(b)$. \hfill (Multiplicative)
        \item For any $a, b\in \calR$, we have $a\leqslant b$ implies $\phi(a) \leq \phi(b)$. \hfill (Order-preserving)
    \end{compactenum}
\end{definition}

We need the following characterization of asymptotic rank.
\begin{theorem}[\cite{strassen88spectrum}; see {\cite[Theorem 3.25]{WZ22spectra}}] \label{thm:strassen}
    Let $\calR$ be a Strassen preordered semiring, and $\calX$ be the asymptotic spectrum of $\calR$.
    Then for any $a\in \calR$, we have
    \[ \AR(a) = \sup_{\phi \in \calX} \phi(a). \]
\end{theorem}

\section{Asymptotic Spectrum of Depth-2 Circuits} \label{sec:spec}

In this section, we introduce an abstract semiring for depth-2 circuits that captures the decomposition
of a matrix into a sum of rank-1 matrices, and apply Strassen's duality theory to obtain
the duality theorem for depth-2 circuits.

We first define a commutative semigroup $(S, \cdot)$ as follows. The semigroup is generated by a free element $\freeM$
and two sets of specified families of positive integers $\{n^{(1)} : n \in \bbZ^+\}, \{n^{(2)} : n \in \bbZ^+\}$.
More precisely, every element in $S$ can be written uniquely in one of the following forms:
\begin{compactitem}
    \item $\freeM^k$ for some $k\in \bbN$.
    \item $\freeM^k n^{(1)}$ for some $k\in \bbN$ and $n\in \bbZ^+$.
    \item $\freeM^k m^{(2)}$ for some $k\in \bbN$ and $m\in \bbZ^+$.
    \item $\freeM^k n^{(1)} m^{(2)}$ for some $k\in \bbN$ and $n, m\in \bbZ^+$.
\end{compactitem}
The multiplication is given by $\freeM^n \cdot \freeM^m = \freeM^{n+m}$ and
$n^{(1)} \cdot m^{(1)} = (nm)^{(1)}$, $n^{(2)} \cdot m^{(2)} = (nm)^{(2)}$.

Then we construct the semiring $\calR$ by the ``group algebra'' of $S$, i.e., the set of formal sums
\[ \calR = \left\{ \sum_{s\in S} a_s s : a_s\in \bbN \right\}, \]
where the sum is taken over finitely many elements of $S$. The addition is the pointwise addition of coefficients,
and the multiplication is given by the distributive law
\[ \left(\sum_{s\in S} a_s s\right) \cdot \left(\sum_{t\in S} b_t t\right) = \sum_{s, t\in S} a_s b_t (s\cdot t). \]

Finally, we define a semiring preorder $\leqslant_M$ respect to a matrix $M$ as follows.
\begin{definition} \label{def:semiring-preorder}
    Let $M$ be an $p\times q$ matrix, let $\leqslant_M$ be the minimal semiring preorder on $\calR$ such that the following properties hold:
    \begin{compactenum}
        \item For any $n \in \bbZ^+$, we have $1\leqslant \freeM$, $1\leqslant n^{(1)}$, $1\leqslant n^{(2)}$.
        \label{item:lower-bound} \hfill (Lower bound) 
        \item For any decomposition $I$ of $M^{\otimes k}$ into rank-1 matrices, we have \label{item:decomposition}
        \[ \freeM^k \leqslant \sum_{i\in I} \nnz(U_i)^{(1)} \nnz(V_i)^{(2)}. \] 
        \item For any $n, m\in \bbZ^+$, we have $n^{(1)} m^{(1)} \leqslant n+m$.
        \label{item:cost-decomposition} \hfill (Cost of decomposition)
    \end{compactenum}
\end{definition}

We need the following lemma to analyze the structure of the semiring preorder $\leqslant_M$.
We say that the relation $a\leqslant_M b$ is a \emph{one step} relation if $a, b$ have the form
$u + sx, u + sy$, where $s\in S$ and $x\leqslant y$ comes from one of
the basic rules in Definition \ref{def:semiring-preorder},
or if $a, b$ have the form $u, u + x$ for some $x\in S$.

\begin{lemma} \label{lem:one-step}
    For any $a, b\in \calR$, we have $a\leqslant_M b$ if and only if there exists
    a chain $a = a_0 \leqslant a_1 \leqslant \cdots \leqslant a_T = b$,
    such that $a_i \leqslant_M a_{i+1}$ is a one step relation for all $i$.
\end{lemma}
\begin{proof}
    The relation $\leqslant_M$ can be generated in the following way. First define
    the set $R_0 \subseteq \calR\times \calR$ by the basic rules of the semiring preorder and
    basic rules of $\leqslant_M$, i.e., $R_0$ consists of the pairs:
    \begin{compactitem}
        \item $(0, a)$, $(a, a)$ for all $a\in \calR$.
        \item $(1, \freeM)$, $(1, n^{(1)})$, $(1, n^{(2)})$ for all $n\in \bbZ^+$.
        \item $(\freeM^k, \sum_{i\in I} \nnz(U_i)^{(1)} \nnz(V_i)^{(2)})$ for all $k\in \bbN$
        and all decompositions $I$ of $M^{\otimes k}$.
        \item $(n^{(1)}m^{(2)}, n+m)$ for all $n, m\in \bbZ^+$.
    \end{compactitem}
    Then for nonnegative integer $i$, we define $R_{i+1}$ inductively by adding the pairs:
    \begin{compactitem}
        \item Transitivity: $(a, c)$ for all $a, b, c\in \calR$ such that $(a, b), (b, c)\in R_i$.
        \item Additivity: $(a+c, b+c)$ for all $a, b\in \calR$ such that $a\leqslant_M b$ in $R_i$ and all $c\in \calR$.
        \item Multiplicativity: $(a\cdot c, b\cdot c)$ for all $a, b\in \calR$ such that $a\leqslant_M b$ in $R_i$ and all $c\in \calR$.
    \end{compactitem}
    Thus, the relation $\leqslant_M$ is the union of all $R_i$.
    
    We now prove the lemma by induction on the smallest $i$ such that $(a,b) \in R_i$. The base case $R_0$ is immediate.
    Now suppose all relations in $R_i$ can be generated by one step relations, then:
    \begin{compactitem}
        \item Transitivity: Suppose $(a, b)\in \calR_i$ and $(b, c)\in \calR_i$ can be generated by one step relations.
        Then combining the chains of one step relations for $(a, b)$ and $(b, c)$ gives a chain for
        $(a, c)\in R_{i+1}$.
        \item Additivity: Suppose $(a, b)\in \calR_i$ has a chain $a_0\leqslant \cdots \leqslant a_T$
        of one step relations. Then for any $c\in \calR$, we have $a_0 + c\leqslant \cdots \leqslant a_T + c$
        is a chain of one step relations for $(a + c, b + c)\in R_{i+1}$.
        \item Multiplicativity: Suppose $(a, b)\in \calR_i$ has a chain $a_0\leqslant \cdots \leqslant a_T$,
        we want to show that each $ac \leqslant bc$ can be generated by one step relations.
        We first assume $c\in S$ is a single element. If the one step relation
        $u + sx \leqslant u + sy$ multiplied by $c$ gives $uc + scx \leqslant uc + scy$,
        which is still a one step relation. The case $u \leqslant u + x$ is similar.
        The general case $c = s_1 + s_2 + \cdots + s_T$ follows from
        the distributive law and the above two cases.
    \end{compactitem}
    In conclusion, all relations in $R_{i+1}$ can be generated by one step relations.
    By induction, all relations in $\leqslant_M$ can be generated by one step relations.
\end{proof}

We study the semiring preorder $\leqslant_M$ by considering a rank function $r_M\colon \calR\to \bbR$ defined as follows.
For any $\freeM^k n^{(1)} m^{(2)} \in S$, we define $r_M(\freeM^k n^{(1)} m^{(2)})$ to be the minimum
of
\[ \sum_{i\in I} n\nnz(U_i) + m\nnz(V_i), \]
where $I$ ranges over all decompositions of $M^{\otimes k}$ into rank-1 matrices.

Then for the rest of the elements in $S$, we define $r_M$ by
\begin{compactitem}
    \item $r_M(\freeM^k) = r_M(\freeM^k 1^{(1)} 1^{(2)})$,
    \item $r_M(\freeM^k n^{(1)}) = r_M(\freeM^k n^{(1)} 1^{(2)})$,
    \item $r_M(\freeM^k m^{(2)}) = r_M(\freeM^k 1^{(1)} m^{(2)})$.
\end{compactitem}
Finally, we define $r_M$ on $\calR$ by
\[ r_M\left(\sum_{s\in S} a_s s\right) = \sum_{s\in S} a_s r_M(s). \]

\begin{lemma}
    The rank function $r_M$ satisfies the following properties:
    \begin{compactenum}
        \item For any $a, b \in \calR$, we have $r_M(a+b) = r_M(a) + r_M(b)$. \hfill (Additive) \label{item:additive}
        \item For any $a, b \in \calR$, we have $r_M(ab) \leq r_M(a) r_M(b)$. \hfill (Submultiplicative) \label{item:submultiplicative}
        \item For any $a \leqslant_M b$, we have $r_M(a) \leq r_M(b)$. \hfill (Order-preserving) \label{item:order-preserving}
    \end{compactenum}
\end{lemma}
\begin{proof}
    Property \ref{item:additive} is immediate from the definition of $r_M$.

    For property \ref{item:submultiplicative}, we first consider the case when $a = \freeM^k n^{(1)} m^{(2)}$ and $b = \freeM^l n'^{(1)} m'^{(2)}$.
    Let $I$ be a decomposition of $M^{\otimes k}$ and $J$ be a decomposition of $M^{\otimes l}$
    achieving the minimum in $r_M(a)$ and $r_M(b)$, respectively. Then consider the decomposition
    of $M^{\otimes (k+l)}$ by taking the Kronecker product of the decompositions in $I$ and $J$, giving
    \begin{align*}
        r_M(ab) &\leq \sum_{i\in I, j\in J} nn'\nnz(U_i\otimes U_j) + mm'\nnz(V_i \otimes V_j)\\
        &= \sum_{i\in I, j\in J} nn'\nnz(U_i) \nnz(U_j) + mm'\nnz(V_i) \nnz(V_j)\\
        &\leq \sum_{i\in I, j\in J} (n \nnz(U_i) + m \nnz(V_i)) (n' \nnz(U_j) + m' \nnz(V_j))\\
        &= \left(\sum_{i\in I} n \nnz(U_i) + m \nnz(V_i)\right)
        \left(\sum_{j\in J} n' \nnz(U_j) + m' \nnz(V_j)\right)\\
        &= r_M(a) r_M(b).
    \end{align*}
    Therefore, we have $r_M(ab) \leq r_M(a) r_M(b)$ for all $a, b\in S$. The general case follows by linearity.

    For property \ref{item:order-preserving}, by lemma \ref{lem:one-step}, we only need
    to prove the statement for one step relations. Moreover, by linearity, we only need to prove
    it for the case
    $r_M(sx) \leq r_M(sy)$ where $s\in S$ and $x\leqslant y$ is from one of the basic rules in Definition \ref{def:semiring-preorder},
    or the case $r_M(0) \leq r_M(x)$. The latter is immediate, since $r_M$ is nonegative
    by our definition.
    \begin{compactitem}
        \item Consider the case when $x, y$ come from part \ref{item:lower-bound} of Definition~\ref{def:semiring-preorder}, say,
        $s = \freeM^k n^{(1)} m^{(2)}$ and $x = 1$, $y = \freeM$, $r^{(1)}$ or $r^{(2)}$.
        Suppose $y = \freeM$, when $M$ is not a zero matrix. For any decomposition of $M^{\otimes (k+1)}$, we can obtain a decomposition of $M^{\otimes k}$
        by removing some rows and columns, thus $r_M(sy) \leq r_M(sx)$.
        Suppose $y = r^{(1)}$, we have
        \begin{align*}
            r_M(sx) &= \min_I \left(\sum_{i\in I} n \nnz (U_i) + m \nnz(V_i)\right)\\
            &\leq \min_I \left(\sum_{i\in I} rn \nnz (U_i) + m \nnz(V_i)\right)\\
            &= r_M(sy).
        \end{align*}
        The case $y = r^{(2)}$ is similar.
        \item Consider the case when $x, y$ come from part \ref{item:decomposition} of Definition~\ref{def:semiring-preorder}, say,
        $s = \freeM^k n^{(1)}m^{(2)}$, $x = \freeM^l$, and $y = \sum_{i\in I} \nnz(U_i)^{(1)} \nnz(V_i)^{(2)}$
        for some decomposition $I$ of $M^{\otimes l}$. In that case, any decomposition of $M^{\otimes k}$
        combined with a decomposition of $M^{\otimes l}$ gives a decomposition of $M^{\otimes (k+l)}$,
        thus we have $r_M(sx) \leq r_M(sy)$.
        \item Consider the case when $x, y$ come from part \ref{item:cost-decomposition} of Definition~\ref{def:semiring-preorder}, say,
        $s = n^{(1)}m^{(1)}$, $x = n'^{(1)}m'^{(2)}$, and $y = n+m$. Then by submultiplicativity, we have
        $r_M(sx) \leq r_M(s)r_M(y) = (n+m)r_M(s) = r_M(sy)$.
    \end{compactitem}
    In conclusion, $r_M$ is order-preserving.
\end{proof}

\begin{lemma} \label{lem:strassen}
    The semiring preorder $\leqslant_M$ is a Strassen preorder.
    Moreover, for any $k\geq 1$, we have $\R(\freeM^k) = r_M(\freeM^k)$
    is the smallest size $\Size(M^{\otimes k})$ of a depth-2 circuit computing $M^{\otimes k}$.
\end{lemma}
\begin{proof}
    For any nonzero $a$, it's not hard to verify that $1\leqslant_M a \leqslant_M r_M(a)$,
    thus the strong Archimedean property holds.
    By definition, we have $r_M(n) = n$ for any $n\in \bbZ^+$, since $r_M$ is order preserving,
    the embedding of natural numbers holds. Thus $\leqslant_M$ is a Strassen preorder.

    By the definition of $r_M(\freeM^k)$ we have $\R(\freeM^k) \leq r_M(\freeM^k)$.
    Let $r = \R(\freeM^k)$. By definition we have $\freeM^k \leqslant_M r$.
    Since $r_M$ is order-preserving, we have $r_M(\freeM^k) \leq r_M(r)=r$, thus $r_M(r) \leq \R(\freeM^k)$.
    Therefore, we have $\R(\freeM^k) = r_M(\freeM^k)$.
    Finally from the definition of $r_M(\freeM^k)$, we have $r_M(\freeM^k) = \Size(M^{\otimes k})$.
\end{proof}

\subsection{The Duality Theorem}

We consider a relaxation of the definition of the size of a circuit $C$, which is defined as
\[ \rho_I(\alpha) = \sum_{i\in I} \nnz(U_i)^\alpha \nnz(V_i)^{1-\alpha}, \]
for $\alpha \in [0,1]$, where $U_i, V_i$ are the $i$th row and column of the circuit $C$.

One can similarly define $\Size_\alpha(M)$ to be the infimum of $\rho_I(\alpha)$ over all decompositions of $M$.
Since $\rho_{(\cdot)}(\alpha)$ is additive and multiplicative, we have $\Size_\alpha(\cdot)$ is subadditive and submultiplicative, which means
the limit
\[ \sigma_\alpha(M) = \lim_{n\to\infty} \Size_\alpha(M^{\otimes n})^{1/n} \]
enjoys the same properties as $\sigma(M)$.

Similarly, we can define $\sigma_{\OR,\alpha}$ and $\sigma_{\PAR,\alpha}$, for the corresponding $\OR$ and $\PAR$ complexities.

\begin{theorem} \label{thm:main}
    For any matrix $M$, we have $\sigma(M) = \sup_{\alpha \in [0, 1]} \sigma_\alpha(M)$.
    The statement also holds for $\OR$ and $\PAR$ complexities.
\end{theorem}

\begin{remark}
    In this section, we provide the proof for the case that the matrix is
    over a field. However, the argument in our proof is \emph{monotone} in the sense that
    it doesn't involve the subtraction operation, therefore, the proof also holds
    for covering or partitioning the entries of the matrix, i.e., statements about
    $\sigma_\OR$ and $\sigma_\PAR$. Intuitively, since our proof doesn't involve the minus sign,
    one can obtain the proof for $\sigma_\OR$ and $\sigma_\PAR$ by replacing the $+$ operation
    with $\cup$ or $\sqcup$, respectively.
\end{remark}

We first show that $\sigma_\alpha$ is a relaxation of $\sigma$.
\begin{lemma} \label{lem:relax}
    For any matrix $M$ and $\alpha\in [0, 1]$, we have $\sigma_\alpha(M) \leq \sigma(M)$.
    The statement also holds for $\OR$ and $\PAR$ complexities.
\end{lemma}
\begin{proof}
    For $\alpha = 0$, $\sigma_\alpha(M)$ is the infimum of
    $\left( \sum_{i\in I} \nnz(V_i) \right)^{1/n}$ over all decompositions. This quantity is clearly bounded by
    $\left( \sum_{i\in I} \nnz(U_i) + \sum_{i\in I} \nnz(V_i) \right)^{1/n}$,
    thus $\sigma_0(M) \leq \sigma(M)$. Similarly, we have $\sigma_1(M) \leq \sigma(M)$.

    Now consider $\alpha\in (0, 1)$. 
    Summing over $i\in I$, we have
    \begin{align*}
        \sum_{i\in I} \nnz(U_i)^\alpha \nnz(V_i)^{1-\alpha} &\leq \frac{1}{\alpha} \sum_{i\in I} \nnz(U_i) + \frac{1}{1-\alpha} \sum_{i\in I} \nnz(V_i)\\
        &\leq \max\left\{\frac{1}{\alpha}, \frac 1{1-\alpha} \right\} \left( \sum_{i\in I} \nnz(U_i) + \nnz(V_i) \right).
    \end{align*}
    Thus for a decomposition $M^{\otimes n} = \sum_{i\in I} U_i V_i^\sfT$, we have
    \[ \sigma_\alpha(M) \leq \max\left\{\frac{1}{\alpha}, \frac 1{1-\alpha} \right\}^{1/n}
    \left( \sum_{i\in I} \nnz(U_i) + \nnz(V_i) \right)^{1/n}. \]
    For any $\epsilon > 0$, we can choose arbitrarily large $n$ such that
    there exists a decomposition with $\left( \sum_{i\in I} \nnz(U_i) + \nnz(V_i) \right)^{1/n} \leq \sigma(M)+\epsilon$,
    thus $\sigma_\alpha(M) \leq \sigma(M) + \epsilon$. Since $\epsilon$ is arbitrary, we have $\sigma_\alpha(M) \leq \sigma(M)$.
\end{proof}

\begin{proof}[Proof of Theorem \ref{thm:main}]
    By Lemma \ref{lem:strassen}, the rank function of the Strassen preorder $\leqslant_M$
    characterizes the size of depth-2 circuits, i.e., $\R(\freeM^k) = \Size(M^{\otimes k})$.
    We thus have $\AR(\freeM) = \sigma(M)$.
    By Strassen duality (Theorem \ref{thm:strassen}), we have
    \[ \AR(\freeM) = \sup_{\phi \in \calX} \phi(\freeM). \]

    For any functional $\phi \in \calX$, let $\alpha, \beta$ be the numbers such that
    $\phi(n^{(1)}) = n^\alpha$ and $\phi(n^{(2)}) = n^\beta$.
    Since $\phi$ is order-preserving,
    we have $1 = \phi(1) \leq \phi(n^{(1)}) = n^\alpha$, which means $\alpha \geq 0$,
    and we similarly have $\beta \geq 0$.
    Also, since $n^{(1)}m^{(2)} \leqslant_M n+m$ for all positive integers $n$ and $m$, we must have
    \[ n^\alpha m^\beta = \phi(n^{(1)})\phi(m^{(2)}) = \phi(n^{(1)}m^{(2)}) \leq \phi(n+m) \leq n+m. \]
    This implies $\alpha + \beta \leq 1$.

    Finally, since
    \[ \freeM^k \leqslant \sum_{i\in I} \nnz(U_i)^{(1)} \nnz(V_i)^{(2)} \]
    for any decomposition $I$ of $M^{\otimes k}$,
    we have
    \begin{align*}
        \phi(\freeM)^k &= \phi(\freeM^k) \\
        &\leq \sum_{i\in I} \nnz(U_i)^{\alpha} \nnz(V_i)^{\beta}\\
        &\leq \sum_{i\in I} \nnz(U_i)^{\alpha} \nnz(V_i)^{1-\alpha}.
    \end{align*}
    Taking the infimum over all decompositions, we have $\phi(\freeM) \leq \sigma_\alpha(M)$.

    Finally, taking the supremum over all $\phi$, we have
    \[ \sigma(M) = \AR(\freeM) \leq \sup_{\alpha\in [0, 1]} \sigma_\alpha(M). \]
    Combined with Lemma \ref{lem:relax}, the theorem follows.
\end{proof}

\subsection{Discussion} \label{sec:discuss}

In this section, we proved a duality theorem for depth-2 circuits via Strassen's duality theory.
However, the proof of Strassen's duality invokes Zorn's lemma (which is equivalent to the Axiom of Choice)
and is a proof by contradiction, so one cannot directly extract an
effective bound on the circuit size from the proof. In particular, when this is applied to a specific matrix $M$,
one could deduce the a bound $a^{n+o(n)}$ on the size of the depth-2 circuits for $M^{\otimes n}$,
but one cannot write down an explicit expression for the $o(n)$ in the exponent,
and it is not clear how to construct a family of circuits achieving the asymptotic bound
(although, in principle, a brute-force search algorithm can eventually find the smallest depth-2 circuit).

This situation is reminiscent of Szemerédi's theorem \cite{szemeredi1975sets} in additive combinatorics. Let $r_k(N)$
denote the size of the largest subset of $[N]$ that avoids $k$-term arithmetic progressions.
Szemerédi's theorem asserts that $r_k(N) = o_k(N)$ for any $k \ge 3$.
Furstenberg~\cite{furstenberg1977ergodic} later provided an alternative proof using
ergodic theory---a method that has since become a powerful tool in tackling other combinatorial problems.
However, his proof is nonconstructive and does not yield an explicit bound on $r_k(N)$. Subsequent
work~\cite{tao2004quantitative} has shown that modifying the proof to obtain a constructive version
and extract explicit bounds requires significant additional effort.

For this reason, we also provide an alternative proof of the duality theorem
for depth-2 circuits by directly analyzing the rebalancing process
and arguing that it achieves the desired bound, in Appendix \ref{sec:tree} below.
The proof provided there is constructive and elementary, and thus might be of independent interest.

Finally, we emphasize that Theorem \ref{thm:main} in this section establishes a connection between the
asymptotic size of depth-2 circuits and the $\alpha$-volume of \emph{all} possible decompositions of
the matrix. While this suffices for our primary objective of improving upper bounds on circuit size,
it still does not exactly answer one of our initial questions: If we restrict to a specific family of decompositions, what is the
``best'' way to construct a circuit for $M^{\otimes n}$ via the rebalancing process?  

In principle, this refinement can also be addressed using the theory of asymptotic spectra. To do so,
one modifies the definition of the preorder $\leqslant_M$ by restricting the comparison relations to
the chosen family of decompositions. Under this adaptation, the rank function $ r_M $ would then
characterize the circuit size achievable through an optimal rebalancing strategy.  

In fact, our alternative proof in Appendix \ref{sec:tree} adopts this perspective. There, we prove a
duality theorem applicable to any \emph{finite family} of decompositions (indeed, the meaning of constructiveness becomes
unclear when handling infinitely many decompositions). Finally, by employing a standard
topological argument to extend the result to the space of all decompositions, we can also recover
Theorem~\ref{thm:main} from that constructive proof.

\section{Improving Circuit Size} \label{sec:improve-circuit}

We need the following lemma to simplify our construction to bound the disjointness matrix.
\begin{lemma}
    For a matrix $M$, we have $\sigma_\alpha(M) = \sigma_{1-\alpha}(M^\sfT)$.
    The statement also holds for $\OR$ and $\PAR$ complexities.
\end{lemma}
\begin{proof}
    For any decomposition
    \[ M^{\otimes n} = \sum_{i\in I} U_i V_i^\sfT, \]
    taking the transpose of both sides, we have
    \[ (M^\sfT)^{\otimes n} = \sum_{i\in I} V_i U_i^\sfT. \]
    Therefore, for any decomposition, the quantity
    \[ \left( \sum_{i\in I} \nnz(U_i)^\alpha \nnz(V_i)^{1-\alpha} \right)^{1/n} \]
    bounds both $\sigma_\alpha(M)$ and $\sigma_{1-\alpha}(M^\sfT)$. Taking the infimum over all decompositions,
    by the definition of $\sigma_\alpha$, gives as desired that $\sigma_\alpha(M) = \sigma_{1-\alpha}(M^\sfT)$.
\end{proof}

Since the disjointness matrix $R$ is symmetric, we only need to focus on bounding $\sigma_\alpha(R)$ for $\alpha \geq 1/2$.

\subsection{Generalizing Sergeev's Construction}  \label{sec:generalize-sergeev}

In this subsection, we consider a strategy to decompose $R^{\otimes d}$ into
rectangles of shape $a\times 1$ or $1\times b$.
We identify the rows and columns of $R^{\otimes d}$ with subsets of $[d]$, so that
$R^{\otimes d}_{S, T} = \iv{S \cap T = \emptyset}$ for $S, T \subseteq [d]$.

For each string $w \in \{ \mathtt{R}, \mathtt{C} \}^{d+1}$, we associate a decomposition $I_w$ of $R^{\otimes d}$ given
by procedure \ref{fig:proc-simple}.

\begin{figure}[htbp]
    \renewcommand{\figurename}{Procedure}

    \centering
    \fbox{\parbox{0.8\textwidth}{
    \textbf{Input:} Positive integer $d$, string $w \in \{ \mathtt{R}, \mathtt{C} \}^{d+1}$.

    \textbf{Output:} Set $I_w$ of rectangles which decompose the matrix $R^{\otimes d}$.

    \textbf{Initialization:} Maintain two integers $r, c$ initially set to $0$, and a set $I_w$ of rectangles, initially empty.

    \textbf{For} $i = 1$ to $d+1$:

    \quad Let $r_i = r$, $c_i = c$.

    \quad \textbf{If} $w_i = \mathtt{R}$ (row):

    \qquad \textbf{For} every $S \subseteq [d]$ of size $|S| = r_i$, construct a rectangle with row $S$, and
    columns $\{ T \subseteq [d] \setminus S : |T| \geq c_i \}$, and add it to $I_w$.
    
    \qquad Increase $r$ by $1$.

    \quad \textbf{Otherwise}, i.e., if $w_i = \mathtt{C}$ (column):

    \qquad \textbf{For} every $T \subseteq [d]$ of size $|T| = c_i$, construct a rectangle with rows
    $\{ S \subseteq [d] \setminus T : |S| \geq r_i \}$, and column $T$, and add it to $I_w$.

    \qquad Increase $c$ by $1$.

    \textbf{Return} $I_w$.
    }}

    \caption{Take one row or column at a time} \label{fig:proc-simple}
\end{figure}

\begin{proposition}
    For any string $w \in \{ \mathtt{R}, \mathtt{C} \}^{d+1}$, procedure \ref{fig:proc-simple}
    constructs a decomposition $I_w$ of $R^{\otimes d}$.
\end{proposition}
\begin{proof}
    Each rectangle which the procedure adds to the output set $I_w$ is either a row or column of $R^{\otimes d}$, restricted to rows and columns which have not already been added to $T_w$ in an earlier step. In iteration step $i$, the procedure either takes all rows of weight $r_i$ (then increments $r$)
    or all columns of weight $c_i$ (then increments $c$). So after all $d+1$ iterations, the $1$s of $R^{\otimes d}$ which are not in any rectangle in $I_w$ must be in rows $S$ of weight $\geq r$
    and columns $T$ of weight $\geq c$. However, since $r+c=d+1$, it follows that any $S,T$ with $|S|\geq r$ and $|T|\geq c$ must intersect. Our decomposition is therefore not missing any entries.
\end{proof}

All the previous constructions can be viewed as special cases of the above procedure.
\begin{compactitem}
    \item The ``one-sided decompositions'' (all rows or all columns) corresponds to the string
    $w = \mathtt{RRR...}$ or $w = \mathtt{CCC...}$.
    \item The construction of \cite[Sec.~7.1]{AGP23KroneckerCircuits} corresponds to the string
    $w = \mathtt{CRRR}$.
    \item The original construction of Sergeev \cite{Sergeev22KroneckerCoverings}
    corresponds to the string $w = \mathtt{CRCR...}$.
\end{compactitem}

\begin{proposition} \label{prop:simple}
    Let $(r_1,\dots,r_{d+1})$ and $(c_1,\dots,c_{d+1})$ be the sequences of $r_i$ and $c_i$ in the construction of $I_w$.
    The $\alpha$-volume of decomposition $I_w$ is
    \[ \rho_w(\alpha) = \sum_{w_i = \mathtt{R}} \binom{d}{r_i}
    \binom{d-r_i}{\geq c_i}^{1-\alpha} + \sum_{w_i = \mathtt{C}} 
    \binom{d}{c_i} \binom{d-c_i}{\geq r_i}^\alpha. \]
\end{proposition}
\begin{proof}
    For each $\mathtt{R}$, we take all rows of weight $|S| = r_i$ and associate each of them with a rectangle
    $(U, V)$, where $U$ is the indicator vector of $S$, and $V$ is the indicator vector of $\{ T : S\cap T = \emptyset, |T| \geq c_i \}$.
    thus we have $\nnz(U) = \binom{d}{r_i}$ and $\nnz(V) = \binom{d-r_i}{\geq c_i}$.
    There are $\binom{d}{r_i}$ such rows in total, so the contribution of the case $w_i = \mathtt{R}$ to the $\alpha$-volume is
    \[ \binom{d}{r_i} \nnz(U)^\alpha \nnz(V)^{1-\alpha} = \binom{d}{r_i} \binom{d-r_i}{\geq c_i}^{1-\alpha}. \]
    The contribution of the case $w_i = \mathtt{C}$ can be similarly computed.
\end{proof}

We first show that the simple decompositions produced by procedure \ref{fig:proc-simple} 
already give an improved upper bound on the circuit size of the disjointness matrix.
Consider $d=18$ and three strings
\begin{align*}
    w_1 &= \mathtt{RCRCRCRCRCRCRCRCRCR},\\
    w_2 &= \mathtt{RRCRCCRCRCRCRRCRCRR},\\
    w_3 &= \mathtt{RRRRRRRRRRRRRRRRRRR}.
\end{align*}

Substituting these into the formula of proposition \ref{prop:simple} gives an upper bound of $\log_2 \sigma(R)$.
The results are plotted in figure \ref{fig:simple}.

\begin{figure}[htbp]
    \centering

            \includegraphics{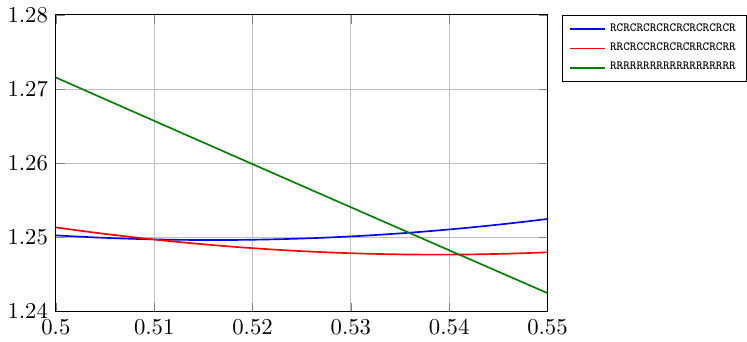}
    \caption{The upper bound of $\log_2 \sigma_\alpha(R)$ given by three simple decompositions}
    \label{fig:simple}
\end{figure}

It can be verified that for the lower envelope of the three curves, the maximal value is achieved at $\alpha = 1/2$.
Therefore, we have an improved upper bound of $\sigma(R)$ by the following corollary.

\begin{corollary}
    The disjointness matrix has
    \[ \sigma_\PAR(R) \leq \rho_{w_1}^{1/18}(1/2) < 2^{1.25026}. \]
\end{corollary}

\begin{remark}
    Numerical verification shows that the minimum of $\rho_{\mathtt{RCRC\dots}}^{1/d}(1/2)$ occurs
    at $d = 18$, a fact initially observed by Sergeev \cite{Sergeev22KroneckerCoverings}. However, his
    approach does not establish that $\rho_{w_1}^{1/18}(1/2)$ can serve as an upper bound
    on $\sigma(R)$. He obtained $\sigma(R) < 2^{1.251}$ through $d=15$.

    With our duality theorem, we now see that this limitation is not due to a missing technical
    lemma in his work. In his construction, only $w_1$ and $w_3$ of the three specified strings are
    considered. However, the maximum of $\min(\rho_{w_1}(\alpha), \rho_{w_3}(\alpha))$ is
    not achieved at $\alpha = 1/2$, making it impossible to bound $\sigma(R) \leq \rho_{w_1}^{1/18}(1/2)$
    solely with these two constructions. The inclusion of the decomposition provided by $w_2$ is essential to attain the desired bound.
\end{remark}

\subsection{Construction with Merging} \label{sec:complex}

In the construction of previous subsection, every rectangle is of shape $a\times 1$ or $1\times b$.
In this subsection, we consider a more complicated construction, where we allow merging of rows or columns.
Intuitively, for a set $S$, we group the rectangle with $S$ being a single row
together with those rectangles with $S \setminus \{x\}$ being a single row, where $x$ iterates over
all the elements of $S$. We merge their common columns into a single rectangle of shape $(|S|+1)\times a$,
and the rest of columns into $|S|$ rectangles of shape $1\times b$. As we will see, in some cases,
this improves the $\alpha$-volume of the decomposition.

For a string $w \in \{ \mathtt{R}, \mathtt{C}, \mathtt{R^+}, \mathtt{C^+} \}^*$, we say its
weight is the length of the string $|w|$ plus the number of $\mathtt{R^+}$ and $\mathtt{C^+}$ in the string.
For any string $w$ of weight $d+1$, associated with an array $\{a_i\}_i$ defined over those indices $i$ such that $w_i\in \{\mathtt{R^+}, \mathtt{C^+}\}$,
we attempt to construct a decomposition $I_{w,a}$ of $R^{\otimes d}$ using procedure \ref{fig:proc-complex} (which, on some inputs, may fail).

\vspace{1cm}
\begin{figure}[htbp]
    \renewcommand{\figurename}{Procedure}

    \centering
    \fbox{\parbox{0.8\textwidth}{
    \textbf{Input:} Positive integer $d$, string $w \in \{ \mathtt{R}, \mathtt{C}, \mathtt{R^+}, \mathtt{C^+} \}^*$ of weight $d+1$, and array $\{a_i\}_i$ of nonnegative integers defined over those indices $i$ such that $w_i\in \{\mathtt{R^+}, \mathtt{C^+}\}$.

    \textbf{Output:} Set $I_{w,a}$ of rectangles which decompose the matrix $R^{\otimes d}$.

    \textbf{Initialization:} Maintain two integers $r, c$ initially set to $0$, and a set $I_{w,a}$ of rectangles, initially empty.

    \textbf{For} $i = 1$ to $|w|$:

    \quad Let $r_i = r$, $c_i = c$.

    \quad \textbf{If} $w_i = \mathtt{R}$ (row):

    \qquad \textbf{For} every $S \subset [d]$ of size $|S| = r_i$, construct a rectangle with row $S$, and
    columns $\{ T \subset [d] \setminus S : |T| \geq c_t \}$, and add it to $I_{w,a}$.
    
    \qquad Increase $r$ by $1$.

    \quad \textbf{If} $w_i = \mathtt{R^+}$ (row merge):

    \qquad Find $a_i$ sets $S_1,\dots,S_{a_i}$ of size $r_i + 1$, such that $|S_j \cap S_k| < r_i$ for all $j\neq k$.
    If such sets do not exist, \textbf{report failure}.
    
    \qquad \textbf{For} each of the sets $S_j$, construct a rectangle with rows $\{ S \subseteq S_j : |S|\geq r_i \}$,
    and columns $\{ T \subset [d] \setminus S_j : |T| \geq c_i \}$, and add it to $I_{w,a}$.

    \qquad \textbf{For} every $S \subseteq [d]$ which either has size $|S| = r_i+1$, or which has size $|S| = r_i$ and which is not a subset of any $S_j$,
    construct a rectangle with row $S$ and columns $\{ T \subset [d] \setminus S : |T| \geq c_i \}$, and add it to $I_{w,a}$.

    \qquad \textbf{For} every $S \subseteq [d]$ which has size $|S| = r_i$ and which is a subset of some $S_j$, construct a rectangle with row $S$,
    and columns $\{ T \subset [d] \setminus S : T \cap S_j \neq \emptyset, |T| \geq c_i \}$.

    \qquad Increase $r$ by $2$.

    \quad \textbf{If} $w_i = \mathtt{C}$ or $\mathtt{C^+}$, do the symmetric operation for columns.

    \textbf{Return} $I_{w,a}$.
    }}

    \caption{Partition with merging} \label{fig:proc-complex}
\end{figure}
\vspace{1cm}

\begin{proposition} \label{prop:complex}
    Let $(r_1,\dots,r_{|w|})$ and $(c_1,\dots,c_{|w|})$ be the sequences of $r_i$ and $c_i$ in the construction of $I_{w,a}$.
    The $\alpha$-volume of decomposition $I_{w,a}$ is
    \begin{align*}
        \rho_{w,a}(\alpha) &= \sum_{w_i = \mathtt{R}} \binom{d}{r_i} \binom{d-r_i}{\geq c_i}^{1-\alpha}
        + \sum_{w_i = \mathtt{C}} \binom{d}{c_i} \binom{d-c_i}{\geq r_i}^\alpha\\
        &+ \sum_{w_i = \mathtt{R^+}} a_i 
        \left[ (r_i+2)^\alpha \binom{d-r_i-1}{\geq c_i}^{1-\alpha} + 
        (r_i+1)\binom{d-r_i-1}{\geq c_i-1}^{1-\alpha} \right]\\
        &+ \phantom{\sum_{w_i = \mathtt{R^+}}} \left[\binom{d}{r_i} - (r_i+1) a_i\right] \binom{d-r_i}{\geq c_i}^{1-\alpha}
        + \left[\binom{d}{r_i+1} - a_i\right] \binom{d-r_i-1}{\geq c_i}^{1-\alpha}\\
        &+ \sum_{w_i = \mathtt{C^+}} a_i \left[
            (c_i+2)^{1-\alpha} \binom{d-c_i-1}{\geq r_i}^\alpha +
            (c_i+1)\binom{d-c_i-1}{\geq r_i-1}^\alpha
        \right]\\
        &+ \phantom{\sum_{w_i = \mathtt{C^+}}} \left[\binom{d}{c_i} - (c_i+1) a_i\right] \binom{d-c_i}{\geq r_i}^\alpha
        + \left[\binom{d}{c_i+1} - a_i\right] \binom{d-c_i-1}{\geq r_i}^\alpha.
    \end{align*}
\end{proposition}
\begin{proof}
    The contribution of the case $w_i = \mathtt{R}$ and $w_i = \mathtt{C}$ is the same as in proposition \ref{prop:simple}.

    For the case $w_i = \mathtt{R^+}$, we have $a_i$ sets $S_1,\dots,S_{a_i}$ of size $r_i+1$. For each $S_j$, the procedure constructs a rectangle with rows $\{ S \subset S_j : |S|\geq r \}$,
    and columns $\{ T \subset [d] \setminus S_j : |T| \geq c \}$. This rectangle is of shape
    $(r_i+2) \times \binom{d - r_i-1}{\geq c_i}$.
    Thus, their total contribution is $a_i (r_i+2)^\alpha \binom{d-r_i-1}{\geq c_i}^{1-\alpha}$.

    Since the procedure requires that $|S_j \cap S_k| < r_i$ for all $j\neq k$, it follows that the families of subsets of each $S_i$ of size $r$
    must be disjoint. For each such subset $S$, the sets
    $\{T : T \cap S_j = \emptyset, |T|\geq c_i \}$ are already covered by previous rectangles, and
    the rest of the columns are $\{ T \subset [d] : T \cap S_j = S_j\setminus S, |T| \geq c_i \}$. 
    Since $T \cap S_j$ has exactly one element, the rectangle has shape $1\times \binom{d-r_i-1}{\geq c_i-1}$.
    Since each $S_j$ has $r+1$ subsets, their total contribution is $a_i (r_i+1) \binom{d-r_i-1}{\geq c_i-1}^{1-\alpha}$.

    For the remaining rows, the number of weight-$r_i$ subsets is $\binom{d}{r_i} - (r_i+1) a_i$,
    and the number of weight-$(r_i+1)$ subsets is $\binom{d}{r_i+1} - a_i$. The contribution of these rows is
    $\left[\binom{d}{r_i} - (r_i+1) a_i\right] \binom{d-r_i}{\geq c_i}^{1-\alpha}$ and
    $\left[\binom{d}{r_i+1} - a_i\right] \binom{d-r_i-1}{\geq c_i}^{1-\alpha}$, respectively.

    This completes the proof for the case $w_i = \mathtt{R^+}$. Finally, the case $w_i = \mathtt{C^+}$ is identical by symmetry.
\end{proof}

In general, procedure \ref{fig:proc-complex} will often fail. However, we identify a particular input of interest where it succeeds, which we will use to give our new circuit construction:

\begin{lemma} \label{lem:complex}
    For $d = 18$, procedure \ref{fig:proc-complex} can successfully construct a decomposition $I_{w,a}$ (without reporting failure) for
    \[ w = \mathtt{RCR^+C^+R^+C^+R^+C^+RCRCR} \]
    with
    \[ a = (-, -, 9, 9, 198, 198, 1260, 1260, -, -, -, -, -). \]
\end{lemma}
\begin{proof}
    We will reduce the packing problem (of finding the sets $S_1, \ldots, S_{a_i}$ needed in procedure \ref{fig:proc-complex}) to the known construction of constant weight binary codes.

    Let $2\leq s \leq d$ be integers. A \emph{binary code} $C$ of weight $s$ is a subset of $\{0, 1\}^d$ with distance $\geq \delta$ such that
    for any $S, T$ in $C$, the symmetric diffrence $|S\Delta T| \geq \delta$.
    It's not hard to see that when $s = r+1$, the existence of sets $S_1,\dots,S_{a}$ satisfying $|S_j \cap S_k| < r$ is equivalent to
    the existence of a binary code of weight $r+1$, distance $4$ having $a$ codewords.

    Let $a = A(d, \delta, s)$ be the maximum size of a binary code of weight $s$ and distance $\geq \delta$.
    It's not hard to see that $A(18, 4, 2) = 9$.
    For larger weight, it is known that
    \[ A(18, 4, 4) = 198, \quad A(18, 4, 6) \geq 1260, \]
    see \cite[Table I-A]{CodeTable06} as a reference.
    Thus the existence of the sets in the construction is guaranteed.
\end{proof}

In fact, $I_{w,a}$ combined with the simple one-sided decomposition already gives an improved upper bound on $\sigma(R)$.

\begin{figure}[htbp]
    \centering

                                                                            \includegraphics{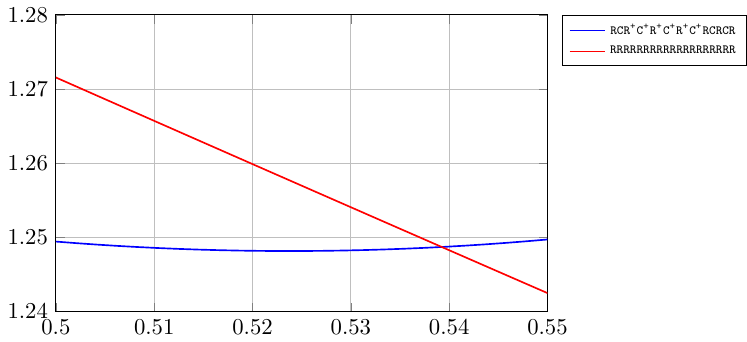}
    \caption{The upper bound of $\log_2 \sigma_\alpha(R)$ given by the decomposition of lemma~\ref{lem:complex}
    and the one-sided decomposition} \label{fig:merge}
\end{figure}

Plugging this into the formula of proposition \ref{prop:complex}, one can numerically verify that
$I_{w,a}$ along with the simple one-sided decomposition gives an upper bound of $\log_2 \sigma(R)$,
see figure \ref{fig:merge}.
\begin{theorem} \label{thm:disjointness_par}
    The disjointness matrix has
    \[ \sigma_\PAR(R) \leq \rho_{w,a}^{1/18}(1/2) < 2^{1.249424}. \]
\end{theorem}

\begin{corollary} \label{cor:depthd}
    For any even $d$, the $N\times N$ disjointness matrix $R_n$ has a depth-$d$ linear circuit of size $O(N^{1 + 0.498848/d})$.
\end{corollary}
\begin{proof}
    We proceed similar to \cite[Theorem 3.4]{alman2021kronecker} to generalize from depth $2$ to higher depths. For simplicity, we consider the case that $d$ divides $n$. Theorem \ref{thm:disjointness_par} tells us that
    $R_{2n/d}$ has a depth-2 circuit $U \times V^\sfT$ of size
    $O(2^{1.249424 \cdot (2n/d)}) = O(N^{2/d + 0.498848/d})$.
    Therefore, $R_n$ has a depth-$d$ circuit given by
    \begin{align*}
        R_n &= \overbrace{R_{2n/d} \otimes \cdots \otimes R_{2n/d}}^{d/2 \text{ times}}\\
        &= \phantom{\otimes} (\overbrace{U \times V^\sfT \times I \times I \times \cdots \times I}^{d \text{ times}})\\
        & \quad  \otimes \cdots\\
        & \quad  \otimes (I \times \cdots \times I \times I \times U \times V^\sfT)\\
        &= (U\otimes I_{2^{(d-2)n}}) \times (V^\sfT \otimes I_{2^{(d-2)n}})
        \times \cdots \times 
        (I_{2^{(d-2)n}} \otimes U) \times (I_{2^{(d-2)n}} \otimes V^\sfT).
    \end{align*}
    The size of the circuit is $O(N^{2/d + 0.498848/d}) \cdot 2^{(d-2)n} = O(N^{1+0.498848/d})$.
\end{proof}

We note that we did not fully optimize the achievable upper bound with the complex construction.
Nonetheless, our result notably breaks the $1.25$ barrier for the first time, which has further
implications for low-depth circuits. By extending the approach of \cite[Theorem 6.6]{alman2021kronecker},
we obtain the following corollary.
\begin{corollary} \label{cor:or-circulant}
    For any positive integer $n$, any field $\mathbb{F}$, and any function $f \colon \{0,1\}^n \to \mathbb{F}$, let $N = 2^n$, and consider the matrix $M_f \in \mathbb{F}^{N \times N}$ which is given by, for $x,y \in \{0,1\}^n$, $$M_f[x,y] = f(x_1 \vee y_1, x_2 \vee y_2, \ldots, x_n \vee y_n).$$ Then, $M_f$ has a depth-4 circuit of size $O(N^{1.249424})$, and more generally, for any even $d$, a depth-$2d$ linear circuit of size $O(N^{1 + 0.498848/d})$.
\end{corollary}

\begin{proof}
    We write $x\vee y = (x_1 \vee y_1, x_2 \vee y_2, \ldots, x_n \vee y_n)$.
    Consider the M\"obius transform $g\colon \{0,1\}^n\to\bbF$ of $f$:
    \[
        g(z) = \sum_{z\geq x} (-1)^{\| z-x\|} f(x),
    \]
    where $z\geq x$ means $z_i \geq x_i$ for all $i$. By construction, this satisfies
    \[ f(x) = \sum_{z\geq x} g(z). \]
    Therefore, the matrix $M_f$ can be expressed as
    \begin{align*}
        M_f[x, y] &= f(x\vee y)\\
        &= \sum_{z\geq x\vee y} g(z)\\
        &= \sum_{z\in \{0, 1\}^n} \iv{x\leq z} g(z) \iv{z\geq y}.
    \end{align*}
    This shows that the matrix $M_f$ can be factorized to the form $M_f
    = \iv{x\leq z}_{x,z} \cdot \Lambda_g \cdot \iv{z\geq y}_{z, y}$, where $\Lambda_g$
    is a diagonal matrix with $g(z)$ on the diagonal. Also note that the comparison in
    $\iv{x\leq z}_{x,z}$ is component-wise, so it can be written as the Kronecker power
    \[ \iv{x\leq z}_{x,z} = \begin{bmatrix}
        1 & 1 \\ 0 & 1
    \end{bmatrix}^{\otimes n}, \]
    and $\iv{z\geq y}_{z,y}$ is the transpose of $\iv{x\leq z}_{x,z}$; both of these are permutations
    of the disjointness matrix $R_n$. Now consider using
    the $N^{1+0.498848/d}$ sized depth-$d$ circuit described in Corollary \ref{cor:depthd}
    to compute $R_n$. The diagonal matrix $\Lambda_g$ can be merged into an adjacent layer,
    giving a depth-$2d$ circuit of size $O(N^{1+0.498848/d})$ which computes $M_f$.
\end{proof}

Note that the Kronecker product of $2\times 2$ matrices $M_1\otimes \cdots \otimes M_n$ is basically a scaling of some $M_f$ matrix.
Therefore, the above result improves the classical bound of $O(N^{1+1/2d})$ via the mixed product property of the Kronecker product.
Moreover, when all the matrices $M_i$ are the Walsh--Hadamard matrix $H$, this yields a depth-$2d$
circuit that outperforms the classical bound of $O(2^{1.25n})$, with only one layer having
unbounded coefficients (the one $\Lambda_g$ was merged into). Known constructions derived from matrix non-rigidity decompositions have multiple layers with unbounded coefficients.

Finally, our construction ultimately gives a partitioning of the disjointness matrix, and we note that there is a known lower bound on the size of a partitioning of the matrix by $\sigma_\PAR(R) \geq \sigma_\OR(R) = \sqrt{5} \geq 2^{1.16}$
\cite[Lemma 4.2]{jukna2013complexity}. Therefore, achieving a depth-2 circuit of size $N^{1+o(1)}$ for the disjointness matrix necessitates a new approach that effectively utilizes
the power of algebraic cancellation.

\section{Fixing Holes} \label{sec:fix-holes}

We adapt the hole fixing technique in the recent development of fast matrix multiplication algorithms
\cite[Lemma 5.6]{DWZ23AsymmetricHashing}, \cite[Theorem 7.2]{WXXZ24MatrixMultiplication} to the setting of depth-2 circuits.

For an $n\times m$ matrix $M$ and an $\epsilon \in (0,1)$, we say a broken copy of $M$ with $\epsilon$-holes is a matrix $M'$ such that
there exist subsets $S\subset [n]$, $T\subset [m]$ with $|S| \leq \epsilon n$, $|T| \leq \epsilon m$,
such that $M'_{ij} = M_{ij}$ for all $i\notin S$, $j\notin T$, and $M'_{ij} = 0$ otherwise.

For any invertible matrix $A\in \GL_n(\bbF)$, we say $A$ is a \emph{permutation with coefficients} if
there exists a permutation $\sigma \in \Sym_n$ and nonzero coefficients $a_1,\dots,a_n\in \bbF^\times$ such that
$A$ maps the $i$th standard basis vector $e_i$ to $Ae_i = a_i e_{\sigma(i)}$ for all $i\in [n]$.

Let $G$ be a subgroup $G \subset \GL_n(\bbF) \times \GL_m(\bbF)$. We write each element $g\in G$ as
$g = (A_g, B_g)$, where $A_g \in \GL_n(\bbF)$ and $B_g \in \GL_m(\bbF)$. If $A_g, B_g$
are permutations with coefficients for all $(A_g, B_g) \in G$, then we say $G$ is a group of permutations with coefficients.

Let $G\subset \GL_n(\bbF)\times \GL_m(\bbF)$ be a group of permutations with coefficients. Consider the following two group actions:
\begin{compactenum}
    \item[(L)] Let $X = [n]$ and $g = (A_g, B_g) \in G$. Then $g$ acts on $X$ by $g\cdot i = j$,
    where $j$ is the unique index such that $A_g e_i = a_i e_j$. \label{item:group-action-L}
    \item[(R)] Let $Y = [m]$ and $g = (A_g, B_g) \in G$. Then $g$ acts on $Y$ by $g\cdot i = j$,
    where $j$ is the unique index such that $B_g e_i = b_i e_j$. \label{item:group-action-R}
\end{compactenum}
Such group actions give the orbit decompositions of $X = X_1 \sqcup \cdots \sqcup X_k$
and $Y = Y_1 \sqcup \cdots \sqcup Y_\ell$.
We say $G$ is \emph{bitransitive} if $G$ acts transitively on both $X$ and $Y$, i.e., $k=\ell=1$.

We aim to prove the following theorem about fixing holes.

\begin{theorem} \label{thm:fix-holes}
    Let $M$ be a matrix of size $n\times m$ over $\bbF$, and $G\subset \GL_n(\bbF)\times \GL_m(\bbF)$ be a group of permutations with coefficients
    which is bitransitive, and the group action $G$ on $n\times m$ matrices
    \[ (A, B) \cdot N = A N B^\sfT \]
    fixes the matrix $M$.

    Then suppose a depth-2 circuit $C$ computes a broken copy $M'$ of $M$ with $\epsilon$-holes,
    with $\epsilon \in (0,1/4]$, degree at most $d$ and size at most $s$. Then there exists a depth-2 circuit $C'$ computing $M$
    with degree at most $(nm)^{O(1/\log (1/\epsilon))}d$ and size at most $(nm)^{O(1/\log (1/\epsilon))}s$.
\end{theorem}

We first prove the following helper lemma showing how we will make use of the group $G$.

\begin{lemma} \label{lem:fix-holes}
    Suppose matrix $M$ and group $G$ satisfy the conditions of Theorem \ref{thm:fix-holes}
    with a broken copy $M'$ of $M$ with $\epsilon$-holes.
    For any subsets $S\subset [n]$, $T\subset [m]$, there exist subsets $S'\subset S$, $T'\subset T$ such that
    \begin{compactitem}
        \item The subsets $S'$ and $T'$ satisfy $|S'| \geq (1-3\epsilon)|S|$
        and $|T'| \geq (1-3\epsilon)|T|$, and
        \item The submatrix $M|_{S',T'}$ of $M$ is also a submatrix of $g\cdot M'$ for some $g\in G$.
    \end{compactitem}
\end{lemma}

\begin{proof}
    Let $S_0, T_0$ be the set of rows and columns that are the holes in $M'$, so we have $|S_0| \leq \epsilon n$ and $|T_0| \leq \epsilon m$.
    Now consider a group element $g = (A_g, B_g) \in G$ taken uniformly random over $G$.
    For each $s \in [n]$, by the orbit-stabilizer theorem of group actions, we have
    $g\cdot s$ is uniformly distributed over $[n]$. Now, letting $g = (A_g, B_g)$ act on $S_0$ by having $A_g$ act
    on each element of $S_0$, we have
    \[ \Ex_{g\in G}[ |(g\cdot S_0) \cap S| ] = \frac{|S_0| \cdot |S|}{n}
    \leq \epsilon |S|, \]
    and by Markov's inequality, we have
    \[ \Pr_{g\in G}[ |(g\cdot S_0) \cap S| \geq 3\epsilon |S| ]
    \leq \frac{\epsilon |S|}{3\epsilon |S|} = \frac 13. \]
    Similarly, if $g = (A_g, B_g)$ acts on $T_0$ by having $B_g$ act
    on each element of $T_0$, then the inequality holds for $T_0$ and $T$:
    \[ \Pr_{g\in G}[ |(g\cdot T_0) \cap T| \geq 3\epsilon |T|] \leq \frac 1 3. \]
    Thus by a union bound, there exists a $g\in G$ such that both inequalities are violated, i.e.,
    \[ |(g\cdot S_0) \cap S| < 3\epsilon |S|, \quad |(g\cdot T_0) \cap T| < 3\epsilon |T|. \]
    
    We take $S' = S \setminus (g\cdot S_0)$ and $T' = T \setminus (g\cdot T_0)$.
    For any $s \in S'$ and $t\in T'$, we have $s\notin g\cdot S_0$ and $t\notin g\cdot T_0$, and so
    \begin{align*}
        (g\cdot M')_{s, t} &= (A_g M' B_g^\sfT)_{s, t}\\
        &= A_{s, g^{-1}\cdot s} \cdot B_{t, g^{-1}\cdot t} \cdot M'_{g^{-1}\cdot s, g^{-1}\cdot t}\\
        &= A_{s, g^{-1}\cdot s} \cdot B_{t, g^{-1}\cdot t} \cdot M_{g^{-1}\cdot s, g^{-1}\cdot t}\\
        &= (g\cdot M)_{s, t}\\
        &= M_{s, t}.
    \end{align*}
    The last equality holds because $G$ fixes the matrix $M$. In conclusion, we have
    $(g\cdot M')|_{S', T'} = M|_{S', T'}$. 
\end{proof}

\begin{proof}[Proof of Theorem \ref{thm:fix-holes}]
    We call the matrices of the form $g\cdot M'$ a $G$-permutation of $M'$, where $g\in G$.
    For any $1\leq x\leq n$ and $1\leq y\leq m$, let $F_{x, y}$ be the least number such that
    every $x\times y$ submatrix of $M$ is the sum of $F_{x, y}$ submatrices of $G$-permutations of $M'$.
    We inductively prove that
    \[ F_{x, y} \leq 2^{\frac{\log (xy)}{\log(1/3\epsilon)} + 1} - 1. \]
    The base case $F_{1, 1} = 1$ is trivial.
    For each $|S| = x$ and $|T| = y$, consider the submatrix $M|_{S, T}$. By Lemma \ref{lem:fix-holes},
    there exist decompositions $S = S_0 \sqcup S'$
    and $T = T_0 \sqcup T'$ such that $|S'| \geq (1-3\epsilon)x$ and $|T'| \geq (1-3\epsilon)y$,
    and $M|_{S', T'}$ is a submatrix of $g\cdot M'$ for some $g\in G$.
    Note that we can decompose $M|_{S, T}$ into the sum of
    \[ M|_{S', T'}, M|_{S', T_0}, M|_{S_0, T}, \]
    where the latter two submatrices are of size $\leq x \times 3\epsilon y$ and $\leq 3\epsilon x \times y$, respectively.
    By the inductive hypothesis, we have $M|_{S, T}$ is the sum of at most
    $1 + F_{x, \lfloor 3\epsilon y\rfloor} + F_{\lfloor 3\epsilon x\rfloor, y}$ submatrices of $G$-permutations of $M'$.
    Thus we have
    \begin{align*}
        F_{x, y} &\leq 1 + F_{x, \lfloor 3\epsilon y\rfloor} + F_{\lfloor 3\epsilon x\rfloor, y}\\
        &\leq 1 + \left(2^{\frac{\log (3\epsilon xy)}{\log(1/3\epsilon)} + 1} - 1\right)\cdot 2 \\
        &= 2^{\frac{\log (xy)}{\log(1/3\epsilon)} + 1} - 1.
    \end{align*}
    This completes the inductive proof that $F_{x, y} \leq 2^{\frac{\log (xy)}{\log(1/3\epsilon)} + 1} - 1$.

    Plugging in $x = n$ and $y = m$ gives that $M$ is the sum of at most
    \[ \mathcal {N} \leq 2^{\frac{\log (nm)}{\log(1/3\epsilon)} + 1} - 1 = (nm)^{O(1/\log(1/\epsilon))} \]
    submatrices of $G$-permutations of $M'$.

    We next show that this decomposition gives a depth-2 circuit computing $M$.

    Let $M'$ be computed by a depth-2 circuit
    \[ M' = \sum_{i\in I} U_i V_i^\sfT. \]
    Then for any $g = (A_g, B_g)\in G$, we have that $g\cdot M'$ is computed by the depth-2 circuit
    \begin{align*}
        g \cdot M' &= A_g M' B_g^\sfT\\
        &= \sum_{i\in I} (A_g U_i) (B_g V_i)^\sfT,
    \end{align*}
    which has degree at most $d$ and size at most $s$. Finally, by adding all the circuits computing
    the $\mathcal N$ submatrices in the decomposition of $M$, we have a depth-2 circuit computing $M$ with degree at most
    $\mathcal{N} d$ and size at most $\mathcal{N} s$.
\end{proof}

\begin{lemma}
    If matrix $M$ is fixed under a bitransitive action, then $M^{\otimes N}$ also has a bitransitive action.
\end{lemma}
\begin{proof}
    Let $M$ be a matrix of size $n\times m$ over $\bbF$,
    and $G$ be a group of permutations with coefficients acting on $p\times q$ matrices that fixes $M$.
    We define the group $G^{\times N}$ to be the $N$-fold direct product of $G$. Consider the mapping $\phi \colon G^{\times N} \to \GL_{n^N}(\bbF)\times \GL_{m^N}(\bbF)$
    defined by
    \[ \phi(g_1,\dots,g_N) = (A_{g_1}\otimes \cdots \otimes A_{g_N}, B_{g_1}\otimes \cdots \otimes B_{g_N}). \]
    We identify the space $\bbF^{n^N} = (\bbF^n)^{\otimes N}$, so its standard basis can be written
    as $e_{(i_1,\dots,i_N)} = e_{i_1} \otimes \cdots \otimes e_{i_N}$, where $i_1,\dots,i_N \in [n]$.

    For each $j\in [N]$, since $G$ is a group of permutations with coefficients, we denote
    $A_{g_j} e_{i_j} = a_j e_{s_j}$ for some $a_j\in \bbF$ and $s_j\in [n]$. Therefore, by the mixed
    product property of the Kronecker product, we have
    \begin{align*}
        (A_{g_1}\otimes \cdots \otimes A_{g_N}) (e_{i_1} \otimes \cdots \otimes e_{i_N})
        &= (A_{g_1} e_{i_1}) \otimes \cdots \otimes (A_{g_N} e_{i_N})\\
        &= a_1\cdots a_N (e_{s_1} \otimes \cdots \otimes e_{s_N}).
    \end{align*}
    It can be further verified that $A_{g_1}\otimes \cdots \otimes A_{g_N}$ is a permutation with coefficients,
    and similarly for $B_{g_1}\otimes \cdots \otimes B_{g_N}$.

    For any sequence $(i_1,\dots,i_N)$ and $(j_1,\dots,j_N)$, by the bitransitivity of $G$, there exist
    $g_1, \dots, g_N$ such that $g_\ell$ maps $i_\ell$ to $j_\ell$ for all $\ell\in [N]$. Thus
    $(g_1,\dots,g_N)$ maps $e_{(i_1,\dots,i_N)}$ to $e_{(j_1,\dots,j_N)}$. This shows that
    $G^{\times N}$ acts transitively on the rows of $n^N\times m^N$ matrices, and similarly on the columns.
\end{proof}

For any $n\times m$ matrix $M$ such that the group action $G$ on $n\times m$ matrices
\[ (A, B) \cdot N = A N B^\sfT \]
fixes the matrix $M$, and any depth-2 circuit $C$ computing $M^{\otimes d}$, we define function $\calD_C^r$ as follows:
\[ \calD_C^r(p_1,\dots,p_k) = \sum_{i_1 + \cdots + i_k = d} \left(\frac{p_1}{|X_1|}\right)^{i_1} \cdots 
\left(\frac{p_k}{|X_k|}\right)^{i_k}
\sum_{\substack{x\in X^{\times d} \\ \type(x) = (i_1,\dots,i_k)}} \log r_C(x), \]
where $\type(x) = (i_1,\dots,i_k)$ denotes the number of elements of $x = (x_1,\dots,x_d)$ in each orbit $X_1,\dots,X_k$, i.e., $i_j$ is the number of $\ell$ such that $x_\ell \in X_j$.
The function is defined over $p\in \Delta^k$, where $\Delta^k$ is the standard simplex in $\bbR^k$, i.e., $\Delta^k = \{p\in \bbR^k : p_i \geq 0, \sum_{i=1}^k p_i = 1\}$.

Alternatively, we can write $\calD_C^r(p_1,\dots,p_k)$ as
\[ \calD_C^r(p_1,\dots,p_k) = \sum_{x\in X^{\times d}} \left( \prod_{o=1}^d \frac{p_{\overline{x_o}}}{|X_{\overline{x_o}}|} \right) \log r_C(x), \]
where $\overline{x_o}$ denotes the orbit of $x_o$ under the group action of $G$.
Similarly, we define $\calD_C^c$ for the column orbits $Y_1,\dots,Y_\ell$ as follows:
\begin{align*}
    \calD_C^c(q_1,\dots,q_\ell) &= \sum_{j_1 + \cdots + j_\ell = d} 
    \left(\frac{q_1}{|Y_1|}\right)^{j_1} \cdots \left(\frac{q_\ell}{|Y_\ell|}\right)^{j_\ell}
    \sum_{\substack{y\in Y^{\times d} \\ \type(y) = (j_1,\dots,j_\ell)}} \log c_C(y)\\
    &= \sum_{y\in Y^{\times d}} \left( \prod_{o=1}^d \frac{q_{\overline{y_o}}}{|Y_{\overline{y_o}}|} \right) \log c_C(y).
\end{align*}

We now state our main result about fixing holes with orbits, which we prove in the next subsection.

\begin{theorem} \label{thm:fix-holes-main}
    Let $M$ be a matrix of size $n\times m$ over $\bbF$, and $G\subset \GL_n(\bbF)\times \GL_m(\bbF)$ be a group of permutations with coefficients
    which has $k$ orbits acting on rows and $\ell$ orbits acting on columns, and such that the group action $G$ on $n\times m$ matrices
    \[ (A, B) \cdot N = A N B^\sfT \]
    fixes the matrix $M$.

    Suppose $C$ is a depth-2 circuit computing $M^{\otimes d}$.
    Then we have
    \[ \log \delta(M) \leq \frac 1d \max \left( \sup_{p \in \Delta^k} \calD_C^r(p),
    \sup_{q \in \Delta^\ell} \calD_C^c(q) \right). \]
    The similar statements for $\OR$ and $\PAR$ complexities also hold.
    In that case, $G$ is replaced by a subgroup of symmetric group $G \subset \Sym_n \times \Sym_m$ (i.e., all the coefficients in the groups of permutations with coefficients must be 1),
    and the group action, orbits, and the function $\calD_C^r$ and $\calD_C^c$ are defined similarly.
\end{theorem}

\subsection{Proof of Theorem \ref{thm:fix-holes-main}}

We only provide the proof when the matrix is over a field; the proof for the $\OR$ and $\PAR$ complexities
can be done similarly.

\begin{lemma} \label{lem:orbit-decomposition}
    For any $N\geq 1$, there exists a group $G_N$ of permutations with coefficients that fixes the matrix $M^{\otimes N}$,
    such that the orbits of rows $X^{\times N}$ are as follows:
    For $x=(x_1,\dots,x_N)$ and $x' = (x'_1,\dots,x'_N)$, we have $x\sim x'$ if and only if
    $\type(x) = \type(x')$, i.e., $\{ Gx_1,
    \dots, Gx_n \} = \{ Gx'_1,
    \dots, Gx'_n \}$ as multisets, where $Gx_i$ denotes the orbit of $x_i$ under the group action of $G$.
    Also, the orbits of columns $Y^{\times N}$ are defined similarly.
\end{lemma}
\begin{proof}
    We consider the following two group actions.
    
    Firstly, consider the $N$-fold direct product of the group $G$,
    i.e., $G^{\times N}$ acting on $X^{\times N}$ and $Y^{\times N}$ entrywise.
    The orbits of $G^{\times N}$ on $X^{\times N}$ are: For $x = (x_1,\dots,x_N)$ and $x' = (x'_1,\dots,x'_N)$,
    we have $x\sim x'$ if and only if $x_i\sim x_i'$ for all $i\in [N]$.
    Similarly for the columns. This group action fixes the matrix $M^{\otimes N}$, since
    for any $(g_1,\dots,g_N)\in G^{\times N}$, we have
    \begin{align*}
        &\phantom{=} (A_{g_1}\otimes \cdots \otimes A_{g_N}) M^{\otimes N} 
        (B_{g_1}\otimes \cdots \otimes B_{g_N})^\sfT\\
        &= (A_{g_1} M B_{g_1}^\sfT) \otimes \cdots \otimes
        (A_{g_N} M B_{g_N}^\sfT) \tag*{(mixed product property)}\\
        &= M^{\otimes N}.
    \end{align*}

    The second group action is the symmetric group $\Sym_N$ acting on $X^{\times N}$ and $Y^{\times N}$ by permuting the coordinates, i.e.,
    for $\sigma\in \Sym_N$, $\sigma\cdot x = (x_{\sigma(1)},\dots,x_{\sigma(N)})$. This induces a permutation on the basis vectors
    of $\bbF^{n^N}$ and $\bbF^{m^N}$, thus we have a group of permutations with coefficients acting on $n^N\times m^N$ matrices.
    This group action also fixes the matrix $M^{\otimes N}$, since for any $\sigma\in \Sym_N$, we have,
    for $x \in X^{\times N}$ and $y\in Y^{\times N}$,
    \begin{align*}
        (\sigma \cdot M^{\otimes N})_{x, y} &= M^{\otimes N}_{\sigma^{-1}\cdot x, \sigma^{-1}\cdot y}\\
        &= \prod_{i=1}^N M_{x_{\sigma^{-1}(i)}, y_{\sigma^{-1}(i)}}\\
        &= \prod_{i=1}^N M_{x_i, y_i}\\
        &= (M^{\otimes N})_{x, y}.
    \end{align*}
    The desired group $G_N$ is the group generated by $G^{\times N}$ and $\Sym_N$ when they are identified
    as subgroups of the group of permutations with coefficients. Thus $G_N$ still fixes the matrix $M^{\otimes N}$,
    and the orbits of rows and columns are as described.
\end{proof}

Let $i = (i_1,\dots,i_k)$ and $j = (j_1,\dots,j_\ell)$
be two types of $X^{\times N}$ and $Y^{\times N}$ respectively. We define
$M^{\otimes N}_{i,j}$ to be the submatrix of $M^{\otimes N}$ indexed by the rows of type $i$ and columns of type $j$,
i.e.,
\[ M^{\otimes N}_{i,j} = \left. M^{\otimes N}\right|_{\substack{x: \type(x) = i\\ y: \type(y) = j}}. \]
The above lemma~\ref{lem:orbit-decomposition} shows that $G_N$ acts bitransitively and fixes the submatrix $M^{\otimes N}_{i,j}$.
Now suppose $N$ is a multiple of $d$. For the circuit $C$ computing $M^{\otimes d}$, its
Kronecker power $C^{\otimes (N/d)}$ computes $M^{\otimes N}$.
A row $x = (x_1,\dots,x_N)$ has degree
\[ r_{C^{\otimes (N/d)}}(x) = \prod_{z=1}^{N/d} r_C(x_{(z-1)d + 1}, \dots, x_{(z-1)d + d}). \]
We need the following lemma to bound the typical degree of a row having type $i$.
\begin{lemma} \label{lem:typical-degree}
    For any $\epsilon > 0$ and type $i$, a random row $x$ of type $i$ has
    \[ \Pr_{x : \type(x) = i}\left[\frac{\log r_{C^{\otimes (N/d)}}(x)}{N/d} \geq \calD_C^r(p) + \epsilon\right] \leq e^{-\Omega(\epsilon^2 N)}, \]
    and similarly, for a random column $y$ of type $j$, we have
    \[ \Pr_{y : \type(y) = j}\left[\frac{\log r_{C^{\otimes (N/d)}}(y)}{N/d} \geq \calD_C^c(q) + \epsilon\right] \leq e^{-\Omega(\epsilon^2 N)}. \]
    Here the hidden constant in $\Omega(\cdot)$ doesn't depend on the types $i, j$.
\end{lemma}
\begin{proof}
    Write $i=(i_1,\dots,i_k)$, and let $p = i / N = (i_1/N,\dots,i_k/N)$ be the probability distribution of the types of the rows. Consider a random
    $x = (x_1,\dots,x_N)$, where each $x_o$ is chosen by first picking a random orbit of $X$ independently according to $p$, i.e., $\Pr[x_o \in X_r] = p_r$, and then
    choosing a random element $x_o$ from the orbit $X_r$ uniformly.

    Then define random variables $T_1,\dots,T_{N/d}$ by
    \[ T_z = \log r_C(x_{(z-1)d + 1}, \dots, x_{(z-1)d + d}). \]
    We have that $T_z$ are independent and identically distributed, and
    \begin{align*}
        \Ex[T_z] &= \sum_{x_1,\dots,x_d \in X} \Pr[x_{1}, \dots, x_{d}] \log r_C(x_1,\dots,x_d)\\
        &= \calD_C^r(p).
    \end{align*}
    Let $T = T_1+\cdots+T_{N/d}$. By Hoeffding's inequality, we have
    \[ \Pr[T \geq \Ex[T] + N\epsilon] \leq e^{-\Omega(N\epsilon^2)}. \]

    Note that $x$ might not have type $i$, but we have
    \[ \Pr[\type(x) = i] = \binom{N}{i_1,\dots,i_k} p_1^{i_1} \cdots p_k^{i_k}
    \stackrel{\text{Lem.~\ref{lem:binom-ineq}}}{\geq} (N+1)^{-k}. \]
    Thus we have the conditional probability
    \[ \Pr[T \geq \Ex[T] + N\epsilon \mid \type(x) = i] \leq (N+1)^k e^{-\Omega(\epsilon^2 N)}. \]
    Note that when conditioning on $\type(x) = i$, $x$ is uniformly distributed over the rows of type $i$.
    Since
    \[ T = \sum_{z=1}^{N/d} \log r_C(x_{(z-1)d + 1}, \dots, x_{(z-1)d + d}) 
    = \log r_{C^{\otimes(N/d)}}(x), \]
    we have
    \[ \Pr_{x : \type(x)=i}\left[\frac{\log r_{C^{\otimes (N/d)}}(x)}{N/d} \geq \calD_C^r(p) + \epsilon\right] \leq e^{-\Omega(\epsilon^2 N)}. \]
    The proof for the columns is similar.
\end{proof}

Now we are able to prove Theorem \ref{thm:fix-holes-main}.
Let $N$ be a multiple of $d$. We split $M^{\otimes N}$ into the sum
\[ M^{\otimes N} = \sum_{i,j} M^{\otimes N}_{i,j}, \]
where $i$ and $j$ range over all types of rows and columns, respectively.
Let
\[ S = \max \left( \sup_{p \in \Delta^k} \calD_C^r(p),
  \sup_{q \in \Delta^\ell} \calD_C^c(q) \right). \]
Then by Lemma \ref{lem:typical-degree}, we have, for each $i$ and $j$, that the degrees of rows and columns
of $C^{\otimes (N/d)}$ are at most $e^{(S+\epsilon)N/d}$, except for a set of rows and columns
of total measure at most $e^{-\Omega(\epsilon^2 N)}$. By Lemma \ref{lem:orbit-decomposition},
the group $G_N$ acts bitransitively on the rows and columns of $M^{\otimes N}$, thus by
Theorem \ref{thm:fix-holes}, we have a circuit computing $M^{\otimes N}_{i,j}$ with degree at most
\[ (nm)^{O(1/\epsilon^2)} e^{(S+\epsilon)N/d} = \exp\left( \frac{N(S + \epsilon)}{d}
+ O\left(\frac{1}{\epsilon^2}\right) \right). \]
Taking $\epsilon = N^{-1/3}$, we have that the degree is at most $\exp(S N/d + O(N^{2/3}))$.

Since there are only $N^{O(1)}$ types of rows and columns, summing over all types, we have
a circuit computing $M^{\otimes N}$ with degree at most $\exp(S N/d + O(N^{2/3}))$.
Taking the limit $N\to\infty$, we have $\delta(M) \leq e^{S/d}$. Theorem~\ref{thm:fix-holes-main} thus follows by taking the logarithm. \hfill \qedsymbol

\subsection{Matrices with Bitransitive Action}

\begin{theorem}[Degree Uniformization] \label{thm:degree-uniformization}
    Let $M$ be a $n\times n$ matrix, which is fixed under a bitransitive action.
    Then $\sigma(M) = n\cdot \delta(M)$. The statements for $\OR$ and $\PAR$ complexities also hold.
\end{theorem}
\begin{proof}
    We first show that $\sigma(M) \leq n\cdot \delta(M)$. By the definition of $\delta(M)$, the matrix $M^{\otimes N}$
    has a depth-2 circuit of maximal degree $\delta(M)^{(1+o(1))N}$. Such a circuit can have size at most $2 n^N \cdot \delta(M)^{(1+o(1))N}$, so we have
    \[ \sigma(M) \leq \liminf_{N\to\infty} \left(2 n^N \cdot \delta(M)^{(1+o(1))N}\right)^{1/N} = n\cdot \delta(M). \]

    We next show that $\sigma(M) \geq n\cdot \delta(M)$. We only need to prove that, if $M$ has a depth-2 circuit of size $s$,
    then $\delta(M) \leq s/n$. Then the statement follows by taking the limit
    \[ \delta(M) \leq \lim_{N\to\infty} \left(\frac{S(M^{\otimes N})}{n^N}\right)^{1/N}
    = \frac{\sigma(M)}{n}. \]

    Since $M$ has a bitransitive action, the orbits of the group action on rows and columns are singletons,
    thus we have
    \[ \calD_C^r(1) = \frac 1 n\sum_{x\in X} \log r_C(x)
    \stackrel{\text{Jensen's Inq}}{\leq} \log \left( \frac 1 n \sum_{x\in X} r_C(x) \right) \leq \log \frac{s}{n}. \]
    Similarly, we have $\calD_C^c(1) \leq \log (s/n)$.
    By Theorem \ref{thm:fix-holes-main}, we have $\log \delta(M) \leq \log(s/n)$, thus $\delta(M) \leq s/n$.
                                                    \\    \\                                                    \end{proof}

Recall that for any abelian group $G$, the character group $\widehat G$ is the set of all group homomorphisms
from $G$ to the multiplicative group of complex numbers $\bbC^\times$.
For any finite abelian group $G$, we call the discrete Fourier transform matrix $\DFT_G$ a $\bbC^{|G|\times |G|}$ matrix
whose rows are indexed by $\widehat G$, and columns are indexed by $G$, such that
$(\DFT_G)_{\chi, g} = \chi(g)$ for all $\chi\in \widehat G$ and $g\in G$.
\begin{corollary} \label{cor:dft_symmetry}
    Let $G$ be a finite abelian group, we have $\sigma(\DFT_G) = |G| \cdot \delta(\DFT_G)$.
\end{corollary}
\begin{proof}
    For any $g,h\in G$ and $\chi\in \widehat G$, since $\chi(g+h) = \chi(g) \chi(h)$, we have
    \[ (\DFT_G)_{\chi, g+h} = (\DFT_G)_{\chi, g} \chi(h). \]
    Thus, considering the group $G$ acting on $|G|\times |G|$ matrices by
    $A_h = I$, and $B_h$ satisfies $B_h e_g = \chi(h)^{-1} e_{g-h}$, we have
    \[ [(A_h, B_h) \cdot \DFT_G]_{\chi, g} = (\DFT_{G})_{\chi, g+h} \chi^{-1}(h)
    = (\DFT_G)_{\chi, g}. \]
    Thus the group action is permutations with coefficients, and is transitive on columns and fixes the matrix $\DFT_G$.
    Since $G$ is also the dual group of $\widehat G$, we have a group of permutations with coefficients,
    acting transitively on rows and fixing the matrix $\DFT_G$.
    These two groups generate a bitransitive action fixing $\DFT_G$, thus by Theorem \ref{thm:degree-uniformization},
    we have $\sigma(\DFT_G) = |G| \cdot \delta(\DFT_G)$.
\end{proof}

Then we can directly transfer the results of Alman, Guan and Padaki \cite{AGP23KroneckerCircuits}
to the setting of bounding the maximal degree of depth-2 circuits.
\begin{corollary} \label{cor:dft}
    The Walsh--Hadamard matrix $H$ has $\delta(H) \leq 2^{0.443}$.
    For any finite abelian group $G$,
    the Fourier matrix $\DFT_{G}$ has $\delta(\DFT_{G}) \leq |G|^{0.5 - a_{|G|}}$,
    where $a_n = \Omega\left(\frac{1}{n^2 \log n}\right)$.
\end{corollary}
\begin{proof}
    The upper bound of Walsh--Hadamard matrix is from the fact that
    $\sigma(H)\leq 2^{1.443}$ \cite[Theorem 1.6]{AGP23KroneckerCircuits}.
    The upper bound of Fourier matrix is from the size upper bound of general $n\times n$
    matrix $M$, where $\sigma(M)\leq n^{1.5-a_n}$ \cite[Theorem 1.4]{AGP23KroneckerCircuits}.
\end{proof}

\subsection{Disjointness Matrix}

\begin{lemma} \label{lem:tensor-density}
    For two circuits $C$ and $D$ computing the matrices $M^{\otimes d}$ and $M^{\otimes e}$, respectively,
    the circuit $C\otimes D$ computes the matrix $R^{\otimes (d+e)}$ and has
    \begin{align*}
        \calD^{r}_{C\otimes D} &= \calD^r_C + \calD^r_D,\\
        \calD^{c}_{C\otimes D} &= \calD^c_C + \calD^c_D.
    \end{align*}
\end{lemma}
\begin{proof}
    We first show the row density function. Let the rows $X$ of $M$ has $k$ orbits.
    For any $p_1+\cdots+p_k = 1$, we have
    \begin{align*}
        \calD_{C\otimes D}^r(p_1,\dots,p_k) &= 
        \sum_{x\in X^{\times (d+e)}} \left( \prod_{o=1}^{d+e} \frac{p_{\overline{x_o}}}{|X_{\overline{x_o}}|} \right) \log r_{C\otimes D}(x)\\
        &= \sum_{x\in X^{\times d}, x'\in X^{\times e}} \left( \prod_{o=1}^{d} \frac{p_{\overline{x_o}}}{|X_{\overline{x_o}}|} \right)
        \left( \prod_{o=1}^{e} \frac{p_{\overline{x'_o}}}{|X_{\overline{x'_o}}|} \right) \log r_{C\otimes D}(x, x')\\
        &= \sum_{x\in X^{\times d}, x'\in X^{\times e}} \left( \prod_{o=1}^{d} \frac{p_{\overline{x_o}}}{|X_{\overline{x_o}}|} \right)
        \left( \prod_{o=1}^{e} \frac{p_{\overline{x'_o}}}{|X_{\overline{x'_o}}|} \right) (\log r_{C}(x)
        + \log r_D(x'))\\
        &= \sum_{x\in X^{\times d}} \left( \prod_{o=1}^{d} \frac{p_{\overline{x_o}}}{|X_{\overline{x_o}}|} \right) \log r_{C}(x)
        + \sum_{x'\in X^{\times e}} \left( \prod_{o=1}^{e} \frac{p_{\overline{x'_o}}}{|X_{\overline{x'_o}}|} \right) \log r_{D}(x')\\
        &= \calD_C^r(p) + \calD_D^r(p).
    \end{align*}
    The column density function can be similarly computed.
            \\    \\    \\            \end{proof}

The disjointness matrix doesn't have a bitransitive action. Each row and column will have a single orbit.
\begin{corollary} \label{cor:disjointness}
    The disjointness matrix has $\delta_\PAR(R) \leq 2^{1/3}$.
\end{corollary}
\begin{proof}
    Consider the following decomposition of  $R^{\otimes 2}$, which is a partition.
    \begin{align*}
        R^{\otimes 2} &= \begin{bmatrix}
            1 & 1 & 1 & 1\\
            1 & 0 & 1 & 0\\
            1 & 1 & 0 & 0\\
            1 & 0 & 0 & 0
        \end{bmatrix}\\
        &= \begin{bmatrix}
            1 \\ 1\\ 1\\ 1
        \end{bmatrix} \begin{bmatrix}
            1 & 0 & 0 & 0
        \end{bmatrix} + \begin{bmatrix}
            1 \\ 0 \\ 0 \\ 0
        \end{bmatrix} \begin{bmatrix}
            0 & 1 & 1 & 1
        \end{bmatrix} + e_2 e_3^\sfT + e_3 e_2^\sfT.
    \end{align*}
    The corresponding depth-2 circuit has row degree $r_C = (2, 2, 2, 1)$
    and column degree $c_C = (1, 2, 2, 1)$, so we have
    \begin{align*}
        \calD_C^r(q, p) &= 2pq, \\
        \calD_C^c(q, p) &= q^2 + 2pq.
    \end{align*}
    Here we take logarithm in base 2, for simplicity. Therefore, considering the circuit $C \otimes C^\sfT$
    computing $R^{\otimes 4}$ (in order to symmetrize our circuit), we have
    \[ \calD_{C\otimes C^\sfT}^r(q, p) = \calD_{C\otimes C^\sfT}^c(q, p) = q^2 + 4pq. \]
    The maximal value of the density function is taken at $\calD_{C\otimes C^\sfT}^r(2/3, 1/3) = 4/3$. Thus by Theorem \ref{thm:fix-holes-main},
    we have $\log_2 \delta_\PAR(R) \leq \frac 14 \cdot \frac 43$, i.e., $\delta_\PAR(R) \leq 2^{1/3}$.
\end{proof}

\begin{remark}
    In general, for a matrix without a bitransitive action, it might be possible to
    bound the maximal degree for $M^{\otimes N}_{i,j}$ for different types $i, j$ via different
    circuits. This suggests that a future direction on bounding $\delta(R)$ might be
    to consider using another circuit to bound the degree of the submatrix $R^{\otimes N}_{(2N/3, N/3), (2N/3, N/3)}$,
    which is the bottleneck of the above construction.
\end{remark}

\section{Algorithms from Low-Degree Circuits} \label{sec:algorithms}

Let $M\in \bbF^{p\times q}$ be a matrix. We consider the algorithmic problem of \emph{sparse
vector-matrix-vector multiplication} defined as follows: Given two vectors $x\in \bbF^{p^d}$ and $y\in \bbF^{q^d}$,
such that $x$ and $y$ are supported on at most $N$ coordinates, compute the product $x^\sfT M^{\otimes d} y$.
We first consider the problem in the algebraic computation model. In this case,
we call the problem $\SVMV_{M,d,N,S,T}$, where the parameters are the matrix $M$, the dimension $d$, the sparsity $N$, and the support sets
$S \subset [p]^d$, $T \subset [q]^d$ of $x$ and $y$, respectively. The input to the problem is
the entries of $x$, $y$, and the entries of $M$.

\begin{theorem} \label{thm:algebraic-algorithm}
    For any matrix $M$ and any $\epsilon > 0$, the problem $\SVMV_{M,d,N,S,T}$ can be solved by an
    arithmetic circuit of size $O(N \cdot (\delta(M) + \epsilon)^d)$.
    Moreover, such a circuit can be generated by an algorithm running in time $\tilde O(N \cdot (\delta(M) + \epsilon)^d)$.
\end{theorem}
\begin{proof}
    For any $\epsilon' > 0$, by the definition of $\delta(M)$, there is a depth-2 circuit $C$ computing $M^{\otimes d_0}$,
    where $C$ has degree at most $(\delta(M) + \epsilon')^{d_0}$ for some $d_0$.
    Also, let $C'$ be a one-sided decomposition depth-2 circuit computing $M$ (in fact,
    any circuit computing $M$ works).

    Then for any $d$, consider the circuit $C^{\otimes l} \otimes C'^{\otimes r}$
    where $l = \lfloor d/d_0 \rfloor$ and $r = d \mod d_0$. This circuit
    $C^{\otimes l} \otimes C'^{\otimes r}$ computes $(M^{\otimes d_0})^{\otimes l}
    \otimes M^{\otimes r} = M^{\otimes d}$ and has degree at most
    $O((\delta(M)+\epsilon')^{d})$.
    Let $M^{d_0} = U V^\sfT$ and $M = U' V'^\sfT$ be the decompositions given by the circuits $C$ and $C'$, respectively. Then the circuit $C^{\otimes l} \otimes C'^{\otimes r}$ computes
    \[ (U V^\sfT)^{\otimes l} \otimes (U' V'^\sfT)^{\otimes r} = (U^{\otimes l} \otimes U'^{\otimes r}) (V^{\otimes l} \otimes V'^{\otimes r})^\sfT. \]
    Then we have
    \begin{align*}
        x^\sfT M^{\otimes d} y &= x^\sfT (U^{\otimes l} \otimes U'^{\otimes r}) (V^{\otimes l} \otimes V'^{\otimes r})^\sfT y\\
        &= \sum_{k} (x^\sfT (U^{\otimes l} \otimes U'^{\otimes r}))_k \cdot (y^\sfT (V^{\otimes l} \otimes V'^{\otimes r}))_k.
    \end{align*}
    For each $(i_1,\dots,i_d) \in \supp(x)$, by the mixed product property, we have
    \begin{align*}
        &\quad (U^{\otimes l} \otimes U'^{\otimes r})_{(i_1,\dots,i_d) , (k_1,\dots,k_d)}\\
        &= U_{(i_1,\dots,i_{d_0}), (k_1,\dots,k_{d_0})} \cdots
        U_{(i_{(l-1)d_0+1},\dots,i_{ld}), (k_{(l-1)d_0+1},\dots,k_{ld})} \cdot U'_{i_{ld+1}, l_{ld+1}}
        \cdots U'_{i_d, k_d}.
    \end{align*}
    When $\supp(x), \supp(y)$ is determined, we show that one can build a arithmetic circuit doing the following:
    \begin{compactitem}
        \item For each $k$ such that $(x^\sfT (U^{\otimes l} \otimes U'^{\otimes r}))_k$ and $(y^\sfT (V^{\otimes l} \otimes V'^{\otimes r}))_k$
        can be nonzero, compute these two values via a sum-of-product circuit.
        \item Compute the inner product of the two values over all such $k$ via a sum-of-product circuit.
    \end{compactitem}
    Then the output of the circuit is the desired value $x^\sfT M^{\otimes d} y$.

    For each $(i_1,\dots,i_d) \in \supp(x)$, the algorithm can iterate
    over all $(k_1,\dots,k_{d_0})\in \supp (U_{(i_1,\dots,i_{d_0})})$, and 
    similarly for the rest, to enumerate all $(k_1,\dots,k_d)$ such that
    $x_i$ has a nonzero contribution to the product $(x^\sfT (U^{\otimes l} \otimes U'^{\otimes r}))_k$.

    Similarly, the algorithm can find all wires starting from $j\in \supp(y)$ to 
    $(y^\sfT (V^{\otimes l} \otimes V'^{\otimes r}))_k$ with nonzero contribution.
    These two steps finishes the linear circuit step.
    The inner product step is trivial.

    By counting on the $i, j$ side, the number of gates is upper bounded by
    \[ O(N \cdot (\delta(M) + \epsilon')^d d^{O(1)}). \]
    Under a careful implementation, the algorithm runs in time $\tilde O(N \cdot (\delta(M) + \epsilon')^d)$.
    Taking $\epsilon' < \epsilon$, we have the desired result.
\end{proof}

\begin{theorem} \label{thm:ov-algorithms}
    The problem $\#\OV$ can be solved in $\tilde O(n \cdot \delta(R)^{(1+\epsilon) d})$ time,
    $\OV$ can be solved in $\tilde O(n \cdot \delta_\OR(R)^{(1+\epsilon) d})$ time,
    and $\#\OV_{\bbZ_m}$ can be solved in $\tilde O(n \cdot \delta(\DFT_{\bbZ_m})^{(1+\epsilon) d})$ time,
    for any $\epsilon > 0$.
\end{theorem}
\begin{proof}
    We first describe the three algorithms, and then measure their running times.
    
    The algorithms for $\#\OV$ and $\OV$ are direct from the previous theorem.
    For vectors $u_1,\dots,u_n, v_1,\dots,v_n \in \{0, 1\}^d$, let the vectors $x, y$
    be the indicator vectors of the sets $\{u_1,\dots,u_n\}, \{v_1,\dots,v_n\}$, respectively.
    Then the problem $\#\OV$ and $\OV$ can be reduced to the sparse vector-matrix-vector multiplication  $x^\top R^{\otimes d} y$.

    Finally, consider the problem $\#\OV_{\bbZ_m}$.
    Our algorithm stems from the identity
    \[ \sum_{i=0}^{m-1} \omega^i = m \cdot \iv{k \equiv 0 \pmod m}. \]
    Thus for $u, v \in \bbZ_m^d$, we have
    \begin{align*}
        \iv{u^\sfT v \equiv 0 \pmod m} &= \frac 1m \sum_{k=0}^{m-1} \omega^{k \cdot u^\sfT v}\\
        &= \frac 1m \sum_{k=0}^{m-1} \omega^{(k \cdot u)^\sfT v}.
    \end{align*}
    Let $x^{(k)}$ be the indicator vector of $\{k\cdot u_i\}$, and $y$ be the indicator vector of $\{v_i\}$. Then the problem $\#\OV_{\bbZ_m}$ can be reduced to sparse vector-matrix-vector multiplication, by
    \begin{align*}
        \#\{u_i^\sfT v_j \equiv 0 \pmod m\} &= \sum_{1\leq i,j\leq n} \frac 1m \sum_{k=0}^{m-1} \omega^{(ku_i)^\sfT v_j} \\
        &= \sum_{1\leq i,j\leq n} \frac 1m \sum_{k=0}^{m-1} \one_{u_i}^\sfT \DFT_{\bbZ_m}^{\otimes d} \one_{v_j} \\
        &= \frac 1 m \sum_{k=0}^{m-1} (x^{(k)})^{\sfT} \DFT_{\bbZ_m}^{\otimes d} y .
    \end{align*}
    Thus the problem $\#\OV_{\bbZ_m}$ is reduced to $m$ instances of sparse vector-matrix-vector multiplication
    for the matrix $\DFT_{\bbZ_m}^{\otimes d}$.

    We now measure the running times. 
    For $\OV$, since operations in the Boolean semiring have constant cost, we have the desired complexity.

    Finally, we argue the bit complexity for $\#\OV$ and $\#\OV_{\bbZ_m}$, especially dealing with the problem
    that $\DFT_{\bbZ_m}$ is defined over $\bbC$.
    First note that the roots of unity of $\bbZ_m$ generate the field $\bbQ(\omega_m) \subset \bbC$,
    where $\omega_m$ is a primitive $m$-th root of unity, which has the cyclic polynomial $\Phi_m(X)$
    as its minimal polynomial.
    By Theorem \ref{thm:field-extension}, we have $\delta_\bbC(\DFT_{\bbZ_m}) = \delta_{\bbQ(\omega_n)}(\DFT_{\bbZ_m})$.
    We then take a closer look at the circuit generated in Theorem \ref{thm:algebraic-algorithm}.
    For any $\epsilon>0$, we have a circuit $C$ over $\bbQ(\omega_n)$ computing $\DFT_{\bbZ_m}^{\otimes d_0}$ with degree at most
    $(\delta(\DFT_{\bbZ_m}) + \epsilon)^{d_0}$.
    We first consider identifying $\bbQ(\omega_m) = \bbQ[X] / \Phi_m(X)$, then converting $C$
    to a circuit $C'$ over $\bbQ[X]$ by replacing each $\omega_m$ with $X$.
    Then we consider extracting the common denominator of the coefficients of $C'$, and let
    $C''$ be the circuit over $\bbZ[X]$ such that $C''/q$ computes the same matrix as $C'$,
    where $q$ is an integer. Then, similar to the proof of Theorem \ref{thm:algebraic-algorithm},
    consider the $\lfloor d/d_0\rfloor$-th Kronecker power of $C''$. Taking a Kronecker
    product with a constant sized circuit, the resulting circuit given by
    Theorem \ref{thm:algebraic-algorithm} computes
    the sparse vector-matrix-vector multiplication problem as a sum of $n \cdot 2^{O(d)}$ terms, where each term is
    a product of $O(d)$ constant integers. Therefore, during the whole computation, the integer polynomial
    always has degree at most $O(d)$ and the coefficients have $O(d + \log n)$ bits. 

    We then transform the output, first passing through the quotient map $\bbZ[X] \to \bbZ[X]/\Phi_m(X)$,
    i.e., computing the polynomial modulo $\Phi_m(X)$. The coefficients of the residue polynomial have
    at most $O(d \cdot (d + \log n))$ bits, and
    the na\"ive long division algorithm computes them in $d^{O(1)}$ time.
    Finally, dividing by the common denominator, we obtain the result of the original problem.
    In conclusion, in this problem, simulating the computation over $\bbQ(\omega_m)$ can be done in
    $d^{O(1)}$ time per arithmetic operation. The time complexity
    of the algorithm is then $\tilde O(n \cdot \delta(\DFT_{\bbZ_m})^{(1+\epsilon) d})$.

    For $\#\OV$, there is no need to consider the field extension, so the analysis is the same as above
    (or even simpler).
\end{proof}

Since $\delta_\OR(R) = 2/\sqrt 3 < 1.155$ (Theorem \ref{thm:disjointness_delta_or}) and $\delta_\PAR(R) \leq 2^{1/3} < 1.26$
(Theorem \ref{thm:disjointness_par}), we have the following corollaries.
\begin{corollary} \label{cor:ov-par}
    The problem $\#\OV$ can be solved in $\tilde O(n \cdot 1.26^d)$ time.
\end{corollary}
\begin{corollary} \label{cor:ov}
    The problem $\OV$ can be solved in $\tilde O(n \cdot 1.155^d)$ time.
\end{corollary}

By Corollary \ref{cor:dft}, $\delta(H)$ is bounded by $2^{0.443} < 1.36$ and
$\delta(\DFT_{\bbZ_m})$ is bounded by $m^{0.5 - a_m}$ for $a_m = \Omega(1/m^2 \log m)$,
giving the following corollary.
\begin{corollary} \label{cor:ov-dft}
    The problem $\#\OV_{\bbZ_2}$ can be solved in $\tilde O(n \cdot 1.36^d)$ time.
    Moreover, the problem $\#\OV_{\bbZ_m}$ can be solved in $\tilde O(n \cdot m^{(1/2-a_m) d})$ time,
    where $a_m = \Omega(1/m^2 \log m)$.
\end{corollary}

\section*{Acknowledgements}
We thank anonymous reviewers, Toni Pitassi, and Ryan Williams for helpful discussions and comments on earlier drafts of this paper.

\bibliographystyle{alpha}
\bibliography{main.bib}

\newcommand{\etalchar}[1]{$^{#1}$}
\begin{thebibliography}{ADW{\etalchar{+}}25}

\bibitem[ADW{\etalchar{+}}25]{alman2024more}
Josh Alman, Ran Duan, Virginia~Vassilevska Williams, Yinzhan Xu, Zixuan Xu, and Renfei Zhou.
\newblock More asymmetry yields faster matrix multiplication.
\newblock In {\em Proceedings of the 2025 Annual ACM-SIAM Symposium on Discrete Algorithms (SODA)}, pages 2005--2039, 2025.

\bibitem[AGP23]{AGP23KroneckerCircuits}
Josh Alman, Yunfeng Guan, and Ashwin Padaki.
\newblock Smaller low-depth circuits for kronecker powers.
\newblock In {\em Proceedings of the 2023 Annual ACM-SIAM Symposium on Discrete Algorithms (SODA)}, pages 4159--4187. SIAM, 2023.

\bibitem[AKW90]{alon1990linear}
Noga Alon, Mauricio Karchmer, and Avi Wigderson.
\newblock Linear circuits over gf(2).
\newblock {\em SIAM Journal on Computing}, 19(6):1064--1067, 1990.

\bibitem[AL25]{alman2025low}
Josh Alman and Jingxun Liang.
\newblock Low rank matrix rigidity: Tight lower bounds and hardness amplification.
\newblock In {\em Proceedings of the 57th Annual ACM Symposium on Theory of Computing}, pages 1383--1394, 2025.

\bibitem[Alm21]{alman2021kronecker}
Josh Alman.
\newblock Kronecker products, low-depth circuits, and matrix rigidity.
\newblock In {\em Proceedings of the 53rd Annual ACM SIGACT Symposium on Theory of Computing}, pages 772--785, 2021.

\bibitem[AR23]{alman2023faster}
Josh Alman and Kevin Rao.
\newblock Faster walsh-hadamard and discrete fourier transforms from matrix non-rigidity.
\newblock In {\em Proceedings of the 55th Annual ACM Symposium on Theory of Computing}, pages 455--462, 2023.

\bibitem[AW17]{AW17}
Josh Alman and Ryan Williams.
\newblock Probabilistic rank and matrix rigidity.
\newblock In {\em Proceedings of the 49th Annual ACM SIGACT Symposium on Theory of Computing}, STOC 2017, page 641–652, New York, NY, USA, 2017. Association for Computing Machinery.

\bibitem[AWY14]{abboud2014more}
Amir Abboud, Ryan Williams, and Huacheng Yu.
\newblock More applications of the polynomial method to algorithm design.
\newblock In {\em Proceedings of the twenty-sixth annual ACM-SIAM symposium on Discrete algorithms}, pages 218--230. SIAM, 2014.

\bibitem[BCKN15]{BCKN15fact}
Hans~L Bodlaender, Marek Cygan, Stefan Kratsch, and Jesper Nederlof.
\newblock Deterministic single exponential time algorithms for connectivity problems parameterized by treewidth.
\newblock {\em Information and Computation}, 243:86--111, 2015.

\bibitem[BGKM23]{bhargava2023fast}
Vishwas Bhargava, Sumanta Ghosh, Mrinal Kumar, and Chandra~Kanta Mohapatra.
\newblock Fast, algebraic multivariate multipoint evaluation in small characteristic and applications.
\newblock {\em Journal of the ACM}, 70(6):1--46, 2023.

\bibitem[BHKK09]{BHKK09halve}
Andreas Bj{\"o}rklund, Thore Husfeldt, Petteri Kaski, and Mikko Koivisto.
\newblock Counting paths and packings in halves.
\newblock In {\em European Symposium on Algorithms}, pages 578--586. Springer, 2009.

\bibitem[BKW19]{bjorklund2019solving}
Andreas Bj{\"o}rklund, Petteri Kaski, and Ryan Williams.
\newblock Solving systems of polynomial equations over gf (2) by a parity-counting self-reduction.
\newblock In {\em International Colloquium on Automata, Languages and Programming}, pages 1--13. Schloss Dagstuhl-Leibniz-Zentrum f{\"u}r Informatik, 2019.

\bibitem[BSSS06]{CodeTable06}
A.~E. Brouwer, J.~B. Shearer, N.~J.~A. Sloane, and W.~D. Smith.
\newblock A new table of constant weight codes.
\newblock {\em IEEE Transactions on Information Theory}, 36(6):1334--1380, 2006.

\bibitem[CILS17]{CILS17cover}
Dmitry Chistikov, Szabolcs Iv\'{a}n, Anna Lubiw, and Jeffrey Shallit.
\newblock {Fractional Coverings, Greedy Coverings, and Rectifier Networks}.
\newblock In Heribert Vollmer and Brigitte Vall\'{e}e, editors, {\em 34th Symposium on Theoretical Aspects of Computer Science (STACS 2017)}, volume~66 of {\em Leibniz International Proceedings in Informatics (LIPIcs)}, pages 23:1--23:14, Dagstuhl, Germany, 2017. Schloss Dagstuhl -- Leibniz-Zentrum f{\"u}r Informatik.

\bibitem[CKN18]{CKN18fact}
Marek Cygan, Stefan Kratsch, and Jesper Nederlof.
\newblock Fast hamiltonicity checking via bases of perfect matchings.
\newblock {\em Journal of the ACM (JACM)}, 65(3):1--46, 2018.

\bibitem[CW20]{chan2020deterministic}
Timothy~M Chan and R~Ryan Williams.
\newblock Deterministic apsp, orthogonal vectors, and more: Quickly derandomizing razborov-smolensky.
\newblock {\em ACM Transactions on Algorithms (TALG)}, 17(1):1--14, 2020.

\bibitem[DKW25]{DKW25kOV}
Anita D{\"u}rr, Evangelos Kipouridis, and Karol W{\k{e}}grzycki.
\newblock Faster algorithms for k-orthogonal vectors in low dimension.
\newblock {\em arXiv preprint arXiv:2507.11098}, 2025.

\bibitem[DL20]{DL19}
Zeev Dvir and Allen Liu.
\newblock Fourier and circulant matrices are not rigid.
\newblock {\em Theory OF Computing}, 16(20):1--48, 2020.

\bibitem[DWZ23]{DWZ23AsymmetricHashing}
Ran Duan, Hongxun Wu, and Renfei Zhou.
\newblock Faster matrix multiplication via asymmetric hashing.
\newblock In {\em 2023 IEEE 64th Annual Symposium on Foundations of Computer Science (FOCS)}, pages 2129--2138. IEEE, 2023.

\bibitem[Fur77]{furstenberg1977ergodic}
Harry Furstenberg.
\newblock Ergodic behavior of diagonal measures and a theorem of szemer{\'e}di on arithmetic progressions.
\newblock {\em J. Anal. Math.}, 31:204--256, 1977.

\bibitem[IP01]{impagliazzo2001complexity}
Russell Impagliazzo and Ramamohan Paturi.
\newblock On the complexity of k-sat.
\newblock {\em Journal of Computer and System Sciences}, 62(2):367--375, 2001.

\bibitem[JS13]{jukna2013complexity}
Stasys Jukna and Igor Sergeev.
\newblock Complexity of linear boolean operators.
\newblock {\em Foundations and Trends{\textregistered} in Theoretical Computer Science}, 9(1):1--123, 2013.

\bibitem[Kiv21]{kivva2021improved}
Bohdan Kivva.
\newblock Improved upper bounds for the rigidity of kronecker products.
\newblock In {\em 46th International Symposium on Mathematical Foundations of Computer Science (MFCS 2021)}. Schloss-Dagstuhl-Leibniz Zentrum f{\"u}r Informatik, 2021.

\bibitem[Lan12]{lang2012algebra}
Serge Lang.
\newblock {\em Algebra}, volume 211.
\newblock Springer Science \& Business Media, 2012.

\bibitem[Ned20a]{ned20survey}
Jesper Nederlof.
\newblock Algorithms for np-hard problems via rank-related parameters of matrices.
\newblock In {\em Treewidth, Kernels, and Algorithms: Essays Dedicated to Hans L. Bodlaender on the Occasion of His 60th Birthday}, pages 145--164. Springer, 2020.

\bibitem[Ned20b]{ned20bitsp}
Jesper Nederlof.
\newblock Bipartite tsp in $o(1.9999^n)$ time, assuming quadratic time matrix multiplication.
\newblock In {\em Proceedings of the 52nd Annual ACM SIGACT Symposium on Theory of Computing}, pages 40--53, 2020.

\bibitem[NW21]{NW21OV}
Jesper Nederlof and Karol W\k{e}grzycki.
\newblock Improving schroeppel and shamir's algorithm for subset sum via orthogonal vectors.
\newblock In {\em Proceedings of the 53rd Annual ACM SIGACT Symposium on Theory of Computing}, pages 1670--1683, 2021.

\bibitem[Pud00]{pudlak2000note}
Pavel Pudl{\'a}k.
\newblock A note on the use of determinant for proving lower bounds on the size of linear circuits.
\newblock {\em Information processing letters}, 74(5-6):197--201, 2000.

\bibitem[ROS94]{roychowdhury1994lower}
Vwani~P Roychowdhury, Alon Orlitsky, and Kai-Yeung Siu.
\newblock Lower bounds on threshold and related circuits via communication complexity.
\newblock {\em IEEE Transactions on Information Theory}, 40(2):467--474, 1994.

\bibitem[RTS00]{radhakrishnan2000bounds}
Jaikumar Radhakrishnan and Amnon Ta-Shma.
\newblock Bounds for dispersers, extractors, and depth-two superconcentrators.
\newblock {\em SIAM Journal on Discrete Mathematics}, 13(1):2--24, 2000.

\bibitem[Ser22]{Sergeev22KroneckerCoverings}
Igor~S Sergeev.
\newblock Notes on the complexity of coverings for kronecker powers of symmetric matrices.
\newblock {\em arXiv preprint arXiv:2212.01776}, 2022.

\bibitem[Str69]{strassen1969gaussian}
Volker Strassen.
\newblock Gaussian elimination is not optimal.
\newblock {\em Numerische mathematik}, 13(4):354--356, 1969.

\bibitem[Str86]{strassen1986asymptotic}
Volker Strassen.
\newblock The asymptotic spectrum of tensors and the exponent of matrix multiplication.
\newblock In {\em 27th Annual Symposium on Foundations of Computer Science (sfcs 1986)}, pages 49--54. IEEE, 1986.

\bibitem[Str87]{strassen1987relative}
Volker Strassen.
\newblock Relative bilinear complexity and matrix multiplication.
\newblock {\em J. Reine Angew. Math.}, 375/376:406--443, 1987.

\bibitem[Str88]{strassen88spectrum}
Volker Strassen.
\newblock The asymptotic spectrum of tensors.
\newblock {\em J. Reine Angew. Math.}, 384:102--152, 1988.

\bibitem[Str91]{strassen1991degeneration}
Volker Strassen.
\newblock Degeneration and complexity of bilinear maps: {Some} asymptotic spectra.
\newblock {\em J. Reine Angew. Math.}, 413:127--180, 1991.

\bibitem[Sze75]{szemeredi1975sets}
Endre Szemer{\'e}di.
\newblock On sets of integers containing no {{\(k\)}} elements in arithmetic progression.
\newblock {\em Acta Arith.}, 27:199--245, 1975.

\bibitem[Tao06]{tao2004quantitative}
Terence Tao.
\newblock A quantitative ergodic theory proof of {Szemer{\'e}di}'s theorem.
\newblock {\em Electron. J. Comb.}, 13(1):research paper r99, 49, 2006.

\bibitem[Val77]{valiant1977graph}
Leslie~G Valiant.
\newblock Graph-theoretic arguments in low-level complexity.
\newblock In {\em International Symposium on Mathematical Foundations of Computer Science}, pages 162--176. Springer, 1977.

\bibitem[VW18]{williams2018some}
Virginia Vassilevska~Williams.
\newblock On some fine-grained questions in algorithms and complexity.
\newblock In {\em Proceedings of the international congress of mathematicians: Rio de janeiro 2018}, pages 3447--3487. World Scientific, 2018.

\bibitem[Wil05]{williams2005new}
Ryan Williams.
\newblock A new algorithm for optimal 2-constraint satisfaction and its implications.
\newblock {\em Theoretical Computer Science}, 348(2-3):357--365, 2005.

\bibitem[Wil18a]{Williams18Reductions}
R~Ryan Williams.
\newblock Counting solutions to polynomial systems via reductions.
\newblock In {\em 1st Symposium on Simplicity in Algorithms (SOSA 2018)}. Schloss Dagstuhl-Leibniz-Zentrum fuer Informatik, 2018.

\bibitem[Wil18b]{williams2018limits}
Richard~Ryan Williams.
\newblock Limits on representing boolean functions by linear combinations of simple functions: Thresholds, relus, and low-degree polynomials.
\newblock In {\em 33rd Computational Complexity Conference (CCC 2018)}. Schloss-Dagstuhl-Leibniz Zentrum f{\"u}r Informatik, 2018.

\bibitem[Wil24]{Williams2024equalityrank}
Ryan Williams.
\newblock The orthogonal vectors conjecture and non-uniform circuit lower bounds.
\newblock In {\em 2024 IEEE 65th Annual Symposium on Foundations of Computer Science (FOCS)}, pages 1372--1387. IEEE, 2024.

\bibitem[WXXZ24]{WXXZ24MatrixMultiplication}
Virginia~Vassilevska Williams, Yinzhan Xu, Zixuan Xu, and Renfei Zhou.
\newblock New bounds for matrix multiplication: from alpha to omega.
\newblock In {\em Proceedings of the 2024 Annual ACM-SIAM Symposium on Discrete Algorithms (SODA)}, pages 3792--3835. SIAM, 2024.

\bibitem[WY14]{williams2014finding}
Ryan Williams and Huacheng Yu.
\newblock Finding orthogonal vectors in discrete structures.
\newblock In {\em Proceedings of the twenty-fifth annual ACM-SIAM symposium on Discrete algorithms}, pages 1867--1877. SIAM, 2014.

\bibitem[WZ25]{WZ22spectra}
Avi Wigderson and Jeroen Zuiddam.
\newblock Asymptotic spectra: Theory, applications and extensions.
\newblock {\em Bulletin of the American Mathematical Society}, 2025.
\newblock to appear.

\bibitem[Yat37]{yates1937design}
Frank Yates.
\newblock The design and analysis of factorial experiments.
\newblock 1937.

\bibitem[Zui19]{Zui19Graph}
Jeroen Zuiddam.
\newblock The asymptotic spectrum of graphs and the shannon capacity.
\newblock {\em Combinatorica}, 39(5):1173--1184, 2019.

\end{thebibliography}

\appendix

\section{Analysis Facts}

\begin{lemma} \label{lem:log-convex}
    For any decomposition $I_t$, the $\alpha$-volume function $\rho_t(\alpha)$ is log-convex.
\end{lemma}
\begin{proof}
    We rewrite $\rho_t(\alpha)$ as
    \begin{align*}
        \rho_t(\alpha) &= \sum_{i\in I_t} \nnz(U_i)^\alpha \nnz(V_i)^{1-\alpha} \\
        &= \sum_{i\in I_t} \nnz(V_i) \left(\frac{\nnz(U_i)}{\nnz(V_i)}\right)^\alpha,
    \end{align*}
    thus we have
    \begin{align*}
        \rho_t'(\alpha) &= \sum_{i\in I_t} \nnz(U_i)^\alpha \nnz(V_i)^{1-\alpha} \ln \left(\frac{\nnz(U_i)}{\nnz(V_i)}\right),\\
        \rho_t''(\alpha) &= \sum_{i\in I_t} \nnz(U_i)^\alpha \nnz(V_i)^{1-\alpha} \ln^2 \left(\frac{\nnz(U_i)}{\nnz(V_i)}\right).
    \end{align*}
    Then by Cauchy--Schwarz inequality, we have
    \begin{align*}
        \rho_t'(\alpha)^2 &= \left( \sum_{i\in I_t} \nnz(U_i)^\alpha \nnz(V_i)^{1-\alpha} \ln \left(\frac{\nnz(U_i)}{\nnz(V_i)}\right) \right)^2\\
        &\leq \left( \sum_{i\in I_t} \nnz(U_i)^\alpha \nnz(V_i)^{1-\alpha} \right) \left( \sum_{i\in I_t} \nnz(U_i)^\alpha \nnz(V_i)^{1-\alpha} \ln^2 \left(\frac{\nnz(U_i)}{\nnz(V_i)}\right) \right)\\
        &\leq \rho_t(\alpha) \rho_t''(\alpha).
    \end{align*}
    Thus by 
    \[ \left( \ln \rho_t(\alpha) \right)'' = 
    \frac{\rho_t''(\alpha) \rho_t(\alpha) - \rho_t'(\alpha)^2}{\rho_t(\alpha)^2} \geq 0, \]
    we have as desired that $\rho_t(\alpha)$ is log-convex.
\end{proof}

\begin{lemma} \label{lem:binom-ineq}
    For nonnegative integers $i_1,\dots,i_k\geq 0$, let $N = i_1 + \cdots + i_k$ and let $p = i / N$. We have
    \[ \binom{N}{i_1,\dots,i_k} p_1^{i_1} \cdots p_k^{i_k}
    \geq \frac 1{(N+1)^k}. \]
\end{lemma}
\begin{proof}
    With loss of generality, we assume all $p_1,\dots,p_k > 0$.
    Consider the function
    \[ F(t_1,\dots,t_k) = \binom{N}{t_1,\dots,t_k} p_1^{t_1} \cdots p_k^{t_k}, \]
    we show that $F(t_1,\dots,t_k)$ is maximized at $t = (i_1,\dots,i_k)$.
    Suppose $t = (t_1,\dots,t_k)$ is not equal to $(i_1,\dots,i_k)$, then $t_x > i_x$
    and $t_y < i_y$ for some $x, y$. We consider $t'$ obtained by decreasing $t_x$ by 1
    and increasing $t_y$ by 1, then we have
    \begin{align*}
        \frac{F(t')}{F(t)} &= \frac{t_x! t_y!}{t'_x! t'_y!} \frac{p_x^{t'_x} p_y^{t'_y}}{p_x^{t_x} p_y^{t_y}}\\
        &= \frac{t_x}{t_y} \frac{p_y}{p_x} \\
        &= \frac{t_x}{i_x} \frac{i_y}{t_y} > 1,
    \end{align*}
    thus $F(t)$ is not maximal. Therefore, $F(i_1,\dots,i_k)$ is maximal, and
    \begin{align*}
        1 &= (p_1+\cdots+p_k)^N\\
        &=\sum_{\substack{t_1+\cdots+t_k=N \\ t_1,\dots,t_k\geq 0}} \binom{N}{t_1,\dots,t_k} p_1^{t_1} \cdots p_k^{t_k}\\
        &\leq (N+1)^k \binom{N}{i_1,\dots,i_k} p_1^{i_1} \cdots p_k^{i_k}.
    \end{align*}
    The lemma follows by rearranging the terms.
\end{proof}

\section{An Elementary Proof of the Duality Theorem} \label{sec:tree}

In this section, we follow the notation and setup of Section \ref{sec:spec}.

We first formulate the rebalancing process as a tree construction.

\begin{definition}
    Let $I_1,\dots,I_r$ be decompositions of $M$.
    We say a rooted tree $T$ computes $M^{\otimes n}$ if
    \begin{compactitem}
        \item Every leaf of $T$ has depth $n$,
        \item Every internal node $x$ of $T$ is associated with a decomposition $t(x)$ of $M$, and
        for every child $y$ of $x$, the edge $(x, y)$ is associated with an
        index $i(x, y) \in I_{t(x)}$ in the decomposition.
        Moreover, these children edges are in one-to-one correspondence with the indices of
        decomposition $I_{t(x)}$.
    \end{compactitem}
    Let $\mathcal T_n$ be the set of all such trees.
\end{definition}

\begin{lemma}
    Such a tree $T$ really computes $M^{\otimes n}$. More precisely:
    
    For each leaf $y$ of $T$, which has a path $y_0, \dots, y_n = y$ from the root $y_0$ to $y$, the pair $(A_y, B_y)$ associated with $y$ is the Kronecker product of all pairs associated with the edges $(y_{i-1}, y_i)$,
    i.e.,
    \begin{align*}
        A_y &= U_{i(y_0, y_1)} \otimes \cdots \otimes U_{i(y_{n-1}, y_n)}, \\
        B_y &= V_{i(y_0, y_1)} \otimes \cdots \otimes V_{i(y_{n-1}, y_n)}.
    \end{align*}
    Then we have
    \[ \sum_{y\in \leaves(T)} A_y B_y^\sfT = M^{\otimes n}. \]
\end{lemma}
\begin{proof}
    The proof is by induction on the depth of the tree. For the base case $n = 0$, the statement is trivial.
    Now suppose the statement holds for $n$, we prove it for $n+1$. Let $T$ be a tree computing $M^{\otimes n+1}$.
    Let $y_1$ be a child of the root $y_0$, then the subtree $T_{y_1}$ rooted at $y_1$ computes $M^{\otimes n}$.
    By the induction hypothesis, we have
    \[ \sum_{y\in \leaves(T_{y_1})} (U_{i(y_1,y_2)}\otimes \cdots \otimes U_{i(y_{n},y_{n+1})}) 
    (V_{i(y_1,y_2)}\otimes \cdots \otimes V_{i(y_{n},y_{n+1})})^\sfT = M^{\otimes n}. \]
    Therefore, we have
    \begin{align*}
        \sum_{y\in \leaves(T)} A_y B_y^\sfT &= 
        \sum_{y_1 \in \operatorname{child}(y_0)}  \sum_{y\in \leaves(T_{y_1})} A_{y} B_{y}^\sfT\\
        &= \sum_{y_1 \in \operatorname{child}(y_0)} 
        (U_{i(y_0,y_1)} V_{i(y_0,y_1)}^\sfT) \otimes \sum_{y\in \leaves(T_{y_1})} \\
        & \quad (U_{i(y_1,y_2)}\otimes \cdots \otimes U_{i(y_{n},y_{n+1})})
        (V_{i(y_1,y_2)}\otimes \cdots \otimes V_{i(y_{n},y_{n+1})})^\sfT\\
        &= \sum_{y_1 \in \operatorname{child}(y_0)} U_{i(y_0,y_1)} V_{i(y_0,y_1)}^\sfT \otimes M^{\otimes n}\\
        &= M^{\otimes n+1}. \qedhere
    \end{align*}
\end{proof}

This decomposition has circuit size
\begin{align*}
    &\quad \sum_{y\in \leaves(T)} \nnz(A_y) + \nnz(B_y)\\
    &= \sum_{y\in \leaves(T)} \nnz(U_{i(y_0, y_1)} \otimes \cdots \otimes U_{i(y_{n-1}, y_n)}) + \nnz(V_{i(y_0, y_1)} \otimes \cdots \otimes V_{i(y_{n-1}, y_n)})\\
    &= \sum_{y\in \leaves(T)} \left(\prod_{i=0}^{n-1} \nnz(U_{i(y_i, y_{i+1})}) + \prod_{i=0}^{n-1} \nnz(V_{i(y_i, y_{i+1})})\right).
\end{align*}

For a tree $T \in \mathcal T_n$ and parameter $\lambda > 0$, we define the skew size of $T$ as
\[ S(\lambda, T) = \sum_{y\in \leaves(T)} \nnz(A_y) + \lambda \nnz(B_y), \]
and the optimal skew size of $M^{\otimes n}$ over all trees in $\mathcal T_n$ as
\[ F_n(\lambda) = \min_{T\in \mathcal T_n} S(\lambda, T). \]

By constructing the tree from the bottom to the top, we have the following recursive formula for $F_n(\lambda)$:
\begin{lemma}
    The optimal skew size $F_n(\lambda)$ can be defined recursively by the
    initial condition $F_0(\lambda) = 1 + \lambda$ and recurrence relation
    \[ F_n(\lambda) = \min_{1\leq t\leq r} \left\{ \sum_{i\in I_t} \nnz(U_i)\cdot F_{n-1}\left( \frac{\nnz(V_i)}{\nnz(U_i)} \lambda \right) \right\}. \]
\end{lemma}

We will show that some basic quantities about the decompositions $I_t$ will determine the asymptotic behavior of $F_n(\lambda)$.
For any $\alpha \in [0, 1]$, we define the $\alpha$-volume of $I_t$ as
\[ \rho_t(\alpha) = \sum_{i\in I_t} \nnz(U_i)^\alpha \nnz(V_i)^{1-\alpha}, \]
and
\[ C(\alpha) = \min_t \rho_t(\alpha). \]

\begin{theorem}[Weak Duality] \label{thm:weak-duality}
    For any $\alpha\in [0, 1]$, we have
    \[ F_n(\lambda) \geq C(\alpha)^n \lambda^\alpha. \]
\end{theorem}

\begin{proof}
    We prove the theorem by induction. For the base case $n = 0$, the inequality becomes $1+\lambda \geq \lambda^\alpha$,
    which is true for all $\alpha\in [0, 1]$.

    Now suppose the inequality holds for $n$, we prove it for $n+1$. By the recurrence relation of $F_n(\lambda)$, we have
    \begin{align*}
        F_{n+1}(\lambda) &= \min_{1\leq t\leq r} \left\{ \sum_{i\in I_t} \nnz(U_i)\cdot F_{n}\left( \frac{\nnz(V_i)}{\nnz(U_i)} \lambda \right) \right\}\\
        &\geq \min_{1\leq t\leq r} \left\{ \sum_{i\in I_t} \nnz(U_i)\cdot C(\alpha)^n \left( \frac{\nnz(V_i)}{\nnz(U_i)} \lambda \right)^\alpha \right\}\\
        &= C(\alpha)^n \lambda^\alpha \min_{1\leq t\leq r} \left\{ \sum_{i\in I_t} \nnz(U_i)^\alpha \nnz(V_i)^{1-\alpha} \right\}\\
        &= C(\alpha)^{n+1} \lambda^\alpha.
    \end{align*}
    The last step is due to the definition of $C(\alpha)$.
\end{proof}

From the above theorem, we can see that the $\alpha$-volume of the decompositions $I_t$, for all $\alpha \in [0,1]$, are barriers to the optimal skew size $F_n(\lambda)$.
We will next show that these barriers actually characterize the behavior of $F_n(\lambda)$.
\begin{theorem}[Strong Duality] \label{thm:strong-duality}
    For any family of decomposition containing the two trivial decompositions, we have
    \[ F_n(\lambda) \leq 2^{o(n)} \cdot \sup_{\alpha\in [0, 1]} C(\alpha)^n \lambda^\alpha, \]
    where the $o(n)$ only depends on the family of decompositions.
\end{theorem}

\subsection{Strong Duality for Special Cases}

To prove Theorem \ref{thm:strong-duality}, we first prove it in the case when the family of decompositions
has a special property, then we will generalize the proof by transferring all families of decompositions to the special case.

\begin{lemma} \label{lem:special-case}
    Suppose that, for the family of decompositions $I_t$, there exist division points
    $0 = \alpha_0 < \alpha_1 < \cdots < \alpha_r = 1$ such that
    \begin{compactenum}
        \item The function $C(\alpha)$ takes $\rho_i(\alpha)$ minimized at $\alpha_{i-1} \leq \alpha \leq \alpha_i$, and
        \label{item:division-point}
        \item The intersection points are log-concave, i.e.,
        \[ \frac{\ln C(\alpha_i) - \ln C(\alpha_{i-1})}{\alpha_i - \alpha_{i-1}}
        > \frac{\ln C(\alpha_{i+1}) - \ln C(\alpha_{i})}{\alpha_{i+1} - \alpha_{i}}. \]
        \label{item:log-concave}
    \end{compactenum}
    Then the statement of Theorem \ref{thm:strong-duality} holds, i.e.,
    \[ F_n(\lambda) \leq O\left(\sup_{\alpha\in [0, 1]} C(\alpha)^n \lambda^\alpha\right). \]
\end{lemma}

\begin{proof}
    By Lemma \ref{lem:log-convex}, each $\rho_t(\alpha)$ is log-convex.
    Therefore, for any $\alpha \notin \{\alpha_0,\dots,\alpha_r\}$,
    we have that $\log(C(\alpha)\lambda^\alpha) = \log \rho_t(\alpha) + \alpha \log \lambda$ is a convex
    function in a neighborhood of $\alpha$. Thus $C(\alpha) \lambda^\alpha$ cannot be maximized at such an $\alpha$.
    Therefore, we can restrict the supremum to $\alpha\in \{\alpha_0,\dots,\alpha_r\}$.

    We prove by induction on $n$ that
    \[ F_n(\lambda) \leq O\left(\sum_{\alpha\in \{\alpha_0,\dots, \alpha_r\}} C(\alpha)^n \lambda^\alpha\right). \]

    For the base case, we have
    \[ F_{0}(\lambda) = 1 + \lambda \leq \sum_{\alpha\in \{\alpha_0,\dots, \alpha_r\}} C(\alpha)^{n_0} \lambda^\alpha. \]

    Now suppose the statement holds for $n - 1$, i.e., there exists a constant $D_{n-1}$ such that
    \[ F_{n-1}(\lambda) \leq D_{n-1}\sum_{\alpha\in \{\alpha_0,\dots, \alpha_r\}} C(\alpha)^{n-1} \lambda^\alpha. \]
    By the recurrence relation of $F_n(\lambda)$, we have,
    for each $1\leq t\leq r$,
    \begin{align*}
        F_n(\lambda) &\leq \sum_{i\in I_t} \nnz(U_i) \cdot F_{n-1}\left( \frac{\nnz(V_i)}{\nnz(U_i)} \lambda \right)\\
        &\leq D_{n-1} \sum_{i\in I_t} \nnz(U_i) \cdot \sum_{\alpha\in \{\alpha_0,\dots, \alpha_r\}} C(\alpha)^{n-1} \left(\frac{\nnz(V_i)}{\nnz(U_i)} \lambda\right)^\alpha\\
        &= D_{n-1} \sum_{\alpha\in \{\alpha_0,\dots, \alpha_r\}} \rho_t(\alpha) \cdot C(\alpha)^{n-1} \lambda^\alpha.
            \end{align*}
    We show that there always exists a $t$ which makes $D_n$ grow slowly.
    Define
    \[ \theta_i = \frac{\ln C(\alpha_i) - \ln C(\alpha_{i-1})}{\alpha_i - \alpha_{i-1} }, \]
    so by condition \ref{item:log-concave} of the lemma statement, we have $\theta_1 > \cdots > \theta_r$.
    
    If $r = 1$, the only choice is $t=1$, and by condition \ref{item:division-point} of the lemma, we have $C(\alpha_0) = \rho_1(\alpha_0)$
    and $C(\alpha_1) = \rho_1(\alpha_1)$, thus we have
    \begin{align*}
        F_n(\lambda) &\leq D_{n-1}\left( \sum_{\alpha\in \{ 0, 1 \}} \rho_1(\alpha) C(\alpha)^{n-1} \lambda^\alpha \right)\\
        &= D_{n-1} \left( \sum_{\alpha\in \{ 0, 1 \}} C(\alpha)^n \lambda^\alpha \right).
    \end{align*}
    Thus we have $D_n = D_{n-1}$, and the statement holds.
    
    For $r > 1$, we take
    \[ t = \begin{cases}
        1 & \text{if } \lambda < \exp \left( - \frac{\theta_1 + \theta_2}{2} n \right),\\
        i & \text{if } \exp \left( - \frac{\theta_{i-1} + \theta_i}{2} n \right) \leq \lambda < \exp \left( - \frac{\theta_i + \theta_{i+1}}{2} n \right),\\
        r & \text{if } \lambda \geq \exp \left( - \frac{\theta_{r-1} + \theta_r}{2} n \right).
    \end{cases} \]
    Thus for $i < t-1$, we have
    \begin{align*}
        \frac{C(\alpha_i)^n \lambda^{\alpha_i}}{C(\alpha_{t-1})^n \lambda^{\alpha_{t-1}}}
        &= \exp \left\{ n(\ln C(\alpha_i) - \ln C(\alpha_{t-1})) + (\alpha_i - \alpha_{t-1}) \ln \lambda \right\}\\
        &\leq \exp \left\{ n(\ln C(\alpha_i) - \ln C(\alpha_{t-1})) - (\alpha_i - \alpha_{t-1}) \frac{\theta_{t-1} + \theta_{t}}{2} n \right\}.
    \end{align*}
    Note that
    \begin{align*}
        \ln C(\alpha_{t-1}) - \ln C(\alpha_i) &= \sum_{j=i+1}^{t-1} (\alpha_{j} - \alpha_{j-1}) \theta_{j} \\
        &\geq \left(\sum_{j=i+1}^{t-1} (\alpha_{j} - \alpha_{j-1})\right) \theta_{t-1}\\
        &= (\alpha_{t-1} - \alpha_i) \theta_{t-1},
    \end{align*}
    so we have
    \begin{align*}
        \frac{C(\alpha_i)^n \lambda^{\alpha_i}}{C(\alpha_{t-1})^n \lambda^{\alpha_{t-1}}} &\leq 
        \exp \left\{ n(\alpha_i - \alpha_{t-1}) \theta_{t-1} - (\alpha_i - \alpha_{t-1}) \frac{\theta_{t-1} + \theta_{t}}{2} n \right\}\\
        &= \exp \left\{ - (\alpha_{t-1} - \alpha_i) \frac{\theta_{t-1} - \theta_{t}}{2} n \right\}
        \leq e^{-\Delta n},
    \end{align*}
    for a constant
    \[ \Delta = \left( \min_{1\leq i\leq r} \alpha_i - \alpha_{i-1} \right)
    \left(\min_{1\leq i < r} \frac{\theta_{i} - \theta_{i+1}}2\right) > 0. \]
    Similarly, for $i > t$, we have
    \begin{align*}
        \frac{C(\alpha_i)^n\lambda^{\alpha_i}}{C(\alpha_t)^n \lambda^{\alpha_t}}
        &= \exp \left\{ n(\ln C(\alpha_i) - \ln C(\alpha_t)) + (\alpha_i - \alpha_t) \ln \lambda \right\}\\
        &\leq \exp \left\{ n(\alpha_i - \alpha_t) \theta_{t+1}  - (\alpha_i - \alpha_t) \frac{\theta_{t} + \theta_{t+1}}{2} n \right\}\\
        &= \exp \left\{ - (\alpha_i - \alpha_t) \frac{\theta_{t} - \theta_{t+1}}{2} n \right\}
        \leq e^{-\Delta n}.
    \end{align*}
    By condition \ref{item:division-point} of the lemma statement, we have $C(\alpha_{t-1}) = \rho_t(\alpha_{t-1})$
    and $C(\alpha_t) = \rho_t(\alpha_t)$. Thus we have
    \begin{align*}
        F_n(\lambda) &\leq D_{n-1} \left( \sum_{\alpha \in \{\alpha_{t-1}, \alpha_t\}} C(\alpha)^n \lambda^\alpha
        + \sum_{\alpha \notin \{\alpha_{t-1}, \alpha_t\}}
        \frac{\rho_t(\alpha)}{C(\alpha)} \cdot C(\alpha)^n \lambda^{\alpha} \right)\\
        &\leq D_{n-1} \cdot (1 + O(e^{-\Delta n})) \left(\sum_{\alpha \in \{\alpha_0,\dots,\alpha_t\}} C(\alpha)^n \lambda^\alpha\right).
    \end{align*}
    This shows that $D_n = (1 + O(e^{-\Delta n})) D_{n-1}$, so there is a global constant $D$ such that
    $D_n \leq D$. Thus we have
    \[ F_n(\lambda) \leq D \sum_{\alpha \in \{\alpha_0 , \dots, \alpha_r\} } C(\alpha)^n \lambda^\alpha 
    \leq (r+1)D \cdot \sup_\alpha C(\alpha)^n \lambda^\alpha, \]
    which completes the proof.
\end{proof}

\subsection{Proof of Theorem \ref{thm:strong-duality}}

Consider the function $C(\alpha) = \min_t \rho_t(\alpha)$. This is a piecewise continuous function which, for some positive integer $r$, 
can be divided into $r$ pieces by the division points $\alpha_0,\dots,\alpha_r$.
For each $1\leq t\leq r$, we have that $\rho_t(\alpha)$ is the minimizer of $C(\alpha)$ in the interval $[\alpha_{t-1}, \alpha_t]$,
where $I_t$ is a decomposition in the given family. We remark that different $t$ may correspond to
the same decomposition $I_t$.

Consider the upper convex hull of the plot of $\ln C(\alpha)$, i.e., the points of graph
\[ \Gamma = \{ (\alpha, \ln C(\alpha)) : 0\leq \alpha \leq 1 \} \]
which is the unique maximizer of a linear function $\ell(x, y) = kx + y$ for some $k \in \mathbb R$.
By Lemma \ref{lem:log-convex}, the vertices in the upper convex hull of $\Gamma$ are a subset of division points $\alpha_0,\dots,\alpha_r$. Say they are the division points with indices $0 = i_0 < \cdots < i_s = r$.

We need the following lemma to adjust the family of decompositions to make the whole sequence of division points
form the upper convex hull of $\Gamma$.
\begin{lemma} \label{lem:adjust-decomposition}
    For any family of decompositions $I_t$ satisfying the conditions of Theorem \ref{thm:strong-duality},
    for any $\epsilon > 0$, there exists an associated family of decompositions $\hat I_t$ of $M^{\otimes n}$
    for some $n$ such that:
    \begin{compactitem}
        \item Every decomposition $\hat I_t$ is the Kronecker product of several $I_t$ with multiplicities.
        \item The division points of $\hat C(\alpha)$ form the upper convex hull of $\hat \Gamma$,
        where $\hat \Gamma$ is the plot of $\ln \hat C(\alpha)$.
        \item The value of $\hat C(\alpha)$ doesn't increase much, i.e.,
        \[ \frac{\ln \hat C(\alpha)}{n} < \ln C(\alpha) + \epsilon. \]
            \end{compactitem}
\end{lemma}
\begin{proof}
    With loss of generality, we assume the decompositions in the family that are neither the
    trivial decompositions nor attaining the minimum of $C(\alpha)$ on an interval of positive length,
    are discarded.

    Let $\Delta$ be the number of division points not in the upper convex hull of $\Gamma$.
    We proceed by induction on $\Delta$.

    For the base case, $\Delta = 0$, we have that the division points form the upper convex hull of $\Gamma$.

    Now suppose $\Delta > 0$. There must exist some $i_j - i_{j-1} > 1$, i.e., there is a division point
    between $\alpha_{i_{j-1}}$ and $\alpha_{i_j}$ which is not in the upper convex hull of $\Gamma$.
    \begin{compactitem}
        \item Consider the two decompositions $I_{i_{j-1}+1}$ and $I_{i_j}$, which give the functions
        $\rho_{i_{j-1}+1}$ and $\rho_{i_j}$. If these two decompositions are the same, then the decompositions
        $I_i$ for $i_{j-1} + 1 < i < i_j$ can be removed, which will only change the value of $C(\alpha)$
        on the interval $[\alpha_{i_{j-1}}, \alpha_{i_j}]$, and will not change the convex hull of $\Gamma$.
        This adjustment will reduce $\Delta$ by at least 1, so we can assume the two decompositions are different.
        \item If $i_{j-1}\neq 0$, we can consider a new family of decompositions $\hat I_t$ by taking
        $\hat I_t = I_t^{\otimes n}$ for $I_t\neq I_{i_{j-1}}$, and
        $\hat I_{t} = I_{i_{j-1}-1}^{\otimes a} \otimes I_{i_{j}}^{\otimes (n-a)}$.
        This makes $\frac 1n\ln \hat \rho_t(\alpha) = \ln \rho_t(\alpha)$ for other decompositions,
        but
        \[ \frac 1n \ln \hat \rho_{i_{j-1}}(\alpha) = x \ln \rho_{i_{j-1}-1}(\alpha)
        + (1-x)\ln \rho_{i_{j-1}}(\alpha), \]
        where $x = a/n$.
        Since $\rho_{i_{j-1}-1}(\alpha) \geq \rho_{i_{j-1}}(\alpha)$ for the part
        $C(\alpha) = \rho_{i_{j-1}}(\alpha)$, as $x$ grows from $0$ to $1$, the value of
        $\hat \rho_{i_{j-1}}(\alpha)$ increases. There therefore exists some minimal $x$ such that, at that point,
        two division points of the graph $\hat\Gamma$ coincide, or one of the division points touches the upper convex hull.
        Thus we can choose $a, n$, such that $a/n > x$ but $a/n$ is close enough to $x$, such that
        \[ \frac 1n \ln \hat C(\alpha) < \ln C(\alpha) + \frac{\epsilon}{2}. \]
        That family of decompositions will reduce $\Delta$ by 1, so starting from $\hat I_t$,
        we can apply the inductive hypothesis with $\epsilon/2$ to get the desired family of decompositions
        $\tilde I_t$ of $M^{\otimes {nN}}$ such that
        \[ \frac 1{N} \ln \tilde C(\alpha) < \ln \hat C(\alpha) + \frac{\epsilon}{2}. \]
        Thus we have
        \[ \frac 1{nN} \ln \tilde C(\alpha) < \frac 1n \left(\ln \tilde C(\alpha) + \frac{\epsilon}{2}\right)
        < \ln C(\alpha) + \epsilon. \]
        It's not hard to verify that the other conditions are satisfied.
        \item Otherwise, if $I_{j}\neq r$, we can similarly construct a new family of decompositions
        by adjusting $I_{i_j}$ similar to the previous case.
        \item The only case left is $i_{j-1} = 0$ and $i_j = r$. In this case, we use the condition
        that the two trivial decompositions are in the family. Let $p$ be the number of nonzero rows of $M$,
        and $q$ be the number of nonzero columns of $M$.
        For the row decomposition $I_{\text{row}}$
        \[ M = \sum_i e_i V_i^\sfT, \]
        it has $\rho_{\text{row}}(1) = p$, and for any decomposition $I_t$,
        \begin{align*}
            \rho_t(1) &= \sum_i \nnz(U_i)^1 \nnz(V_i)^0\\
            &= \sum_i \nnz(U_i) \geq p.
        \end{align*}
        We can conclude that $C(1) = p$. Similarly, the column decomposition $I_{\text{col}}$ gives $C(0) = q$.
        When $i_{j-1} = 0$ and $i_j = r$, the convex hull of $\Gamma$ is the line segment connecting $(0, \ln q)$ and $(1, \ln p)$.
        If $q \geq p$, i.e., $C(0) \geq C(1)$, note that the column decomposition has
        \[ \rho_{\text{col}}(\alpha) = \sum_i \nnz(U_i)^\alpha \]
        which is monotone increasing in $\alpha$. Thus we can use it to adjust $I_1$, just like the previous cases.
        Similarly, if $q < p$, we can adjust $I_r$ by the row decomposition. \qedhere
    \end{compactitem}
\end{proof}

Now for any $\epsilon > 0$, by Lemma \ref{lem:adjust-decomposition}, we have a family of decompositions
$\hat I_t$ of $M^{\otimes n}$, and we can apply Lemma \ref{lem:special-case} to get
\begin{align*}
    \hat F_N(\lambda) &\leq O\left( \sup_{\alpha\in [0, 1]} \hat C(\alpha)^N \lambda^\alpha \right)\\
    &\leq O\left( \sup_{\alpha\in [0, 1]} \lambda^\alpha \exp \left\{ N \ln \hat C (\alpha) \right\} \right)\\
    &\leq O\left( \sup_{\alpha\in [0, 1]} \lambda^\alpha \exp \left\{ N (\ln C(\alpha) + n\epsilon) \right\} \right)\\
    &\leq O\left( e^{\epsilon nN} \sup_{\alpha\in [0, 1]} C(\alpha)^N \lambda^\alpha \right).
\end{align*}
Since $\hat F_N$ is the optimal skew size of the tree computing $(M^{\otimes n})^{\otimes N}$,
where each decomposition $\hat I_t$ is a combination of several $I_t$, we have $F_{nN}(\lambda) \leq \hat F_N(\lambda)$,
thus we have
\[ F_{nN}(\lambda) \leq O\left( e^{\epsilon nN} \sup_{\alpha\in [0, 1]} C(\alpha)^{nN} \lambda^\alpha \right). \]
Since we can control $F_N$ at all multiples of $n$, this immediately implies
\[ F_{N}(\lambda) \leq O\left( e^{\epsilon N} \sup_{\alpha\in [0, 1]} C(\alpha)^{N} \lambda^\alpha \right). \]

Since $\epsilon$ is arbitrary, the statement of Theorem \ref{thm:strong-duality} is proved. \hfill \qedsymbol

\subsection{Alternative Proof of Theorem \ref{thm:main}}

Let $\sigma = \sup_\alpha \sigma_\alpha(M)$. For any $\epsilon > 0$, we have $\sigma_\alpha(M) < \sigma + \epsilon$
for all $\alpha\in [0, 1]$. By definition of $\sigma_\alpha(M)$, there is a sufficiently large $n$ and a decomposition $I_\alpha$ of $M^{\otimes n}$ such that
\[ M^{\otimes n} = \sum_{i\in I} U_i V_i^\sfT, \]
such that its $\alpha$-volume is bounded by
\[ \rho(\alpha)^{1/n} = \left(\sum_{i\in I} \nnz(U_i)^\alpha \nnz(V_i)^{1-\alpha}\right)^{1/n} < \sigma + \epsilon. \]
Since $\rho(\alpha)^{1/n}$ is a continuous function of $\alpha$, we have a neighborhood $U_\alpha \subset [0, 1]$ of $\alpha$
satisfying $\rho(\alpha')^{1/n} < \sigma + \epsilon$ for all $\alpha' \in U_\alpha$.
Therefore, the set $\{U_\alpha : \alpha\in [0, 1]\}$ forms an open cover of $[0, 1]$, thus there exists a finite subcover,
forming a finite family of decompositions $I_1,\dots,I_r$ such that
\begin{center}
    $\displaystyle \min_{1\leq t\leq r} \rho_t(\alpha)^{1/n_t} < \sigma + \epsilon$
    \quad for all \quad $\alpha\in [0, 1]$.
\end{center}
where $I_t$ is a decomposition of $M^{\otimes n_t}$. Take $N$ to be a common multiple of $n_1,\dots,n_r$. We have that $I'_t = I_t^{\otimes N/n_t}$ is a decomposition of $M^{\otimes N}$, and its corresponding
volume is $\rho'_t(\alpha)^{1/N} = \rho_t(\alpha)^{1/n_t}$.

By Theorem \ref{thm:strong-duality}, we have a $(\sigma+\epsilon)^{(1+o(1))mN}$ sized circuit
computing $M^{\otimes mN}$, thus we have $\sigma(M) \leq \sigma + \epsilon$. Since $\epsilon$ is arbitrary,
we have $\sigma(M) \leq \sigma$. \hfill \qedsymbol

\section{Covering the Disjointness Matrix}

In this section we determine the value of $\sigma_\OR(R)$ and $\delta_\OR(R)$.

Recall that $R^{\otimes d}_{((1-p)d, pd), ((1-q)d, qd)}$ is the submatrix indexed by
$\binom{[d]}{pd} \times \binom{[d]}{qd}$. In this section we write as $R^{\otimes d}_{pd,qd}$
as a shorthand to denote that submatrix.

\begin{lemma} \label{lem:cover}
    For any $p, q, \mu$ satisfying $p\leq \mu\leq 1-q$, the submatrix $R^{\otimes d}_{pd,qd}$ can be covered by
    \[ O(d) \cdot \frac{\binom{d}{\mu d}} {\binom{(1-p-q)d}{(\mu-p)d}} \]
    many rectangles of shape $\binom{\mu d}{pd}\times \binom{(1-\mu)d}{qd}$. Also, this covering satisfies that each row is covered at most
    \[ O(d) \cdot \frac{\binom{(1-p)d}{(\mu-p)d}}{\binom{(1-p-q)d}{(\mu-p)d}} \] times,
    and each column is covered at most
    \[ O(d) \cdot \frac{\binom{(1-q)d}{(1-\mu-q)d}}{\binom{(1-p-q)d}{(1-\mu-q)d}} \] times. The $O(d)$ term does not depend on $\mu, p, q$.
\end{lemma}
\begin{proof}
    Consider a subset $U\subset [d]$ of size $\mu d$, let $\mathcal{S} \subset \binom{[d]}{pd}$
    be the family of sets $S$ such that $S \subset U$, and $\mathcal{T} \subset \binom{d}{qd}$
    be the family of sets $T$ such that $T \subset [d] \setminus U$. By this definition,
    $\mathcal{S} \times \mathcal{T}$ is a subrectangle of $R^{\otimes d}_{pd,qd}$ since $S\cap T
    = \emptyset$ for any $S\in \mathcal{S}$ and $T\in \mathcal{T}$.

    Note that $|\mathcal{S}| = \binom{\mu d}{pd}$ and $|\mathcal{T}| = \binom{(1-\mu)d}{qd}$,
    which gives the claimed shape.

    Now let $U$ be a subset chosen randomly from $\binom{[d]}{\mu d}$. We then write
    $\mathcal{S}_U, \mathcal{T}$ to denote the families $\mathcal{S},\mathcal{T}_U$ given above.
    By symmetry, for each pair of $S \cap T = \emptyset$, we have $S\in \mathcal{S}_U,
    T\in \mathcal{T}_U$ if and only if $S\subset U \subset [d] \setminus T$, so there are
    exactly $\binom{(1-p-q)d}{(\mu-p)d}$ many $U$s satisfying the condition.
    We thus have
    \[ \Pr[S\in \mathcal{S}_U,
    T\in \mathcal{T}_U] = \frac{\binom{(1-p-q)d}{(\mu-p)d}}{\binom{d}{\mu d}} =: P. \]

    Consider random subsets $U_1,\dots,U_{t}$ of size $\mu d$ generated independently
    where
    $t = \lceil 10d / P \rceil$.
    Each pair of $S, T$ is not covered with probability
    $(1-P)^t \leq e^{-Pt} \leq e^{-10d}$.
    Since there are $\leq 3^d$ such pairs $S, T$, by a union bound, with probability $1-e^{-\Omega(d)}$,
    the rectangles $\mathcal{S}_{U_i} \times \mathcal{T}_{U_i}$ cover the submatrix
    $R^{\otimes d}_{pd, qd}$.

    For a fixed $S$ and a random $U$, we have $\Pr[ S\in \mathcal S_U ] = \binom{(1-p) d}{(\mu - p)d} / \binom{d}{\mu d}$.
    Therefore, the expected number of rectangles covering a certain row $r(S)=\#\{ U_i : S\in \mathcal{S}_{U_i} \}$ is
    \begin{align*}
        \bbE[ r(S) ] &=
        t \cdot \frac{\binom{(1-p) d}{(\mu - p)d}}{\binom{d}{\mu d}}\\
        &= \left(\frac{10d}{P}+O(1)\right)\frac{\binom{(1-p) d}{(\mu - p)d}}{\binom{d}{\mu d}}\\
        &= (10d + O(1)) \cdot \frac{\binom{(1-p)d}{(\mu-p)d}}{\binom{(1-p-q)d}{(\mu-p)d}} \geq 10d.
    \end{align*}
    A Chernoff bound with parameter $\delta = 1$ therefore gives that
    \[ \Pr[r(S) \geq 2 \bbE[r(S)]] \leq e^{-10d/3}, \]
    and since there are only $\leq 2^d$ many such rows, we have, with probability $1-e^{-\Omega(d)}$,
    that all rows are covered by at most $2\bbE[r(S)]$ rectangles.
    The argument is similar for columns. The column is covered by
    $c(T) = \#\{U_i : T\in\mathcal{T}_{U_i}\}$, which satisfies
    \[ \bbE[c(T)] = (10d+O(1)) \cdot \frac{\binom{(1-q)d}{(1-\mu-q)d}}{\binom{(1-p-q)d}{(1-\mu-q)d}}. \]
    Since the probability of any of the three kinds of bad events happening is $e^{-\Omega(d)} \ll 1$,
    it follows that such a covering always exists when $d$ is sufficiently large.
\end{proof}

\begin{theorem}
    $\sigma_\OR(R) = \sqrt 5$.
\end{theorem}
For the lower bound ($\sigma_\OR(R)\geq \sqrt 5$), see \cite[Lemma 4.2]{jukna2013complexity}. 
The upper bound is implicit in \cite[Remark of Theorem 9]{CILS17cover}; we sketch the proof here, omitting the details of a straightforward but messy optimization in the last step.
\begin{proof}
    For any $p, q$, we can take $p\leq \mu\leq 1-q$ and by Lemma \ref{lem:cover},
    the submatrix $R^{\otimes d}_{pd, qd}$ has a depth-2 circuit of size
    \begin{align*}
        &\quad O(d) \cdot \frac{\binom{d}{\mu d}} {\binom{(1-p-q)d}{(\mu-p)d}} \cdot 
        \max \left\{ \binom{\mu d}{pd}, \binom{(1-\mu) d}{qd} \right\}\\
        &= d^{O(1)} \frac{2^{\H(\mu) d}}{2^{(1-p-q) \H\left(\frac{\mu-p}{1-p-q}\right)d}}
        \max \left\{ 2^{\mu \H(\frac {p}{\mu}) d}, 2^{ (1-\mu) \H\left(\frac{q}{1-\mu}\right)d} \right\}\\
        &= d^{O(1)} \exp \left[ d\ln 2 \cdot \left[ \H(\mu) - (1-p-q) \H\left(\frac{\mu-p}{1-p-q}\right)
        + \max \left\{ \mu \H\left(\frac {p}{\mu}\right), (1-\mu) \H\left(\frac{q}{1-\mu}\right) \right\} \right] \right].
    \end{align*}
    Let
    \[ f_{p, q}(\mu) = \H(\mu) - (1-p-q) \H\left(\frac{\mu-p}{1-p-q}\right)
    + \max \left\{ \mu \H\left(\frac {p}{\mu}\right), (1-\mu) \H\left(\frac{q}{1-\mu}\right) \right\} , \]
    so by the continuity of $f$, we have
    \[ \min_{\substack{\mu \in \{0/d,1/d,\dots,d/d\} \\ p\leq \mu \leq 1-q}} 2^{f_{p,q}(\mu) d} \leq 2^{o(d)} \cdot
    \inf_{p\leq \mu \leq 1-q} 2^{f_{p,q}(\mu)d}. \]
    Then the total size of the circuit is at most
    \begin{align*}
        \sum_{\substack{p,q\geq 0\\ p+q\leq 1\\ p,q\in\{0/d,\dots,d/d\}}} \left(2^{o(d)} \cdot
        \inf_{p\leq \mu \leq 1-q} 2^{f_{p,q}(\mu)d}\right)
        &\leq d^{O(1)} \cdot 2^{o(d)} \cdot \sup_{\substack{p,q\geq 0 \\ p+q\leq 1}}
        \inf_{p\leq \mu \leq 1-q} 2^{f_{p,q}(\mu)}\\
        &= \left(\sup_{\substack{p,q\geq 0 \\ p+q\leq 1}}
        \inf_{p\leq \mu \leq 1-q} 2^{f_{p,q}(\mu)}\right)^{d+o(d)}.
    \end{align*}
    Therefore, we have relaxed the problem to bounding the function
    \[ \sigma_\OR(R) \leq \sup_{\substack{p,q\geq 0 \\ p+q\leq 1}}
    \inf_{p\leq \mu \leq 1-q} 2^{f_{p,q}(\mu)} \]
    over the real numbers. It can be verified that the maximum is taken at $p=q=0.4$ and $\mu=0.5$,
    and $2^{f_{0.4, 0.4}(0.5)} = \sqrt 5$.
\end{proof}

\begin{theorem} \label{thm:disjointness_delta_or}
    $\delta_\OR(R) = \frac{2}{\sqrt 3}$.
\end{theorem}

The analysis of the upper bound already appeared in \cite[Theorem 6.1]{NW21OV}.

\begin{proof}
    We first prove the upper bound.
    By Lemma \ref{lem:cover}, for any $p\leq \mu \leq 1-q$, we know that the submatrix $R^{\otimes d}_{pd,pd}$ can be covered by
    \[ O(d) \cdot \frac{\binom{(1-p)d}{(\mu-p)d}}{\binom{(1-p-q)d}{(\mu-p)d}}
    \leq d^{O(1)} \frac{2^{(1-p)\H(\frac{\mu-p}{1-p})d}}{2^{(1-p-q)\H(\frac{\mu-p}{1-p-q})d}} \]
    many rectangles, and each row is covered at most
    \[ O(d) \cdot \frac{\binom{(1-q)d}{(1-\mu-q)d}}{\binom{(1-p-q)d}{(1-\mu-q)d}} = 
    d^{O(1)} \frac{2^{(1-q)\H(\frac{1-\mu-q}{1-q})d}}{2^{(1-p-q)\H(\frac{\mu-p}{1-p-q})d}}. \]

    Similar to the upper bound of $\sigma(R)$, we can upper bound $\delta(R)$ by the following
    continuous optimization problem
    \[ \delta_\OR(R) \leq \sup_{\substack{p,q\geq 0 \\ p+q\leq 1}} 
    \inf_{p \leq \mu \leq 1-q} 2^{f(p, q, \mu)}, \]
    where
    \begin{align*}
        f(p, q, \mu) &= \max \left\{ (1-p)\H\left(\frac{\mu-p}{1-p}\right),
        (1-q)\H\left(\frac{1-\mu-q}{1-q}\right) \right\} - (1-p-q)\H\left(\frac{\mu-p}{1-p-q}\right).
    \end{align*}
    It can be verified that the maximum is taken at $p=q=1/3$ and $\mu=0.5$, and
    $2^{f(1/3, 1/3, 0.5)} = 2 / \sqrt 3$.

    For the lower bound, we focus on the submatrix $R^{\otimes d}_{pd, pd}$. A degree lower bound
    for this submatrix directly gives a lower bound for the whole matrix. Note that for each
    subrectangle $\mathcal S \times \mathcal T$ of $R^{\otimes d}_{pd, pd}$, considering the sets
    $U = \bigcup_{S\in \mathcal S} S$ and $V = \bigcup_{T\in \mathcal T} T$, we have that $U$ and $V$
    are disjoint, since otherwise there is a pair $S, T$ such that $S\cap T \neq \emptyset$.
    Therefore, any subrectangle of $R^{\otimes d}_{pd, pd}$ can be covered by a rectangle of
    form $\mathcal{S}_U \times \mathcal{T}_U$. Each rectangle $\mathcal{S} \times \mathcal{T}$
    contributes $|\mathcal{S}| + |\mathcal{T}|$ to the sum of the degrees of the circuit.
    We then consider the following relaxation
    \begin{align*}
        |\mathcal{S}| + |\mathcal{T}| &= |\mathcal{S}\times \mathcal{T}| \left( \frac{1}{|\mathcal{S}|}
        + \frac 1{|\mathcal T|} \right)\\
        &\geq |\mathcal{S}\times \mathcal{T}| \left( \frac{1}{|\mathcal{S}_U|}
        + \frac 1{|\mathcal T_U|} \right)\\
        &\geq |\mathcal{S}\times \mathcal{T}| \cdot \max \left\{ \binom{\mu d}{pd}^{-1}, \binom{(1-\mu)d}{pd}^{-1} \right\},
    \end{align*}
    where $\mu = |U|/d$. By the monotonicity of the binomial coefficient, we see that this maximum is minimized
    at $\mu = 1/2$, so
    \[ |\mathcal{S}| + |\mathcal{T}| \geq |\mathcal{S}\times \mathcal{T}| \cdot \binom{d/2}{pd}^{-1}. \]
    Let $\mathcal{S}_i\times \mathcal{T}_i$ be the families of rectangles covering the submatrix. By double counting, we have
    \begin{align*}
        \sum_{S\in \binom{[d]}{pd}} r(S) + c(S) &= 
        \sum_{i\in I} |\mathcal{S}_i| + |\mathcal{T}_i|\\
        &\geq \sum_{i\in I} |\mathcal{S}_i| \cdot |\mathcal{T}_i| \cdot \binom{d/2}{pd}^{-1}\\
        &\geq \#\left\{ S, T \in \binom{[d]}{pd} : S\cap T = \emptyset \right\}\cdot \binom{d/2}{pd}^{-1}\\
        &= \binom{d}{pd,pd,(1-2p)d} \cdot \binom{d/2}{pd}^{-1}.
    \end{align*}
    Therefore, the minimal degree has
    \begin{align*}
        &\geq \binom{d}{pd,pd,(1-2p)d} \cdot \binom{d/2}{pd}^{-1} \cdot \binom{d}{pd}^{-1}\\
        &\geq \binom{(1-p)d}{pd} \cdot \binom{d/2}{pd}^{-1}\\
        &\geq d^{-O(1)} \cdot 2^{((1-p)\H(p/(1-p)) - (1/2) \H(2p))d},
    \end{align*}
    which means $\delta_\OR(R) \geq 2^{(1-p)\H(p/(1-p)) - (1/2) \H(2p)}$. Taking $p = 1/3$, we get
    $\delta_\OR(R) \geq 2 / \sqrt{3}$.
\end{proof}

\section{Connections with Other Lower Bounds}

\begin{theorem} \label{thm:ovc-and-hadamard}
    The Orthogonal Vectors Conjecture implies that the Walsh--Hadamard matrix $H_d$ of
    side-length $N$ does not have a depth-2 linear circuit of size $N^{1+o(1)}$.
\end{theorem}
\begin{proof}
    Suppose $H_d$ has a depth-2 linear circuit of size $N^{1+o(1)}$, i.e., that $\sigma(H) = 2$. Since $H$ is the DFT matrix of $\bbZ_2$,
    by Corollary \ref{cor:dft_symmetry}, we have $\delta(H) = 1$. By Theorem \ref{thm:ov-algorithms},
    $\#\OV_{\bbZ_2}$ can be solved in $\tilde O(n 2^{\epsilon d})$ time for any $\epsilon > 0$.
    In other words, for any $c > 0$ and $\epsilon > 0$, the $\#\OV_{\bbZ_2}$ problem over dimension $d=c\log n$ can be solved
    in time $n^{1+\epsilon}$. Using a reduction from $\#\OV$ to $\#\OV_{\bbZ_2}$ of Williams~\cite[Corollary 6]{Williams18Reductions}, it follows that
    $\#\OV$ in dimension $O(\log n)$ can be solved in time $n^{1+o(1)}$, which contradicts the Orthogonal Vectors Conjecture.
\end{proof}

Again, by Theorem \ref{thm:ov-algorithms}, we have the following conditional lower bound
for the disjointness matrix $R_d$.
\begin{theorem} \label{thm:ovc-and-disjointness}
    The Orthogonal Vectors Conjecture implies that the disjointness matrix $R_d$ of side-length $N$ does not have a depth-2 linear circuit
    with degree $N^{o(1)}$ on the input and output layers.
\end{theorem}

\begin{theorem} \label{thm:threshold-circuit}
    If the disjointness matrix $R_d$ of side-length $N$ does not have a depth-2 linear circuit of size $N^{1 + o(1)}$, then Boolean Inner Product on $n$-bit vectors does not have     $2^{o(n)}$-size $\mathsf{SUM} \circ \mathsf{ETHR}$ circuit.
\end{theorem}
\begin{proof}
    The equality rank of disjointness matrix $R_n$ can be upper bounded by the size of $\mathsf{SUM} \circ \mathsf{ETHR}$ circuit computing
    Boolean Inner Product on $n$-bit vectors \cite[Theorem 2]{Williams2024equalityrank}. Therefore,
    assuming that Boolean Inner Product on $n$-bit vectors has $2^{o(n)}$-size $\mathsf{SUM} \circ \mathsf{ETHR}$ circuit,
    the equality rank of $R_n$ is $2^{o(n)}$, i.e., $\theta(R) = 1$. By Proposition \ref{prop:basic-asymptotic-inequalities},
    we have $\sigma(R) = 2$, i.e., $R_n=R^{\otimes n}$ has a depth-2 linear circuit of size $N^{1+o(1)}$.
\end{proof}

Similarly, since $(H_n + \one \one^\sfT)/2$ represents the matrix of the
Boolean Inner Product Mod 2 problem, we have the following.
\begin{theorem} \label{thm:clb-hadamard}
    If the disjointness matrix $H_n$ of side-length $N$ does not have a depth-2 linear circuit of size $N^{1 + o(1)}$,
    then Boolean Inner Product Mod 2 on $n$-bit vectors does not have     $2^{o(n)}$-size $\mathsf{SUM} \circ \mathsf{ETHR}$ circuit.
\end{theorem}

\end{document}